\documentclass[reqno]{amsart}
\usepackage {CJKutf8}
\usepackage{amssymb}
\usepackage[T1]{fontenc}
\usepackage[utf8]{inputenc}
\usepackage[abbrev]{amsrefs} % abbrev option to reduce given name to initials
\usepackage{hyperref}

\usepackage{mathtools}
\usepackage[scr]{rsfso}
\usepackage[margin=3cm]{geometry}
\usepackage{lmodern}
\usepackage{tikz}
\usetikzlibrary{decorations.pathmorphing, shapes.geometric,arrows, fit}
\usetikzlibrary{positioning}
\usepackage{amssymb}
\usepackage{enumitem}

\usepackage{todonotes}

\usepackage{float} % to fix tables

\theoremstyle{plain}%
\newtheorem{theorem}{Theorem}%  meant for continuous numbers
\newtheorem{proposition}[theorem]{Proposition}% 
\newtheorem{lemma}[theorem]{Lemma}%  meant for continuous numbers
\newtheorem{corollary}[theorem]{Corollary}%  meant for continuous numbers
\numberwithin{equation}{section}
\numberwithin{theorem}{section}

\theoremstyle{remark}%
\newtheorem{remark}{Remark}%

\theoremstyle{definition}%
\newcommand{\inn}[2]{\left\langle#1,\,#2\right\rangle}

\DeclareMathOperator{\supp}{supp}

\DeclareMathOperator{\dist}{dist}

\newcommand{\di}{\partial}

\DeclareMathOperator{\diam}{diam}

\newcommand{\br}[1]{ \big\langle#1 \big\rangle}
\newcommand{\si}{\sigma}
\newcommand{\eps}{\epsilon}

\newcommand{\g}{\gamma}
\newcommand{\al}{\alpha}

\newcommand{\Rb}{\mathbb{R}}
\newcommand{\N}{\mathbb{N}}

\newcommand{\om}{\omega}

\renewcommand{\l}{\lambda} % renewed from slash l
\newcommand{\abs}[1]{\ensuremath{\left\lvert#1\right\rvert}}
\newcommand{\norm}[1]{\ensuremath{\left\lVert#1\right\rVert}}
\newcommand{\Nf}[4][]{N_{#2,#3#4}^{#1}}
\newcommand{\sbr}[1]{\left[#1\right]}
\newcommand{\Set}[1]{\left\{#1\right\}}

\newcommand{\md}[6]{\ensuremath{
		\ifinner
		\tfrac{\partial{^{#2}}#1}{\partial{#3^{#4}}\partial{#5^{#6}}}
		\else
		\tfrac{\partial{^{#2}}#1}{\partial{#3^{#4}}\partial{#5^{#6}}}
		\fi
}}
\newcommand{\del}[1]{\left(#1\right)}
\newcommand{\thmref}[1]{Theorem~\ref{#1}}

\newcommand{\secref}[1]{Section~\ref{#1}}
\newcommand{\lemref}[1]{Lemma~\ref{#1}}
\newcommand{\propref}[1]{Proposition~\ref{#1}}
\newcommand{\remref}[1]{Remark~\ref{#1}}
\newcommand{\figref}[1]{Figure~\ref{#1}}
\newcommand{\corref}[1]{Corollary~\ref{#1}}

\numberwithin{equation}{section}

\usepackage[normalem]{ulem}
\usepackage{soul}
\setstcolor{blue}

\definecolor{green}{rgb}{0.0, 0.5, 0.5}

\definecolor{lgray}{gray}{0.9}
\definecolor{llgray}{gray}{0.95}
\definecolor{lllgray}{gray}{0.975}

\newcommand{\R}{\mathbb{R}}

\newcommand{\rd}{\mathrm{d}}
\newcommand{\me}{\mathrm{e}}
\newcommand{\Mr}[1]{\mathrm{#1}}

\newcommand{\scrP}{\mathscr{P}}

\newcommand{\cB}{\mathcal{B}}
\newcommand{\cD}{\mathcal{D}}

\newcommand{\cE}{\mathcal{E}}
\newcommand{\cF}{\mathcal{F}}
\newcommand{\cG}{\mathcal{G}}

\newcommand{\cR}{\mathcal{R}}

\newcommand{\nc}{\newcommand}

\nc{\h}{\delta}
\nc{\G}{\Gamma}
\nc{\et}{\eta} 
\nc{\gam}{\gamma}
\nc{\ka}{\kappa}
\nc{\lam}{\lambda}
\nc{\Lam}{\Lambda}
\nc{\ta}{\tau}
\nc{\w}{\omega}
\nc{\io}{\iota}
\nc{\s}{\sigma}
\nc{\vphi}{\varphi}

\nc{\e}{\epsilon}

\renewcommand{\k}{\kappa}

\nc{\ran}{\rangle}
\nc{\lan}{\langle}

\nc{\bfone}{{\bf 1}}
\nc{\1}{{\bf 1}}

\nc{\dd}{\mathrm{d}}
\newcommand{\DETAILS}[1]{}

\usepackage{array}
\usepackage{multirow}

\newcommand{\cp}{\mathrm{c}}

\newcommand{\Rem}{\mathrm{Rem}}

\newcommand{\phidel}[1]{\inn{\phi}{{#1}\phi}}

\newcommand{\ordel}[1]{\inn{\varphi_r}{{#1}\varphi_r}}

\newcommand{\phirdel}[1]{\inn{\phi_r}{{#1}\phi_r}}

\newcommand{\undel}[1]{\inn{\varphi}{{#1}\psi}}
\newcommand{\urdel}[1]{\inn{\varphi_r}{{#1}\psi_r}}

\newcommand{\wndel}[1]{\inn{\psi}{{#1}\psi}}

\newcommand{\wtdel}[1]{\inn{\psi_t}{{#1}\psi_t}}

\newcommand{\wzdel}[1]{\inn{\psi_0}{{#1}\psi_0}}

\DeclareMathOperator{\mI}{\mathrm{I}}
\DeclareMathOperator{\mII}{\mathrm{II}}
\DeclareMathOperator{\mIII}{\mathrm{III}}

\newcommand{\Norm}[1]{{\left\vert\kern-0.25ex\left\vert\kern-0.25ex\left\vert #1 
		\right\vert\kern-0.25ex\right\vert\kern-0.25ex\right\vert}}
\begin{document}

	\title[]{On the quantum dynamics of long-ranged Bose-Hubbard Hamiltonians}
	
		\author{Marius Lemm}
	\address{Department of Mathematics, University of T\"ubingen, 72076 T\"ubingen, Germany }
	\email{marius.lemm@uni-tuebingen.de}
	
	\author{Carla Rubiliani}
	\address{Department of Mathematics, University of T\"ubingen, 72076 T\"ubingen, Germany }
	\email{carla.rubiliani@uni-tuebingen.de}
	
	\author{Jingxuan Zhang
		%		(\begin{CJK*}{UTF8}{gbsn}张景宣
			%		\end{CJK*})
	}
	\address{Yau Mathematical Sciences Center, Tsinghua University, Beijing 100084, China }
	\email{jingxuan@tsinghua.edu.cn}

	\date{May 3, 2025}
	%\keywords{Maximal propagation speed; Lieb-Robinson bounds; quantum many-body  systems; quantum information; quantum light cones}
	\begin{abstract}
		We study the quantum dynamics generated by Bose-Hubbard Hamiltonians with long-ranged (power law) terms. 
		We prove two ballistic propagation bounds for suitable initial states: (i) A bound on all moments of the local particle number for all power law exponents $\alpha>d+1$ in $d$ dimensions, the sharp condition.  (ii) The first thermodynamically stable Lieb-Robinson bound (LRB) for these Hamiltonians. 
		To handle the long-ranged and unbounded terms, we further develop the  multiscale ASTLO (adiabatic space time localization observables) method introduced in our recent work \cite{lemm2023microscopic}.
	\end{abstract}

	\maketitle
	
	%\setcounter{tocdepth}{1}
	%\tableofcontents
	
	\section{Introduction}\label{secIntro}
	
	In this paper, we study the  quantum dynamics of macroscopically large quantum systems at positive density with strong, long-range interactions.	
	These models are relevant, e.g., for describing the physics of ultracold gases in optical lattices, which  are the main platforms on which experimental physicists can perform  quantum simulation \cite{bloch2012quantum,greiner2008optical,kaufman2021quantum
	}.

	We consider a system of $N$ bosonic quantum particles whose positions are restricted to a discrete, finite subset $\Lambda\subset\mathbb{R}^d,\,d\ge1$.  	
	The state of the $N$-boson system is mathematically given by a permutation-symmetric element of the Hilbert space $\bigotimes^N_{j=1} \ell^2(\Lambda)$. It can be identified with a function
	\begin{equation}
		\begin{aligned}
			\label{eq:introHilber}
			&\psi(x_1,\ldots,x_N),\qquad \textnormal{with } x_1,\ldots ,x_N\in \Lambda
	\end{aligned}
	\end{equation}
     satisfying $\psi(x_{\pi(1)},\ldots,x_{\pi(N)})= \psi(x_1,\ldots,x_N)$ for all permutations $\pi\in S_N$.
	
	The system state at time $t$ is obtained by solving the time-dependent Schrödinger equation
	\begin{equation}\label{eq:introSE}
		i\partial_t \psi_t(x_1,\ldots,x_N) = H\psi_t(x_1,\ldots,x_N).
	\end{equation}
	where $H$, the Hamiltonian, is a linear self-adjoint operator of the form
	\begin{equation}\label{eq:introH}
		H=\sum_{i=1}^N J^{i}+\sum_{1\leq i<j\leq N} V(x_i,x_j).
	\end{equation}
	Here $J^{i}$ acts as a fixed $|\Lambda|\times |\Lambda|$ Hermitian matrix $J$ on the $i$-th particle; it represents the hopping of the $i$-th particle between lattice sites\footnote{In physics, $\Lambda$ is usually called a ``lattice'' and its elements are called ``sites''. We adopt this terminology here without associating any mathematical meaning to it.}  $x,y\in\Lambda$.
	Moreover, $V:\Lambda\times \Lambda\mathbb\,\rightarrow\, \R$ represents the interaction between particles. The precise definitions can be found in Section \ref{secSetup}.

	We are interested in systems which combine a few analytical challenges:
	\begin{enumerate}[label=(\alph*)]
		\item 
		We are interested in \textit{macroscopically large systems at positive  density}, i.e., we aim for bounds that are stable in the thermodynamic limit $N,\,\abs{\Lam}\to\infty$ with fixed density $N/\abs{\Lam}=\rho>0$. The main large parameter is the number of particles $N$.  While the PDE \eqref{eq:introSE} is linear, the presence of $N$ affects the number of tensor powers that is taken in \eqref{eq:introHilber} and so the problem is effectively one in extremely large dimension. This is widely known as the ``curse of dimensionality'' of the quantum many-body problem. Meaningful results thus need to be independent of $N$, i.e., of dimension.\footnote{For few-body dynamics, i.e., when $N$ is held fixed, standard PDE techniques allow to derive ballistic bounds, but the methodology therein does not extend to systems at positive spatial density that we consider here. This means that even deriving a ballistic propagation bound that is uniform in $N$ lies outside the scope of usual PDE techniques (e.g., the well-known Radin-Simon ballistic upper bound \cite{radin1978invariant}). Instead, one needs to exploit the non-trivial local tensor structure that is inherent in quantum many-body physics through the choice of Hilbert space \eqref{eq:introHilber} and Hamiltonian \eqref{eq:introH}.}

		\item The Bose-Hubbard Hamiltonians \eqref{eq:introH} are \textit{strongly interacting}. Mathematically, this means that the interaction term $\sum_{i,j} V(x_i,x_j)$ does not have any kind of small parameter in front of it that would weaken the effect of interactions (e.g., in the widely studied mean-field regime one places a prefactor $\frac{1}{N}$ in front of $\sum_{i,j} V(x_i,x_j)$ \cite{hepp1974classical,ginibre1979classical,rodnianski2009quantum}). Also, we work at fixed density and so we are also not in a ``dilute'' regime in which interactions are strong but rare, such as the well-studied Gross-Pitaevskii regime \cite{erdHos2009rigorous,erdos2010derivation,brennecke2019gross}. Since our systems are not associated with any small parameter, the kind of effective evolution equations (and other expansion techniques around weakly interacting systems) that are commonplace in the literature are not available to us. In fact, they are not meaningful here because there is \textit{no small parameter} to expand in.
		
		\item We will also allow for both    $J^i$ and $V$ to be \textit{long-ranged} (or ``heavy-tailed''), i.e., we allow for power-law decay. More precisely, recalling that $d$ denotes the spatial dimension, we make the following power-law decay assumptions
		\begin{equation}\label{eq:long-range_assumption}
		|J^{i}_{x_i y_i}|\leq C |x_i-y_i|^{-\alpha},\quad 
|V(x_i,x_j)|\leq C |x_i-x_j|^{-\alpha},\qquad 
\textnormal{for all $1\le i \le N$},
		\end{equation}
		Here $\al>0$ is a fixed power-law exponent that we aim to choose as small as possible, so that the results hold with the heaviest possible tails.
		
		\item For bosons, the local interactions are strong in another, non-parametric sense: they are effectively \textit{unbounded}\footnote{More precisely, the size of local interactions can grow in the total particle number $N$, which amounts to unboundedness in the limit  $N\to\infty$.} even on a lattice, because  arbitrarily many of the $N$ bosons can accumulate at a single lattice site.

	\end{enumerate}

	We are interested in the unitary Schr\"odinger dynamics $e^{-\mathrm{i}tH}$ generated by the self-adjoint operator \eqref{eq:introSE} via Stone's theorem. 
	Our goal is to derive \textit{robust propagation bounds} that hold for the whole class of Hamiltonians \eqref{eq:introH} and that are thermodynamically stable, i.e., that are uniform as $N\to\infty$. 
	From a PDE perspective, our goal is to  derive \textit{a priori} bounds which capture the fundamental limitations of these quantum many-body systems to rapidly transport either particles and/or correlations. Such propagation bounds on correlations are commonly known as ``Lieb-Robinson bounds'' after \cite{lieb1972finite}.
	%These propagation bounds can be both viewed as fundamental physical statements about the fastest possible quantum transport in these systems or as a priori bounds in many applications; both perspectives are equally valid.
	The wide interest in quantum propagation bounds stems from the fact that they are decisive tools both for structural results in the context of operator algebras  (e.g., they provide the existence of the thermodynamic limit  \cite{brattelioperator,nachtergaele2006propagation}) and for analytical proofs in mathematical quantum information theory \cite{hastings2004lieb,hastings2006spectral,hastings2007area}. On top of this, they play a role in the rigorous mathematical definition of a ``quantum  phase'' \cite{hastings2005quasiadiabatic,hastings2007quantum,bachmann2012automorphic}.
	Due to these various uses and others,  propagation bounds have had profound impact on the area of mathematical physics  that investigates quantum many-body systems through analytical methods. Accordingly, deriving these essential bounds for systems that lie beyond the existing analytical methods has developed into a dynamic research area within mathematical physics and our paper falls within this area.
	
	The original Lieb-Robinson bound \cite{lieb1972finite} and its extensions in \cite{hastings2004decay,nachtergaele2006lieb} successfully address points (a) and (b) above, but these methods did not satisfactorily address points (c) and (d), i.e., long-ranged and/or unbounded interactions.  In the past 5 years, methodological breakthroughs have allowed to comprehensively treat models that are long-ranged  and bounded, i.e., models that satisfy (a)-(c) but not (d) \cite{chen2019finite,kuwahara2020strictly,tran2020hierarchy,tran2021lieb,tran2021optimal}. For unbounded interactions (which arise, e.g., in bosonic models such as \eqref{eq:introH}), the picture immediately becomes quite complicated. Unbounded interactions are not only challenging from an analytical-methodological perspective, their physical behavior is also remarkably rich.: there exist simple examples of bosonic Hamiltonians that appear to be  well-behaved in various standard ways (translation-invariant, one-dimensional, with nearest-neighbor interactions), but that provably display exponentially accelerating particle and correlation transport \cite{eisert2009supersonic}.   Most research on propagation bounds for unbounded interactions has focused on the paradigmatic Bose-Hubbard type Hamiltonians \cite{schuch2011information,wang2020tightening,yin2022finite,faupin2022lieb,faupin2022maximal,sigal2022propagation,lemm2023information,lemm2023microscopic,kuwahara2024effective,kuwahara2024enhanced,van2024optimal,lemm2024local,faupin2025macroscopic}; see also \cite{nachtergaele2009lieb,woods2016dynamical} which considers different but related lattice oscillator systems. Also continuum systems have been considered which are unbounded for different reasons \cite{gebert2020lieb,bachmann2024lieb,hinrichs2024lieb}.  
	
	The long-ranged Bose-Hubbard Hamiltonians \eqref{eq:introH} we consider comprise all of the aforementioned methodological obstacles  (a)-(d) simultaneously. 
	While our understanding of how to treat these obstacles individually has improved a lot in recent years, there is no way to separate them and so they compound in difficulty. While some aspects of  long-ranged Bose-Hubbard Hamiltonians  \eqref{eq:introH} have been considered in a few recent works \cite{faupin2022lieb,faupin2022maximal,sigal2022propagation,lemm2023information,lemm2023microscopic,van2024optimal}, only the very rough notion of ``macroscopic transport'' has been satisfactorily addressed.  The following two central dynamical problems have hitherto remained unresolved.  Before we state these, we mention that it has been realized that there is a marked difference in the behavior of \textit{particle propagation bounds} and \textit{correlation propagation bounds} (i.e., Lieb-Robinson bounds) and these two types of problems have to be treated somewhat differently; the same will be true for us here.
	
	\begin{itemize}
		\item \textit{Problem 1: Uniform particle transport bound.} Controlling the growth of high moments of the local boson number $n_x^p$ in time is essential to control the local accumulation of bosons  under many-body dynamics.
		By developing a  multiscale version \cite{lemm2023microscopic} of the ASTLO (adiabatic space time localization observable) scheme introduced in \cite{faupin2022lieb}, we  have recently been able to control the $p$-th moments if $\al>\max\{3dp/2+1,2d+1\}$ grows with $p$. 
		However, considering a fixed $\alpha$ and then controlling arbitrarily high moments has been open.
		\item \textit{Problem 2: Thermodynamically stable Lieb-Robinson bound (LRB).} On the one hand, for finite-range Bose-Hubbard Hamiltonians, thermodynamically stable LRBs (meaning LRBs that are independent of the total particle number) have been proved for bounded-density initial states \cite{kuwahara2024effective,kuwahara2024enhanced}, but with velocities that are growing in time.  On the other hand, for initial states satisfying stronger assumptions, namely states with particle-free regions,  bounded-velocity (``ballistic'') LRBs have been proved in \cite{faupin2022lieb,sigal2022propagation,lemm2023information}. However, the error terms appearing in \cite{faupin2022lieb,sigal2022propagation,lemm2023information} still depend on the total particle number $N$ and are thus divergent in the thermodynamic limit $N\to\infty$. As thermodynamic stability is a fundamental requirement, it is of paramount importance to close this gap in the literature.
	\end{itemize}
	Addressing these open problems would rigorously establish for the first time that ballistic propagation of particles and correlations is a rather universal property that extends to the dynamics of many-body Hamiltonians for which all four methodological obstacles (a)--(d) are present.

	In this paper, we resolve both of these open problems with the following results. 
	Specifically, let  $H_\Lam$ be a Hamiltonian of the form \eqref{eq:introH} satisfying the long-range condition \eqref{eq:long-range_assumption}. Denote $B_h:=\Set{x\in\Lam:\abs{x}\le h}$ and $B_h^\cp:=\Lam\setminus B_h$. Informally, our main results address Problems 1 and 2 as follows.
	\begin{itemize}
		\item \textbf{Uniform particle transport bound.}  Let $\al>d+1$ and assume that the initial  density lies between two positive numbers $\l_1<\l_2$.  Let $N_X$ be the self-adjoint operator encoding the number of particles inside a region $X\subset \Lam$ (see \eqref{NXdef} for the definition). Then for all $p\ge1$, there exist constants $C,\,v>0$ such that  for all $r,\,\eta>0$, 
		\begin{align}
			\notag
			\wtdel{N_{B_r}^p} \leq C  \inn{\psi_0}{{N^p_{B_{r+\eta}}}\psi_0} \qquad \textnormal{for } t<\eta/v. 
		\end{align}
		The main point here is that for any fixed $\al>d+1$, the result holds for all $p$ (compare \cite{lemm2023information,lemm2023microscopic}).
		\item \textbf{Thermodynamically stable Lieb-Robinson bound for initial states with particle-free regions.} 	Assume $\al>3d+2$ and assume that the initial state has density upper bounded by $\l>0$ as well as an annular particle-free region of width $\eta>0$ (i.e., $n_x\psi_0=0$ for $r \le \abs{x}\le r+\eta$). Then, for   any local observables $A,\,B$ supported inside $B_r$, $B_{r+\eta}^\cp$ respectively and all sufficiently large $\eta$,
		\begin{align}
			\notag
			\wzdel{[e^{itH_\Lambda}Ae^{-itH_\Lambda},B]}\le C\abs{t}\norm{A}\norm{B} \eta^{-\beta}\qquad (t<\eta/v). 
		\end{align}
		Here $C,\,v,\beta>0$ are independent of $t$, $\eta$, $\psi_0$, $A$, and $B$. The main point here is that $v$ is a constant (compare \cite{kuwahara2024effective,kuwahara2024enhanced}) and the right-hand side is independent of $N$ (compare \cite{faupin2022lieb,sigal2022propagation,lemm2023information}).
	\end{itemize}

	See \secref{secAssump} for the detailed setup and precise statements, which are more general in various minor ways.  These results extend and improve corresponding ones in earlier works \cite{faupin2022lieb,sigal2022propagation,lemm2023information,lemm2023microscopic} in the aforementioned key ways. In physical terms, we provide   natural sufficient conditions on the initial states  that ensure uniformly bounded speed of propagation of particles and information in bosonic systems with long-range hopping and interactions.\footnote{While we are interested in long-ranged bosonic systems for purely mathematical reasons (existing methodology breaks down for them), we mention in passing that  systems with long-range hopping are also a frontier in physics \cite{roses2021simulating,wang2023quantum,bespalova2024simulating} and it may be interesting to consider if there are any practical physical implications of our results.} 
	
	A few remarks on the results are in order. First, the condition $\alpha>d+1$ for the uniform particle propagation is known to be sharp by studying the non-interacting case ($V=0$) \cite{tran2020hierarchy}. Second, our results yield an explicit universal bound, $\kappa$, on the maximal speed of particle and information propagation. This $\kappa$ is naturally tied to the norm of one-particle group velocity operator. 	This is to be compared to \cite{kuwahara2024effective}, which treated short-range hopping and found a particle propagation bound with an almost linear light cone (speed growing logarithmically in time) and derived LRBs with substantially super-ballistic speed for bounded-density initial states in $d>1$. Our results establish that long-range Bose-Hubbard Hamiltonians  in initial states with particle-free regions do \textit{not} show any of this unusual behavior and instead adhere to the usual physical principle of locality: propagation is controlled by fixed maximal speed that depends on system parameters and does not grow with time.	
	
	Third, regarding the assumption on the initial states, notice that they are different between our two results: our uniform particle propagation bounds assume upper and lower bounds on the density. This includes the experimentally relevant Mott state with unit filling \cite{cheneau2012light}. For the second result, the bounded-speed Lieb-Robinson bounds, we require an initial particle free region. Indeed, since protocols exist even in the short-range case \cite{kuwahara2024effective} that show super-ballistic information propagation for general  bounded-density initial states, an additional assumption on the initial state is needed to obtain a bounded-speed LRB. It is a separate problem to derive LRBs for bounded-density initial states under the power-law condition $\alpha>d+1$, since these LRBs will almost certainly be super-ballistic (i.e., involve speeds that grow in time) as in \cite{kuwahara2024effective,kuwahara2024enhanced}. This would likely require complementing our uniform particle propagation bound with to-be-developed long-range variants of certain methods in \cite{kuwahara2024effective,kuwahara2024enhanced}. Since super-ballistic LRBs belong to a  different physical regime, we consider this a separate problem that we will consider in a separate work.\\

	We achieve our results by further developing the analytic toolkit known as the ASTLO method, an acronym that stands for adiabatic space time localization observables developed in the many-body context in \cite{faupin2022lieb,faupin2022maximal,sigal2022propagation,lemm2023information,lemm2023microscopic}. To resolve problem (i), we introduce a new logarithmic ASTLO which serves to ``renormalize'' the contribution from regions where the density is roughly uniform. 	For (ii), we devise a multiscale scheme building on our previous work \cite{lemm2023microscopic}, for the first time in the context of Lieb-Robinson bounds. We explain our proof strategy in more details after we present the setup and main theorems. We expect that the new methodology we develop here will be useful in related quantum problems in which both technical obstacles (unbounded and long-ranged) are present.\\

	\textbf{Organization of the paper.}
	In \secref{secSetup}, we formally state our setup and main results, Theorems \ref{thm4} and \ref{thm1}. 
    
    In \secref{sec prop 1 and 2}, we prove Theorem \ref{thm4} for the first moment ($p=1$) by implementing a  logarithmically renormalized ASTLO idea in the multiscale scheme developed in \cite{lemm2023microscopic} which in turn is an iteration of the original bosonic ASTLO scheme \cite{faupin2022lieb}. The proof strategy is explained in the small flowchart \ref{flow chart 2}.

	The proof of Theorem \ref{thm4} for general $p$ spans Sections \ref{second mom log}--\ref{second mom log b} and blends a multiscale induction with an induction on moment order; it is more involved and summarized in the larger flowchart \figref{flow chart 1}. More precisely, in \secref{second mom log}, we introduce the ``main ASTLO estimates'' and show how they imply \thmref{thm4}. 
        In \secref{sec proof astlo est}, we run the refined multiscale ASTLO machine to reduce the proof of the ``main ASTLO estimates'' to proving a remainder estimate. The remainder estimate  is then proved in \secref{second mom log b} by a multiscale induction. This completes the proof of \thmref{thm4} for $p\geq 2$. 
    
    Section \ref{secPfLCAgen} contains the proof of \thmref{thm1}, with some technical lemmas deferred to \secref{secPfLems}. In the Appendix, we recall and adapt basic properties in the ASTLO method to our present purposes.
	
	\section{Setup and Main Results}\label{secSetup}
	We fix a finite connected subset $\Lam\subset\mathbb Z^d$ which we will call lattice. The entire paper generalizes to any graph embedded in $\mathbb R^d$ which has polynomial volume growth and satisfies that any two distinct sites are separated by a distance greater than one.
    
	In the following, we use the convenient Fock space formalism, also known as second quantization. Throughout, we work with a fixed particle number $N$. We emphasize that working in Fock space is completely equivalent and does not affect the difficulty of the mathematical problems we consider in any way; it is mainly useful because the algebra takes care of some bookkeeping. 
	
	We consider a system of $N$ bosons whose positions lie in $\Lambda$ a finite discrete subset of $\mathbb R^d$. Mathematically, we work on the Hilbert space that is defined as the bosonic Fock space
	\begin{equation}\label{eq:fock}
		\mathcal F(\ell^2(\Lambda))=  \bigoplus_{N=1}^\infty \mathcal S_N\left(\bigotimes_{j=1}^N \ell^2(\Lambda)\right),
	\end{equation}
	where $\mathcal S_N$ denotes the projection onto the permutation-symmetric subspace of ${\bigoplus_{N=1}^\infty \bigotimes_{j=1}^N \ell^2(\Lambda)}$.
	On this Fock space, we consider a realization of the canonical commutation relations, i.e.,
	\[
	[a_x,a_y^*]=\delta_{x,y},\quad     [a_x^*,a_y^*]=   [a_x,a_y]=0,\qquad \forall x,y\in \Lambda.
	\]
	On $\mathcal F(\ell^2(\Lambda))$, we consider Bose-Hubbard type  Hamiltonians of the form
	%	. {The Hamiltonian is a linear unbounded operator 
		\begin{align}\label{H} %{H\equiv} 
			\boxed{ H_\Lam :=  \sum_{x,y \in \Lambda} J_{xy} a_x^*a_y + \frac12\sum_{\substack{x,y\in \Lam
						%			\\x\ne y
				}}  V_{xy} a_x^*a_y^*a_ya_x.}
		\end{align}  
		Here,         
		%		$n_x:=a_x^* a_x$ is the local number operator, 
		$J=(J_{xy})$ is a Hermitian $|\Lambda|\times |\Lambda|$ matrix representing the energy of individual particles, and 
		{ $V=(V_{xy})$	is a real-symmetric matrix, i.e.~$V_{xy}=\overline{ V_{xy}}=V_{yx}$.}

		For background on the second quantization formalism, see e.g.~\cite{reed1978iv}. 	Moreover, \cite[App.\ A]{faupin2022lieb} proves that $H_\Lam$ is self-adjoint on 
		the dense domain
		\begin{align}
			\notag
			\mathcal D(H_\Lam)=\{(\psi_N)_{N\geq 0}\in \mathcal F(\ell^2(\Lambda))\,:\, \sum_{N\geq 0}\|H_\Lambda \psi_N\|^2<\infty \}.
		\end{align}
		
		We consider solutions $\psi_t$ to the time-dependent many-body Schr\"odinger equation 	\begin{align}\label{SE}
			i\di_t\psi_t = H_\Lam \psi_t\quad \textnormal{ on }\quad \cF(\ell^2(\Lam)),
		\end{align} with initial datum satisfying the well-posedness condition
		\begin{align}\label{g0-cond}
			\psi_0\in \cD(H_\Lam)\cap\cD(N_\Lam^{p/2}),\quad p\ge1.
		\end{align}
		Here and below,		for any subset  $X\subset \Lam$, we define the particle number operator
		\begin{align}
			\label{NXdef}
			N_X:=\sum_{x\in X}n_x.
		\end{align}
		Since $H_\Lam$ is self-adjoint, it generates a unitary group $e^{-\mathrm{i}tH_\Lam}$ on $F(\ell^2(\Lambda))$ via Stone's theorem. This unitary group encodes the solution operator to the PDE \eqref{SE} and we will study its propagation properties. As usual, we will use a dual formulation in which the unitary acts by conjugation on operators (known as the ``Heisenberg picture'' in physics).
		
		\subsection{Power-law decay assumption}\label{secAssump}

		Let $\al>d+1$ and set
		\begin{align}
			\label{CJdef}
			C_{J,\al}:=&	\sup_{x,y\in\Lam}\abs{J_{xy}}(1+\abs{x-y})^{\al},\\
			C_{V,\al}:=& \sup_{x,y\in\Lam} \abs{V _{xy}}(1+\abs{x-y})^{\al}\label{CVdef}.
			%			\\
			%			.
		\end{align} These constants then capture the power-law decay because, for any $x\ne y$,
		\begin{align}
		    \abs{J_{xy}}\leq C_{J,\al}\abs{x-y}^{-\al},\qquad  \abs{V _{xy}} \leq C_{V,\al}\abs{x-y}^{-\al}. \label{no 1}
		\end{align}

		Our main results below depend on \eqref{CJdef}--\eqref{CVdef} for appropriate values of $\al$. It is understood that $C_{J,\al},C_{V,\al}$ are $\mathcal O(1)$-constants that are independent of $\Lambda$.

		% For the particle transport bounds in \thmref{thm4}, we require 
		% \begin{align}
			% 	\label{Jcond}
			% 	C_{J,\al}<\infty.
			% \end{align}
		% For the Lieb-Robinson bounds in \thmref{thm3}, in addition to \eqref{Jcond}, we further  require
		% \begin{align}
			% 	C_{V,\al}<\infty.\notag
			% \end{align}

		A central role is played by the first moment of the hopping matrix,
		\begin{equation}\label{kappa}
			\kappa:=  \sup_{x\in\Lam}\sum_{y\in\Lam}\abs{J_{xy}}\abs{x-y}.
		\end{equation}
		By Schur's test, it is easy to check that $\kappa$ is an upper bound of the $1$-particle momentum operator. 
		Using $\abs{J_{xy}}\leq C_{J,\al}(1+\abs{x-y})^{-\al}$ for $\alpha>d+1$, we see that $\kappa$ is bounded by a dimensional constant times $C_{J,\alpha}$, uniformly in $\Lambda$.

		In the remainder of this paper we consider $\alpha$ to be fixed and so we will drop the dependence of \eqref{CJdef}--\eqref{CVdef} on $\al$ whenever no confusion arises.\\
		
		\textbf{Notation.}
		In what follows, for a subset $X\subset \Lam$, we denote by 
		%$\diam X:=\sup_{x,y\in X}{\abs{x-y}}$ its diameter,  
		$X^\cp:=\Lam\setminus X$ its complement in $\Lam$, $d_X(x)\equiv \dist(\Set{x},X):=\inf_{y\in X}\abs{x-y}$ the distance function to $X$, $X_\xi:=\Set{x\in\Lam:d_X(x)\le \xi}$ (see \figref{fig:Xxi} below), 
		%and $N_X$  the local particle number operator: \begin{equation}\label{1.4}	N_X := \sum_{x\in X}a^*_xa_x. \end{equation}  We note that $[N_X,N_Y]=0$ for any domains $X,\,Y\subset \Lam$,
		and $X_\xi^\cp$ is always understood as $(X_\xi)^\cp$.
		$\cF\equiv \cF(\ell^2(\Lambda))$ stands for the bosonic Fock space over $\Lam$. 
		Finally, 		given an operator $A$ on $\mathcal F$ and an initial state $\psi_0\in \mathcal F$, we abbreviate 
		\begin{align}\notag
			\br{A}_t:=\br{\psi_t,A\psi_t},\qquad \psi_t:=e^{-iH_\Lam t}\psi_0.
		\end{align}
		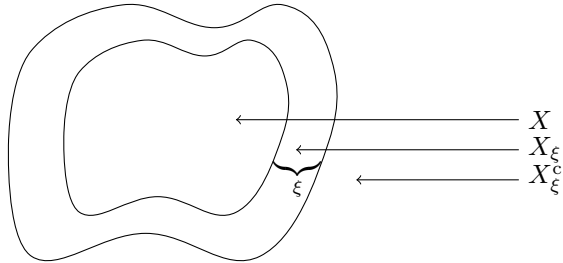
\begin{figure}[H]
			\centering
			\begin{tikzpicture}[scale=.8]
				\draw  plot[scale=.52,smooth, tension=.7] coordinates {(-3,0.5) (-2.5,2.5) (-.5,3.5) (1.5,3) (3,3.5) (4,2.5) (4,0.5) (2.5,-2) (0,-1.5) (-2.5,-2) (-3,0.5)};
				
				\draw  plot[shift={(-0.2,-0.25)}, scale=.76,smooth, tension=.7] coordinates {(-3,0.5) (-2.5,2.5) (-.5,3.5) (1.5,3) (3,3.5) (4,2.5) (4,0.5) (2.5,-2) (0,-1.5) (-2.5,-2) (-3,0.5)};
				
				\draw [->] (6,.5)--(1.3,.5);
				\node [right] at (6,.5) {$X$};
				
				\draw [->] (6,0)--(2.3,0);
				\node [right] at (6,0) {$X_\xi$};
				
				\draw [->] (6,-.5)--(3.3,-.5);
				\node [right] at (6,-.5) {$X_\xi^\cp$};
				
				\node [below] at (2.32,0) {$\underbrace{}_{\xi}$};
			\end{tikzpicture}
			\caption{Schematic diagram illustrating the notations above.}
			\label{fig:Xxi}
		\end{figure}

		\subsection{Main result 1: uniform particle propagation bound}\label{secContrDen}
		We present our first main result, uniform particle propagation bounds for initial states with uniformly bounded density from both above and below, see \eqref{CD}. The result heavily builds on our previous work \cite{lemm2023microscopic} and improves that approach in several key ways.  
		Specifically, the result allows weaker decay of the hopping terms of the Hamiltonian \eqref{H}, and the decay condition is now independent of the moment $p$ of the number operator that we bound. More precisely, we are able to improve the $p$-dependent power-law decay condition from \cite{lemm2023microscopic} $\al>\max\Set{\frac{3dp}{2}+1,2d+1}$ to the sharp
		one
		\[
		\al>d+1
		\]
		for every $p\ge 1$. This already gives an improvement  for the case $p=1$.		
		\begin{theorem}[Main result 1, particle propagation bound]\label{thm4}
			Let $\al>d+1$. Consider an initial state satisfying the bounded-density assumption
			\begin{align}\label{CD}
				(\l_1r^d)^q\le \langle N_{B_r(x)}^q\rangle_0\le (\l_2r^d)^q \qquad  (1\le q\le p,\; x\in \Lam,\; r\ge1),
			\end{align}
			for some $\infty>\l_2>\l_1>0$ and $p\ge1$. {Then, for any  $v>12\kappa$ and $\delta_0>0$,  there exist 
				\begin{align}
					&\rho=\rho( v,\delta_0,d,\l_1,\l_2,p,\al,C_{J,\alpha})>0\label{rho 1}\\
					&C=( v,\delta_0,d,\l_1,\l_2,p),\notag
				\end{align} 
				with $C=0$ for $p=1$ and $C>0$ for $p>1,$ such that 
				for all $R>r\ge\rho$ with $R-r>\delta_0r$, there hold}
			\begin{align}
				\br{N^p_{B_{r}}}_t\le& \br{N^p_{B_{R}}}_0\exp\Set{\frac{C}{R^d}+\frac{vt}{R-r}} \label{main ub 1}\qquad &&(0\le vt\le R-r),\\
				\br{N^p_{B_{R}}}_t\ge& \br{N^p_{B_{r}}}_0\exp\Set{-\del{\frac{C}{R^d}+\frac{vt}{R-r}}}&& (0\le vt\le R-r).\label{main lb 1}
			\end{align}
		\end{theorem}
		
		The proof of this theorem is found at the end of \secref{sec prop 1 and 2} for $p=1$ and        
		completed in \secref{second mom log}--\ref{second mom log b} for all $p$. It requires a number of new ideas to refine the multiscale ASTLO (adiabatic space time localization observables) scheme introduced in our previous paper \cite{lemm2023microscopic}. Since these are ideas that are implemented at every stage of a multiscale induction, the details are slightly technical and so we make an effort to give an overview of the proof strategy (and what is new compared to \cite{lemm2023microscopic}) using flowcharts in the respective proof sections.

		\begin{remark}
			\begin{enumerate}[label=(\roman*)]
				\item 
				The right-hand side of the inequalities \eqref{main ub 1} and \eqref{main lb 1}  contains a multiplicative error term. We emphasize that this term is harmless; it is uniformly bounded from above and below (and usually close to $1$) in the relevant situations: First, for large $R\gg 1$, $\me^{C/R^d}$ rapidly converges to $1$. Second, inside the range of validity $vt\le R-r$, the second term is bounded by $\me^{\frac{vt}{R-r}}\leq \me$, a universal constant. We only use the slightly more precise form $\me^{\frac{vt}{R-r}}$ because it converges to $1$ as $t\to 0$.
				
				\item 
				
				The proof controls particle propagation down to microscopic, i.e.,  $\mathcal O(1)$ length scales as desired. In the above theorem, the precise value of the minimal length scale is called $\rho$ and it depends on various parameters as shown in \eqref{rho 1}, but not on the system size parameters $\Lambda$ or $N$; in this sense it is $\mathcal O(1)$. This can be relaxed in two ways: First, since we use a downwards multiscale induction scheme that proceeds from larger to smaller length scales, the controlled density assumption \eqref{CD} can in fact be relaxed to only hold for the scales $\rho$ and larger scales. Second, for smaller $r$ one can still extend the upper bound \eqref{main lb 1} by embedding the relevant ball in $B_\rho$ at the sole cost of gaining a constant factor (so the right-hand side would no longer go to $1$ as $R\to \infty$).

				\item 
				One can equivalently view \eqref{main ub 1}--\eqref{main lb 1} as locally weighted-$\ell^2$ estimates for the solution $\psi_t$. Define
				\begin{align}
					\Norm{\psi}_{r,p}:=\norm{N_{B_r}^{p/2}\psi}_{\ell^2}.\notag
				\end{align}
				Then \thmref{thm4} can be formulated as a priori estimates for \eqref{SE}: for initial states $\psi_0$ satisfying \eqref{CD} and all $R>r\ge \rho$, 
				\begin{align*}
					\Norm{\psi_t}_{r,p}^2\le& Ce^{vt/(R-r)}\Norm{\psi_0}_{R,p}^2\qquad &&(0\le vt\le R-r),\\
					\Norm{\psi_t}_{R,p}^2\ge&  Ce^{-vt/(R-r)}\Norm{\psi_0}_{r,p}^2 && (0\le vt\le R-r).
				\end{align*}
			\end{enumerate}
		\end{remark}

		\subsection{Lieb-Robinson--type bounds}\label{secBddDen}
		We present our second main result, thermodynamically stable Lieb-Robinson bounds for long-range bosonic Hamiltonians.
		
		We begin with some  preliminary notations.    	We say that an operator $A$ acting on the bosonic Fock space $\cF\equiv \cF(\ell^2(\Lambda))$ (see \eqref{eq:fock}) is \textit{localized} in $X\subset \Lam$ if \begin{align}\label{A-loc}
			\sbr{A,a_x^\#}=0\qquad  (x\in X^\cp,a_x^\#=a_x, a_x^*).
		\end{align}
		Denote by $\supp A$ the intersection of all  $X$ s.th.~\eqref{A-loc} holds. 
		We say $A$ is localized in $X$ if and only if  $\supp A\subset X$. Define the set of operators
		\begin{equation}\label{BX}
			\cB_X :=\Set{A\in\cB(\cF):[A,N_\Lam]=0,\, \supp A\subset X},
		\end{equation}
		where  $\cB(\cF)$ is the space of bounded operators on $\cF$. 
		
		As usual for Lieb-Robinson bounds, we study the dual evolution to \eqref{SE}, the Heisenberg evolution on operators $A\in \mathcal B_X$, 
		\begin{equation}\label{1.5}
			\al_t(A):=\me^{itH_\Lam}A\me^{-itH_\Lam}.
		\end{equation}
		For a subset $S\subset \Lam$ and $A\in\mathcal B_S$, we define the \textit{localized (Heisenberg) evolution} by
		\begin{equation}\label{alloc}
			\al_t^S(A):=\me^{itH_S}A\me^{-itH_S},
		\end{equation}
		where $H_S$ is defined by \eqref{H} but with $S$ in place of $\Lam$. We recall $C_{J,\al}$, $C_{V,\al}$ defined in \eqref{CJdef}--\eqref{CVdef}.

		\begin{theorem}[Main result 2, light-cone approximation]\label{thm1}
			Assume that
			\begin{align}\label{alCond}
				{\al>3d+1}.
			\end{align} 
			Let $X\subset \Lam$. Assume that the initial state satisfies the following two conditions.
			\begin{itemize}
				\item \textit{Density bound:}  For some $\l>0$,
				\begin{align}\label{UDBp}
					\wzdel{ N_{B_r(x)}^q}\le(\l r^d)^q \qquad  (q=1,2,\,x\in\Lam,\, r\ge1).
				\end{align}
				\item \textit{Particle-free shell:}
				For some $X\subset\Lam$ and  $\xi\ge1$,
				\begin{equation}\label{locCond}
					N_{X_{2\xi}\setminus X}\psi_0=0.
				\end{equation}
			\end{itemize}
			Then, for every $v>2\kappa$,  there exists a positive constant
			\begin{align}
				\label{C1C2'}
				C=C(\al,d, C_{J,\alpha},C_{V,\alpha}, v,\l,X),	
			\end{align}
			such that for all $\xi\ge\max(2,\diam X)$,  $Y\subset\Lam$ with $\dist(X, Y)\ge2\xi$, and  operators $A\in \cB_X$, $B\in \cB_Y$,
			there holds
			{			\begin{align}\label{lcae-gen}
					|\br{B(\al_t(A)-\al_t^{X_\xi}(A))}_0|\le C\norm{A}\norm{B} \abs{t}  \xi^{-\beta} ,
			\end{align}}
			for all $\abs{t}< \xi/v$. The error exponent $\beta$ is given by
			\begin{align}
				\label{betaDef}\beta:=& \lfloor\al-3d-1\rfloor.
			\end{align}
		\end{theorem}
		This result is proved in Section \ref{secPfLCAgen}. See \remref{remXdep} for discussion on the $X$-dependence in estimate \eqref{lcae-gen}. In particular, the constant \eqref{C1C2'} is independent of the size of $X$.

		It can be easily verified using  definitions \eqref{A-loc}, \eqref{alloc}, and the relation $[H_X,N_\Lam]=0$ that, for any $A\in\cB_X$, the local evolution   of $A$ remains localized,  as $A\in\cB_X$ implies   $\al_t^S(A)\in\cB_S$  for all $t$ and subset $S$ containing $X$. 
		In contrast, even for localized $A\in\cB_X$, the full evolution $\al_t(A)$ in general cannot remain localized \textit{in any nontrivial time interval}. In other words, under the full Heisenberg evolution, the support of an observable generally spreads over the entire domain $\Lam$ immediately for any $t>0$. 
		
		Theorem \ref{thm1} shows that under suitable conditions on the initial states, the full evolution $\al_t(A)$ can be well-approximated, \textit{up to a remainder uniform in $\Lam$}, by the localized evolution $\al_t^{X_\xi}(A)$ supported  inside the light cone $X_\xi$. This is closely related to a Lieb-Robinson bound on the commutator. Indeed, the latter follows straightforwardly, as we note  below.

		\begin{theorem}[Thermodynamically stable Lieb-Robinson bound for long-range bosons] 
			\label{thm3}
			Suppose the assumptions of \thmref{thm1} hold. {Then, for every $v>2\kappa$, there exists {$C=C(\al,d, C_{J,\alpha},C_{V,\alpha}, v,\l,X)>0$} s.th.~for all }  $Y\subset\Lam$ with $\dist(X, Y)\ge2\xi$ and  operators $A\in \cB_X,\,B\in \cB_Y$, 
			we have the following estimate for all $\abs{t}< \xi/v$:
			\begin{align}
				\abs{\br{[ \al_t (A), B] }_0}  
				\le C \norm{A}\norm{B}\abs{t}\xi^{-\beta}.\label{LRB}
			\end{align}
		\end{theorem}

		In contrast to prior results of this kind for long-range Bose-Hubbard Hamiltonians \cite{faupin2022lieb,sigal2022propagation,lemm2023information}, the right-hand side is fully independent of the total particle number, which makes it thermodynamically stable.

		\begin{proof}
			We introduce the remainder term corresponding to the l.h.s.~of \eqref{lcae-gen}:
			\begin{align}\label{Rem-def 1}
				\Rem_t(A):=\al_t(A)-\al_t^{X_\xi}(A).
			\end{align}
			Since the localized evolution $\al_t^{X_\xi}(A)$ and $B$ are respectively localized in $X_\xi$ and $Y\subset X_{2\xi}^\cp$,  it follows that  $[\al_t^{X_\xi}(A),B]=0$ and so \eqref{Rem-def 1} implies
			\begin{align}
				[\al_t(A), B] &=  [ \Rem_t(A), B]. \label{At-B-com}
			\end{align} 
			Furthermore, from definition \eqref{Rem-def 1} follow the relations % 
			\begin{align}
				\label{RemRel}
				(\Rem_t(A))^*=\Rem_t(A^*),\qquad\abs{\br{\Rem_t(A)B}_0}=\abs{\br{B^*(\Rem_t(A))^*}_0 }.
			\end{align}
			By Theorem \ref{thm1}, for $v>2\kappa$ and $\abs{t}<\xi/v$, there exists $C>0$ s.th.
			\begin{equation}\label{62'}
				\begin{aligned}
					&\abs{\br{B\,\Rem_t(A)  }_0}\le   C \norm{A}\norm{B}\abs{t} {\xi^{-\beta} }.
				\end{aligned}
			\end{equation}
			Thanks to \eqref{RemRel}, replacing $A,\,B$ in \eqref{62'} by $A^*,\,B^*$ yields the same estimate on $\abs{\br{\Rem_t(A)B}_0}$: 
			\begin{align}
				\label{62''}
				\abs{\br{\Rem_t(A)\,B}_0}\le C \norm{A}\norm{B} \abs{t}\xi^{-\beta} .
			\end{align}
			The desired estimate \eqref{LRB} follows from equality \eqref{At-B-com}, the triangle inequality   $\abs{\br{[\Rem_t(A), B]}_0}\le		\abs{\br{\Rem_t (A)B }_0}+\abs{\br{B\,\Rem_t(A)  }_0}$, estimates \eqref{62'}--\eqref{62''}.
		\end{proof}

		We close with a few remarks on the results presented in this section.
		\begin{remark} \label{remXdep}
			\begin{enumerate}[label=(\roman*)]
				\item 
				
				The particle-free shell condition \eqref{locCond}, which goes back to \cite[Eq.~(1.4)]{faupin2022lieb}, stipulates that there exists an empty shell region of width $\xi$ around some fixed domain $X$. Here we prove for the first time a bound that is valid for long-range Hamiltonians that is thermodynamicall stable, i.e., independent of the total particle number. For the experts familiar with the early, seminal work \cite{schuch2011information}, we point out that their setting is different because it places all particles initially in a ball. This makes it not suitable to study thermodynamic stability, because in this setting the global particle number is the same as the particle number in the initial ball $B$, i.e., $\wzdel{N_\Lam^p} =\wzdel{N_B^p}$ for all $p\ge1$.
				\item The particle-free shell condition can be dropped, but at the steep price of adding an additional term on the right-hand side of the form
				\[
				{C_1\br{N_{X_{2\xi}\setminus X}N_\Lam }_0},\qquad C_1=C_1(\al,d, C_{J,\alpha},C_{V,\alpha}, v,X).
				\]
				See \secref{sec75} for details. We do not discuss this extension further.
				\item
				The dependence on $X$ in the constant $C$ in these results comes  from \eqref{C1C2'}. The latter, in turn, depends only  on the number of balls needed to cover the boundary of $X$; see the proof of \thmref{thmAnnMVB} for details. By adjusting the geometric construction of the ASTLO, it is possible to remove the $X$-dependence altogether.
				\item 		In \cite{kuwahara2024effective}, the authors obtained Lieb-Robinson bounds that are	independent of the strength of $V$ (i.e.,  $C_{V,\al}$ from \eqref{CVdef}) for finite-range interactions by resorting to a clever maneuver exploiting a suitable interaction picture. 
				%See the bottom half of their page 67 for the main idea. 
				This maneuver utilizes the finite-range assumption in a fundamental way and is not expected to hold in the same way for long-range interactions. As our focus in the present paper is on long-range Hamiltonians, it is natural to expect that dependence on the interaction $V$ enters in our bounds through $C_{V,\al}$. 
				\item  Through the technical lemmas in \secref{secPfLems},   for any integer $Q\ge1$, Thms.~\ref{thm1}--\ref{thm3} naturally extend to the class of Hamiltonians
				\begin{align}
					\label{HQ}
					H_{\Lam}^Q :=  \sum_{x,y \in \Lambda} J_{xy} a_x^*a_y + \frac12\sum_{\substack{x,y\in \Lam
							%			\\x\ne y
					}}\sum_{q=1}^Q V^{(q)}_{xy} n_x^{q/2}n_y^{q/2},
				\end{align}	
				with real-symmetric matrices 	 $V^{(q)}=(V^{(q)}_{xy})$. We give the precise statement in \propref{propThm2gen}.

			\end{enumerate}
		\end{remark}

		\section{Proof of  \thmref{thm4} for \texorpdfstring{$p=1$}{p=1}}\label{sec prop 1 and 2}

		We begin with the proof of \thmref{thm4} for  the first moment case $p=1$.  The outline of the upcoming subsections is as follows: 
		\begin{itemize}
			\item In \secref{sec7.0}, we introduce the main ingredient of the proof, namely the ASTLOs, together with some preliminary notations. 
			\item In \secref{sec7.1}, \textit{assuming} the conclusion of \thmref{thm4} with $p=1$ and all 
			\begin{align}
				\label{aCond}
				1<	R/r<4,
			\end{align}
			we prove the same conclusion for all $R/r\ge 4$.  
			\item In \secref{first mom log}, we prove \thmref{thm4} for $p=1$ under condition \eqref{aCond} and for time up to some $T_1>0$.
			\item In \secref{secTest}, we prove that $T_1\ge (R-r)/v$.
			\item In \secref{sec proof thm4} we complete the proof of \thmref{thm4} for the first moment case, thanks to the lower bound on $T_1$ obtained in the previous section.
			\item In \secref{int astlo for p=1} we show an integral estimate for the ASTLO, that will constitute part of the induction basis to show \thmref{thm4} for $p\ge2$.                                                 
		\end{itemize}

		\subsection{ASTLOs and Their Properties}
		\label{sec7.0}
		Given $v>\kappa$ and $R>r\ge1$  (see \eqref{kappa}), we define the quantity ${\s=(R,r)}$ and write \begin{align}
			v':=&\frac{v+\kappa}{2}\label{v'def},\\
			\label{sSiDef}
			s_{\s}:=&\frac{2(R-r)}{3v},\\
			\omega:=&v-v'.\label{omega def}
		\end{align}
		Let $\om>0$. 	Consider the function class 
		\begin{equation}\label{classE}
			\mathcal{E}:=\left\{
			f\in C^{\infty} (\mathbb{R})\left\vert \begin{aligned} &f\geq0,\: f\equiv 0 \text{ on } (-\infty,\omega/2],f\equiv 1 \text{ on } [\omega,\infty) \\ &f'\geq 0,\:\sqrt{f'}\in C_c^{\infty} (\mathbb{R}),\:\supp f'\subset(\omega/2,\omega)
			\end{aligned}\right.\right\}.
		\end{equation}
		For functions $f\in\cE$, we introduce the following change of variables (c.f.~\figref{figfpm}):
		\begin{align}
			f\left(\frac{R-v't-\abs{x}}{s_{\s}}\right)&=:f_{-}(\abs{x}_{t,\s}),\label{f-}\\
			f\left(\frac{(R+2r)/3+v't-\abs{x}}{s_{\s}}\right)&=:f_{+}(\abs{x}_{t,\s})\label{f+}.
		\end{align}
		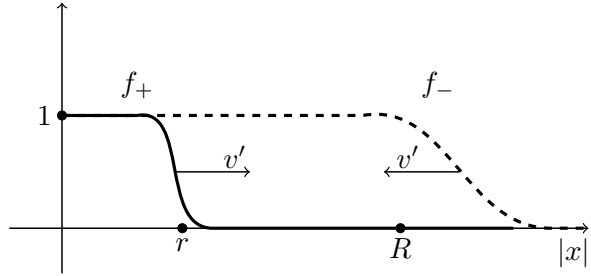
\begin{figure}[H]
			\centering
			
			\begin{tikzpicture}[scale=1]
				\draw [very thick] (0,1.5)--(1,1.5) [out=10, in=-180] to (2,0)--(6,0);
				\draw [dashed, very thick] (0,1.5)--(4,1.5) [out=10, in=-180] to (6.5,0)--(7,0);
				\draw [->,line width=0.2mm](0,-0.6)--(0,3);
				\draw [->,line width=0.2mm](-0.7,0)--(7,0);
				\draw [->,line width=0.2mm](1.5,0.75)--(2.5,0.75);
				\draw [<-,line width=0.2mm](4.28,0.75)--(5.28,0.75);
				\node[scale=1.1] at (4.6, 0.95) {$v'$};
				\node[scale=1.1] at (2.3, 0.95) {$v'$};
				\node[scale=1.1, below] at (1.6, 0) {$r$};
				\node[scale=1.1, below] at (4.5, 0) {$R$};
				\draw [fill] (1.6,0) circle [radius=0.06];
				\draw [fill] (4.5,0) circle [radius=0.06];
				\node[scale=1.1] at (1, 1.9) {$f_+$};
				\node[scale=1.1] at (5, 1.9) {$f_{-}$};
				\node[scale=1.1, left] at (0, 1.5) {$1$}; 
				\node[scale=1.1, below] at (6.8, 0) {$\abs{x}$};
				\draw [fill] (0,1.5) circle [radius=0.06];
				
			\end{tikzpicture} 
			\caption{Schematic diagram illustrating the movement of $f_+$ and $f_-$.}   \label{figfpm}
		\end{figure}

		% \begin{figure}
			% 	\centering
			
			% 	\begin{tikzpicture}[scale=.8]

				% 		\draw[black,line width=0.2mm]
				% 		%(0,0) circle (4.5)
				% 		(0,0) circle (1.5)
				% 		(0,0) circle (3.);
				
				% 		\draw[->, line width=0.4mm] (1.8,1)  parabola (0.2,0.2);
				% 		\draw[<-, dashed, line width=0.4mm] (3.5,-2)  parabola (0.2,-0.5);
				% 		\draw[->, dashed, line width=0.4mm] (3,2.8)  parabola (0.2,0.8);
				%                    %\draw[<-,>=stealth', line width=0.4mm] (0,-1)  parabola (1.5,-2);
				%                     \draw[<-, line width=0.4mm]  (1.2,-2)  parabola (0,-1);
				% 		\node[scale=1.1] at (-1.3, 1.3) {$2^{k}$};
				% 		\node[scale=1.1] at (-2.43, 2.43) {$2^{k+1}$};
				% 		\node[scale=1.1] at (-1.3, 1.3) {$2^{k}$};
				% 		\node[scale=1.1] at (-0.7, -0.9) {I};
				% 		\node[scale=1.1] at (-1.2, -2) {II};
				% 		\node[scale=1.1] at (-2, -3) {III};
				
				% 	\end{tikzpicture}    

			% 	\caption{Controlling $g(\abs{x}_{+,ts})$ will tell us that most particles can not jump from region I to region III, but only until region I, while the $f(\abs{x}_{-,ts})$ will tell us that particles coming from region II can make it until region I, but particles coming from region III can not make it. }
			% \end{figure}

		Given these functions, we define the ASTLO (\textit{adiabatic space-time localization observables}) as
		\begin{align}\label{ASTLOdef}
			\Nf[]{f_{\#}}{t}{\s}:=\sum_{x\in\Lambda}f_\#(\abs{x}_{t,\s})n_x.
		\end{align}

		{In essence, $f_-(\abs{x}_{t,\s})$ will allow us to obtain upper bounds on the propagation, controlling particles going in the ball $B_{R}$, while $f_+(\abs{x}_{t,\s})$ will allow us to control particle moving out of $B_{r}$, which will give us the lower bound on particle propagation.}
		
		{The basic properties of the ASTLOs are summarized below as several lemmas.}
		First, we relate the ASTLOs to the number operator associated to the propagation regions.
		\begin{lemma}\label{ASTLO and N}
			Given $R>r>0$ and function $f\in\cE$,
			the following holds for all $t\le s_\si$:
			\begin{align}
				{\Nf{f_-}{0}{\s}}\le&{N_{B_{R}}} \label{f and B at t=0},\\
				{\Nf{f_-}{t}{\s}}\,\ge&{N_{B_{r}}}\label{f and b at t}, \\
				{ N_{f_+,0\s} } \;\ge & N_{B_{r}} \label{g and b 0},\\
				{ N_{f_+,t\s}}\;\,\le&   N_{B_{R}}  \label{g and B t}.
			\end{align}
		\end{lemma}
		The proof can be found in Appendix \ref{proof lemma aslto and N}.

		Next, we relate $N_{f_+,t\si}$ and $N_{f_-,t\si}$ with each other:
		\begin{lemma}\label{lem7.3}
			Given $R>r>0$ and function $f\in\cE$,
			\begin{align}
				{\Nf{f+}{t}{l}}\le{\Nf{f-}{t}{l}},\qquad 		 t\le\frac{ R-r}{3v'}.\notag
			\end{align}
			
		\end{lemma}
		The proof can be found in Appendix \ref{secPfLem7.3}.
		
		Lastly, we record an important analytical property for functions in the class $\cE$:
		{		\begin{lemma}\label{new taylor}
				Given any $0<\eps \le 1$ and  function $f\in \cE$, there exists $C_{f}>0$ s.th., 			for $u:=\sqrt{f'}
				$,
				\begin{align}\label{fExp}
					\abs{f(x)-f(y)}\le u(x)u(y)\abs{x-y}+C_f\abs{x-y}^{1+\eps}\quad  (x,y\in\R).
				\end{align}
				
		\end{lemma}}
		This lemma is proved in Appendix \ref{secPfExp}

		\subsection{Proof of Particle Propagation Estimates (\thmref{thm4}) for \texorpdfstring{$R/r$}{R/r} large}
		\label{sec7.1}

		The main result of this section is the following:
		\begin{proposition}\label{prop7}
			Assume \eqref{main ub 1}--\eqref{main lb 1} hold {for some $v_0>0$ and $R/r$ satisfying \eqref{aCond}}. 	Then \eqref{main ub 1}--\eqref{main lb 1} hold  {for $v:=4v_0$ and all $R/r>4$}. 
		\end{proposition}
		\begin{proof}
			
			We  start by showing \eqref{main ub 1}.
			Since $R/r> 4$, it follows that there exists some   $k=k(R,r)>0$, such that 
			\begin{align}\label{case 2 again}
				{r<2^{k}<2^{k+1}\le R<2^{k+2}}.
			\end{align}
			Applying \eqref{main ub 1} for $v_0$ and $(r',R')=(2^{k},2^{k+1})$, which is possible since $R'/r'=2$ satisfies \eqref{aCond},  we find
			\begin{align}
				\br{N^p_{B_{2^{k}}}}_t\le \exp\Set{\frac{C}{2^{d(k+1)}}+\frac{v_0t}{2^{k}}}\br{N^p_{B_{2^{k+1}}}}_0\qquad(0\le v_0t\le 2^{k}). \label{715}
			\end{align}
			Using \eqref{case 2 again} as a `bridge', we deduce from \eqref{715} that
			\begin{align} 
				\br{N^p_{B_{r}}}_t\le \exp\Set{\frac{C}{2^{d(k+1)}}+\frac{v_0t}{2^{k}}}\br{N^p_{B_{R}}}_0\qquad(0\le v_0t\le 2^{k}).  \label{al}
			\end{align}
			Again by \eqref{case 2 again} we have
			\begin{align}\label{717}
				%			2^{k}\le 
				R-r\le R\le 4\cdot 2^{k}.
			\end{align}
			Using \eqref{717} to interchange $2^k$ in \eqref{al} to $R-r$, and \eqref{case 2 again} to interchange $R$ and $2^{k+1}$, we find
			\begin{align}
				\br{N^p_{B_{r}}}_t\le \exp\Set{\frac{2^dC}{R^d}+\frac{4v_0t}{R-r}}\br{N^p_{B_{R}}}_0\qquad(0\le 4v_0t\le R-r).  \label{717'}
			\end{align}
			From \eqref{717'} we conclude  \eqref{main ub 1} with $v=4v_0$. We can show the lower bound \eqref{main lb 1} via similar reasoning.
		\end{proof} 
		
		\subsection{Proof of Particle Propagation Estimate (\thmref{thm4}) assuming $R/r$ small}\label{first mom log}
		
		In this section we derive upper and lower bound on the time evolution of the number operator restricted on the propagation regions, assuming $R$ and $r$ satisfy \eqref{aCond}. The proof strategy builds on the work done in \cite{lemm2023microscopic}, with new fundamental improvements. 
		\begin{itemize}
			\item  The first step consists in obtaining time estimates for the ASTLO, and here lies the first new key idea. Instead of differentiating the ASTLO themselves, we differentiate their logarithm. This procedure is fundamental as it creates an additional term at denominator
			\begin{align}\label{after int intro}
				\pm\frac{d}{dt}\log\br{\Nf{f_\mp}{t}{l}}_t\le
				&(\kappa-v') s_l^{-1}\frac{\br{ N_{f_\mp',tl}}_t}{\br{\Nf{f_\mp}{t}{l}}_t}  \notag\\
				&+\frac{ C_{f}}{{s_l^{1+\eps}}} \left(\frac{\br{N_{B_{{R_l}}}}_t}{\br{\Nf{f_\mp}{t}{l}}_t}+\sum_{x\in B_{{R_l}-\eps s/2}}\sum_{y\in\Lambda}\abs{J_{xy}}|x-y|^{1+\eps}\frac{\br{ n_y}_t }{\br{\Nf{f_\mp}{t}{l}}_t}\right),
			\end{align}
			which will allow us to better control the contribution from far away particles.
			\vspace{1.5mm}
			\item The second new tool comes into play as we want to avoid the new term in denominator of \eqref{after int intro} becoming too small. To this end we introduce a first bad time, $T_1$, until which all scales are well-behaved, namely not too empty, and $\br{\Nf{f_\mp}{t}{l}}_t$ is big enough.
			\vspace{1.5mm}
			\item To show that the bad time  $T_1$ is not too small compared to the  length scales at play we need to derive lower bound for the particle propagation until time $T_1$. This is where the $f_+$ comes into play.
			\vspace{1.5mm}
			\item To be able to show the lower bound on particle propagation we need to first derive upper bound on the particle propagation, here is where $f_-$ is necessary.
			\vspace{1.5mm}
			\item To show the upper bound we need to bound the remainder term between bracket appearing in \eqref{after int intro}.  We bound the remainder term thanks to a downward multiscale induction as in \cite{lemm2023microscopic}. Compared to \cite{lemm2023microscopic} we are able to better control the remainder thanks to the additional term at denominator that was produced in the first step.
		\end{itemize}
		
		In the flowchart in \figref{flow chart 2} we give a schematic idea of the proof of \thmref{thm4} for $p=1$ and small $R/r$. Notice that this section is the starting point for the proof of \thmref{thm4} for higher $p$.

		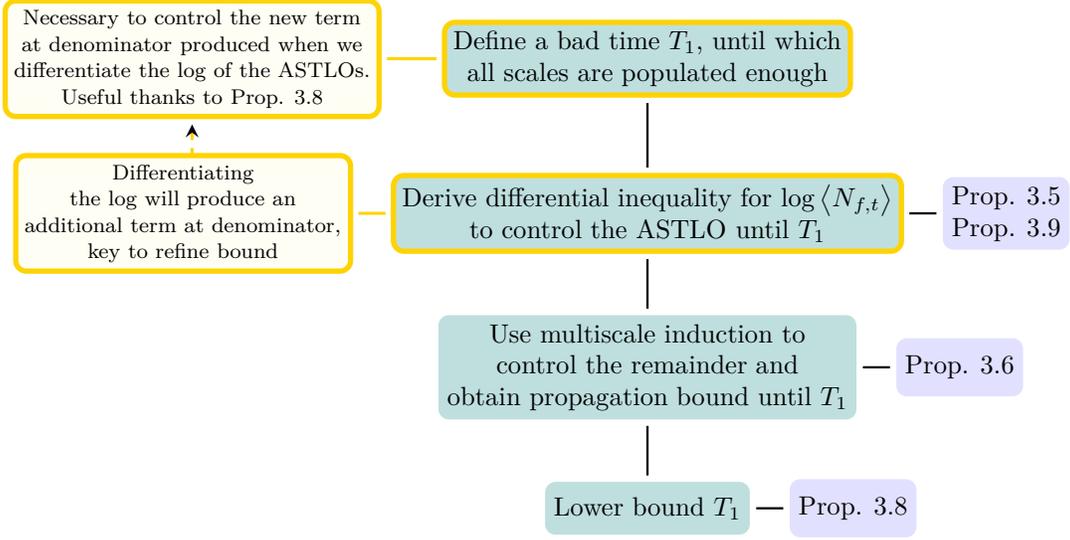
\begin{figure}[h]
			\centering
			\begin{tikzpicture}[
				node distance=1.5cm and 3cm,  line width=0.3mm,
				box/.style = {rounded corners, fill=blue!12, align=center, minimum height=2.2em},
				box2/.style = {rounded corners, fill=green!25, align=center, minimum height=2em},
				arrow/.style = {thick, -Stealth},
				every label/.append style = {font=\scriptsize, align=center},
				box3/.style = {rounded corners, fill=purple!20, align=center, minimum height=2.3em, font=\large}  , 
				highlight/.style={ line width=2pt, draw=yellow!85!red},
				box4/.style={highlight,rounded corners, fill=yellow!5, align=center, font=\Small},
				arrow/.style={-stealth,shorten <=2pt, shorten >=2pt},
				shorten <=2pt, shorten >=2pt,
				myfit/.style = {rounded corners, draw, dashed},
				scale=0.7
				]
				
				\node [box2, highlight] (badtime) {Define a bad time $T_1$, until which\\ all scales are populated enough};
				\node [box4, font=\Small,left=8mm of badtime] (new 2) {Necessary to control the new term\\ at denominator  produced when we\\differentiate the log of the ASTLOs.\\
					Useful thanks to Prop. \ref{Prop 4 lb and ub 1} };
				
				%\node [box2, below=0.8cm of badtime] (step 1) { Prove 1) for $p=1$ to control time\\ evolution of the ASTLO until $T_1$};
				\node [box2,highlight, below=1cm of badtime] (step 1) { Derive differential inequality for $\log\br{ N_{f,t}}$\\ to control the ASTLO until $T_1$};
				\node[box, right=.5cm of step 1] (step1 1) {Prop. \ref{Prop 1 ub 1} \\Prop. \ref{prop ind basis}};
				\node[box4, left=5mm of step 1, font=\Small] (new 1) {Differentiating\\the log will produce an\\additional term at denominator,\\ key to refine bound};

				\node [box2, below=0.8 cm of step 1] (ms) {Use multiscale induction to\\ control the remainder and \\obtain propagation bound until $T_1$};
				\node [box, right=0.5cm of ms] (ms 1) {Prop. \ref{cor remainder}};
				
				\node [box2, below=0.8cm of ms] (lb t) {Lower bound $T_1$};
				\node [box, right=0.5cm of lb t ] (lb t1) {Prop. \ref{Prop 4 lb and ub 1}};
				\draw [arrow, dashed, draw=yellow!85!red, line width=0.4mm] ([yshift=-8mm]new 2.south)--(new 2);
				\draw [ draw=yellow!85!red, line width=0.4mm] (new 2)--(badtime);
				\draw [ draw=yellow!85!red, line width=0.4mm] (new 1)--(step 1);
				\draw (badtime)--(step 1);
				\draw (step 1)--(ms);
				\draw (ms)--(lb t);
				
				\draw (lb t1)--(lb t);
				\draw (ms)--(ms 1);
				\draw (step 1)--(step1 1);
				%\node[below=5mm of lb t ]{ };
				\node[above=5mm of badtime] { };

			\end{tikzpicture}
			\caption{Flowchart illustrating the proof of \thmref{thm4} for $p=1$ and small $R/r$. The yellow boxes correspond to new ideas that go beyond our prior multiscale induction scheme \cite{lemm2023microscopic}.}
			\label{flow chart 2}
		\end{figure}

		To have a better understanding on how this section fits in with the complete proof the reader can refer to \figref{flow chart 1}.

		\newcommand{\ctE}{\tilde \cE}	
		
		We start with some notation. 
		Let $R>r\ge1$. Set \begin{align}
			\notag
			a:=R/r,\qquad b:=\log_{R/r}r,
		\end{align} or, equivalently,
		\begin{align}\notag
			R=a^{b+1},\qquad r=a^b.
		\end{align}
		Notice that $b\ge0$ for $R>r\ge 1$. For any $l\in\N_0$, we define
		\begin{align}
			\s_l\equiv(R_l,r_l):=&(a^{b+l+1},a^{b+l})\label{silDef},\\ 
			s_l:=&s_{\si_l}= \frac{2(R_l-r_l)}{3v},\label{slDef}
		\end{align}
		where \eqref{slDef} is according to \eqref{sSiDef}.
		Recall definition \eqref{classE} of the function class $\cE$ . From this section onwards, we focus on those function in $\cE$, such that   
		\begin{align}\label{Cf cond}
			C_{f}\le 2 \tilde C:=2\inf_{g\in\cE}C_g,
		\end{align}
		where the constant $C_g$, is the same appearing in \eqref{fExp}, and is defined as in \eqref{Cf def}. We name the set of function in $\cE$ satisfying \eqref{Cf cond}, $\tilde \cE$. Condition \eqref{Cf cond} is reasonable since $C_g$ is a positive and finite constant.
		Recall definitions \eqref{f-}--\eqref{f+}.
		For any $l\in\N_0$ and $f\in\ctE$ we  define
		\begin{align}
			f_\pm(\abs{x}_{t,l}):=&f_\pm(\abs{x}_{t,\si_l})\notag,\\
			\Nf[]{f_{\pm}}{t}{l}:=&\sum_{x\in\Lambda}f_\pm(\abs{x}_{t,l})n_x,\notag\\
			\cG:=& \left\{f_+,f_-\right\},\notag
		\end{align}
		and the quantity
		\begin{align}
			&T_1:=\inf \left\{t>0 \,:\, \exists\, \tilde l\in\N_0 \text{ s.t. }\min_{u\in\cG}\frac{\br{\Nf{u}{t}{\tilde l}}_t}{\l_1 {r^d_{\tilde l}}}\le \frac{1}{\me} \right\}.\label{T1Def}
		\end{align}
		
		Finally, let
		\begin{align}\label{t0}
			{			t_l:=\min\left(\frac{R_l-r_l}{3v},T_1\right)=\min\left(\frac{(a-1)a^{b+l}}{3v},T_1\right).}
		\end{align}
		{The definition of \eqref{T1Def} essentially says that until time $t_l$, particle distribution at every scale is not too dilute. More precisely, definitions \eqref{T1Def}--\eqref{t0} imply
			\begin{align}
				\label{tlprop}
				\min_{u\in\cG}\frac{\br{\Nf{u}{t}{l}}_t}{\l_1 {r^d_{ l}}}\ge \frac{1}{\me} \qquad (l\in\mathbb{N}_0,t\le t_l).
			\end{align}
		}

		{In the remainder of this section, we prove bounds on particle propagation up to time $t_l$.}

		For $\al>d+1$, define
		\begin{align}
			\label{epsDef}
			\eps:=\frac{\al-d-1}{2}>0.
		\end{align}
		We have the following result:
		
		\begin{proposition}\label{Prop 1 ub 1}
			{Let $\al>d+1$ and let \eqref{CD} hold for $p=1$ and some $0<\l_1<\l_2<\infty$.} 
			Then, for any ${v} > \kappa$, $f\in\ctE$, 
			{$\delta_0\in(0,1)$, $\mu>0$, and $\eps>0$ from \eqref{epsDef}, there exists \begin{align}
					\label{rhoDef}
					{\rho=\rho(v,\delta_0,\mu,\eps)}>0,
				\end{align} 
				such that for all $l\in\mathbb{N}_0$, {$R>r\ge\rho$} with $R-r>\delta_0r$,} and $t\le t_l$
			\begin{align}
				\pm\br{\Nf{f_\mp}{t}{l}}_t&\le\pm\br{\Nf{f_\mp}{0}{l}}_0 {\exp \left\{\pm\frac{  t}{\mu \l_1s_{l}r_l^d } \sup_{0\le \ta\le t_l}\cR_{1,l}(\ta)\right\}},\label{bound integral final f}
			\end{align}
			where we set
			\begin{align}
				&\cR_{1,l}(\ta):={\br{N_{B_{r_l}}}_\ta}+\sum_{x\in B_{{R_l}}}\sum_{\substack{y\in\Lambda}}\abs{J_{xy}}|x-y|^{1+\eps}{\br{ n_y}_\ta }.\label{P1Def}
			\end{align} 
			
		\end{proposition}
		
		\begin{proof}
			1.	We  prove \eqref{bound integral final f} by adapting the strategy used in \cite[Prop.~4.2]{lemm2023microscopic}. {Take any $f\in\tilde \cE$.} We start by differentiating  the logarithm of the ASTLO with respect to time.
			\begin{align}\label{differentiatin log with f}
				\frac{d}{dt}\log\br{\Nf{f_-}{t}{l}}_t=\frac{1}{\br{\Nf{f_-}{t}{l}}_t}\left(-\frac{v'}{s_l}\br{\Nf{f_-'}{t}{l}}_t+\br{\left[iH,\Nf{f_-}{t}{l}\right]}_t\right),
			\end{align}
			{where, recall, $v'>\kappa$ (see \eqref{v'def})}.

			Thanks to \lemref{new taylor}, we can follow the proof of  \cite[Prop.~4.2]{lemm2023microscopic}   to bound the second term in the r.h.s. of \eqref{differentiatin log with f}. This way we arrive at
			\begin{align}\label{bound of derivative of log f, after all steps}
				\frac{d}{dt}\log\br{\Nf{f_-}{t}{l}}_t\le
				&\frac{\kappa-v'}{ s_l}\frac{\br{ N_{f_-',tl}}_t}{\br{\Nf{f_-}{t}{l}}_t}  \notag\\
				&+\frac{ C_{f}}{{s_l^{1+\eps}}} \left(\frac{\br{N_{B_{{R_l}}}}_t}{\br{\Nf{f_-}{t}{l}}_t}+\sum_{x\in B_{{R_l}-\omega s/2}}\sum_{y\in\Lambda}\abs{J_{xy}}|x-y|^{1+\eps}\frac{\br{ n_y}_t }{\br{\Nf{f_-}{t}{l}}_t}\right).
			\end{align}
			We note that, since $\kappa\le v'$, the first summand on the r.h.s. of \eqref{bound of derivative of log f, after all steps} is negative, so we can drop it. This way, applying also \eqref{Cf cond}, we obtain
			\begin{align}\label{5 o}
				\frac{d}{dt}\log\br{\Nf{f_-}{t}{l}}_t\le
				\frac{ 2\tilde C}{{s_l^{1+\eps}}} \left(\frac{\br{N_{B_{R_l}}}_t}{\br{\Nf{f_-}{t}{l}}_t}+\sum_{x\in B_{{R_l}}}\sum_{y\in\Lambda}\abs{J_{xy}}|x-y|^{1+\eps}\frac{\br{ n_y}_t }{\br{\Nf{f_-}{t}{l}}_t}\right).
			\end{align}
			Integrating both sides of \eqref{5 o}  up to time $t\le t_l$ and recalling the definition of $t_l$ (see \eqref{tlprop}), it follows that
			\begin{align}\label{right after integral log f}
				\log\br{\Nf{f_-}{t}{l}}_{t}-\log\br{\Nf{f_-}{0}{l}}_0
				\le &\frac{ 2\tilde C}{s_l^{1+\eps}} \int_0^{t}\left( \frac{\br{N_{B_{{R_l}}}}_\ta}{\me^{-1}r^d_{ l}\l_1}+\sum_{x\in B_{{R_l}}}\sum_{y\in\Lambda}\abs{J_{xy}}|x-y|^{1+\eps}\frac{\br{ n_y}_\ta }{\me^{-1}r^d_{ l}\l_1}\right)d\ta\notag\\
				\le &  \frac{2\me\tilde Ct}{\l_1{s_l^{1+\eps}}r_l^d} \sup_{0\le \ta\le t_l}\cR_{1,l}(\ta),		\end{align}
			where $\cR_{1,l}$ is defined in \eqref{P1Def}.
			
			Step 2.	For any $r\ge \rho$ with
			\begin{align}
				\label{Cdef2}
				\rho=\rho(v,\delta_0,\mu,\eps):=\frac{3v}{2}\frac{1}{\delta_0}[2\me\tilde C\mu]^{1/\eps},
			\end{align}
			it is easily verified using definitions \eqref{silDef} and \eqref{slDef}, and recalling $R-r>\delta_0 r$, that
			\begin{align}
				\notag
				s_l^\eps\ge 2\me\tilde C\mu\qquad(l=0,1,\ldots).
			\end{align}
			This, together with estimate \eqref{right after integral log f}, implies
			\begin{align}
				\label{638}
				\log\br{\Nf{ f_-}{t}{l}}_{t}-\log\br{\Nf{ f_-}{0}{l}}_0\le \frac{  t}{\mu\l_1s_lr_l^d} \sup_{0\le \ta\le t_l}\cR_{1,l}(\ta)\qquad (t\le t_l).  
			\end{align}
			Finally, exponentiating  inequality \eqref{638} yields \eqref{bound integral final f}. 
			
			Step 3. It remains to show \eqref{bound integral final f} for $f_+$. To this end, we follow a similar procedure as in step 1. The only difference is that when differentiating $\log\br{\Nf{f_+}{t}{l}}_t$ with respect to time, there is a change of sign, in fact the following holds.
			\begin{align}\label{1 1 1}
				- \frac{d}{dt}\log\br{\Nf{f_+}{t}{l}}_t=\frac{1}{\br{\Nf{f_+}{t}{l}}_t}\left(-\frac{v}{s_l}\br{\Nf{f_+'}{t}{l}}_t+\br{-\left[iH,\Nf{f_+}{t}{l}\right]}_t\right).
			\end{align}
			Recall the definition of $\tilde \cE$, then for any $f\in\tilde \cE$ it holds that 
			\begin{align}\notag
				\text{if } \abs{x}\ge\frac{R_l+2r_l}{3}+v't-\frac{s\omega}{2}\quad \text{then} \quad f_+\del{\abs{x}_{t,l}}=0.
			\end{align}
			Therefore,
			\begin{align}\notag
				\abs{f_+(\abs{x}_{tl})-f_+(\abs{y}_{tl})}\le \abs{f_+(\abs{x}_{tl})-f_+(\abs{y}_{tl})}\del{\chi_{\tilde B}(x)+\chi_{\tilde B}(y)},
			\end{align}
			where
			\begin{align}\notag
				\tilde B\coloneqq B_{(R_l+2r_l)/3+v't-s\omega/2}.
			\end{align}
			Recalling the definition of $s_l$ \eqref{slDef} and $\omega$ \eqref{omega def}, since $v'<v$, one can show that, for every $t\le t_l< s_l$, $\tilde B\subset B_{R_l}$, which implies
			\begin{align}\label{where if f+}
				\abs{f_+(\abs{x}_{tl})-f_+(\abs{y}_{tl})}\le \abs{f_+(\abs{x}_{tl})-f_+(\abs{y}_{tl})}\del{\chi_{ B_{R_l}}(x)+\chi_{ B_{R_l}}(y)}.
			\end{align}
			Inequality \eqref{where if f+} plays the same role as \cite[ex~4.35]{lemm2023microscopic}, and it allows us to estimate the commutator in \eqref{1 1 1} as before, 
			\begin{align}\label{after est comm lb 1}
				-\frac{d}{dt}\log\br{\Nf{f_+}{t}{l}}_t&\le
				\frac{\kappa-v' }{s_l}\frac{\br{ N_{f_+',tl}}_t}{\br{\Nf{f_+}{t}{l}}_t}  \notag\\
				&+\frac{ C_{f}}{{s_l^{1+\eps}}} \left(\frac{\br{N_{B_{{R_l}}}}_t}{\br{\Nf{f_+}{t}{l}}_t}+\sum_{x\in B_{{R_l}}}\sum_{\substack{y\in\Lambda}}\abs{J_{xy}}|x-y|^{1+\eps}\frac{\br{ n_y}_t }{\br{\Nf{f_+}{t}{l}}_t}\right).  
			\end{align}
			Again, we drop the negative term on the r.h.s. of \eqref{after est comm lb 1} and we apply \eqref{Cf cond} to obtain
			\begin{align}\label{6 o}
				-\frac{d}{dt}\log\br{\Nf{f_+}{t}{l}}_t\le\frac{ 2\tilde C}{{s_l^{1+\eps}}} \left(\frac{\br{N_{B_{{R_l}}}}_t}{\br{\Nf{f_+}{t}{l}}_t}+\sum_{x\in B_{{R_l}}}\sum_{\substack{y\in\Lambda}}\abs{J_{xy}}|x-y|^{1+\eps}\frac{\br{ n_y}_t }{\br{\Nf{f_+}{t}{l}}_t}\right).  
			\end{align}
			Integrating both sides  of \eqref{6 o} up until time $t\le t_l$ and using  \eqref{tlprop}, we obtain
			\begin{align}
				\log\frac{\br{\Nf{f_+}{0}{l}}_{0}}{\br{\Nf{f_+}{t}{l}}_{t}}
				\le\frac{ 2\me\tilde Ct}{\l_1 s_l^{1+\eps}{r^d_l}} \sup_{0\le \ta\le t_l}\cR_{1,l}(\ta).\notag
			\end{align}
			For every $r\ge\rho$, with $\rho$ as in \eqref{Cdef2}, we obtain
			\begin{align}\label{befor getting rid of g}
				\log\frac{\br{\Nf{f_+}{0}{l}}_{0}}{\br{\Nf{f_+}{t}{l}}_{t}}
				\le\frac{ t}{\l_1\mu s_lr_l^d} \sup_{0\le \ta\le t_l}\cR_{1,l}(\ta), 
			\end{align}
			We conclude the proof exponentiating both sides of \eqref{befor getting rid of g} to obtain \eqref{bound integral final f}.
		\end{proof}

		Next, we prove an upper bound for the remainder term $\cR_{1,l}$ defined in \eqref{P1Def}.  
		Let $\eps$ be as defined in \eqref{epsDef}. Then, there exists a constant ${C>0}$ independent of $\Lambda$ such that
		\begin{align}\label{k2Cond}
			\kappa_\eps:=\sup_{x\in\Lam}\sum_{y\in\Lam} {\abs{J_{xy}}} \abs{x-y}^{1+\eps}\le C.
		\end{align}
		We have the following result:
		\begin{proposition}[Remainder estimate for $p=1$]\label{cor remainder}
			Let the assumptions of \propref{Prop 1 ub 1} hold. Then, for any $v\ge \kappa$ $\delta_0\in(0,1)$, 	
			there exists 			\begin{align}
				\label{CPdef}
				C_{\cR,1}=&C_{\cR,1}(\al,d,C_J,\l_2,\delta_0)>0,
			\end{align}
			such that for 
			\begin{align}
				{\rho=\rho(\l_1,v,\delta_0,C_{\cR,1},\eps)}>0,\notag
			\end{align} (see definition in \eqref{rhoDef}),
			all $l\in\mathbb{N}_0$, and {$R>r\ge\rho$} with $R-r>\delta_0r$ and $a\equiv R/r$ satisfying \eqref{aCond}, the following holds
			\begin{align}\label{bound rem 1}
				{r_l^{-d}}\sup_{0\le\ta\le t_l}\cR_{1,l}(\ta)\le C_{\cR,1}.
			\end{align}
		\end{proposition}

		\begin{proof}

			We prove \eqref{bound rem 1} by a downwards induction on the scale parameter $l$. 
			
			1.		\textbf{Base case.} Fix a large $L>0$ s.th.~the domain $\Lam\subset B_L$.  We prove that there exists a large integer $M = M(L) > 0$ s.th. \eqref{bound rem 1}
			holds for all $ l \ge M$.

			Recalling definition \eqref{P1Def} and using \eqref{k2Cond}, the upper bound in density assumption \eqref{CD}, and the fact that the total particle number $N_\Lam$ is conserved,
			we  generously upper bound the remainder term as
			\begin{align}
				r_l^{-d}\cR_{1,l}(\ta)\le&r_l^{-d}{\br{N_\Lam}_\ta}+r_l^{-d}\sum_{x\in \Lambda}\sum_{y\in\Lambda}\abs{J_{xy}}|x-y|^{1+\eps}{\br{ n_y}_\ta }\notag  \\
				\le& r_l^{-d}{\br{N_\Lam}_0}+\kappa_\eps r_l^{-d}{\br{N_\Lam}_0}\notag\\
				\le& {\l_2}(1+\kappa_\eps)\frac{L^d}{a^{d(l+b)}}.\label{ind basis }
			\end{align} 
			For all $l\ge M$ with
			\begin{align}
				\notag
				M:=\log_aL,
			\end{align}
			the desired estimate \eqref{bound rem 1} now follows from \eqref{ind basis }. This proves the induction base.
			
			2.	 \textbf{Induction step.}   Now, {\textit{we assume \eqref{bound rem 1} holds for all $l+1, l+2,\ldots$}}, and we show it for $l$. 
			
			{We first apply \propref{Prop 1 ub 1} with \begin{align}
					\label{muChoice}
					\mu=\l_1^{-1}C_{\cR,1}.
				\end{align} Inserting the choice \eqref{muChoice} into  \eqref{bound integral final f}, and using the geometric properties of the ASTLO proved in \lemref{ASTLO and N}, we find that the induction hypothesis implies the following estimate {\textit{for all $j\ge l+1$ and $r\ge\rho$},}
				\begin{align}
					\sup_{0<t\le t_{j}} {\br{N_{B_{r_{j}}}}_{t}}&\le {\br{N_{B_{R_{j}}}}_0}e^{{ t_{j}}{s_{j}^{-1}}}.\label{final UB moment1}
					% \inf_{0<t\le t_{j}} {\br{N_{B_{R_{j}}}}_{t}}&\ge {\br{N_{B_{r_{j}}}}_0}e^{{ -t_{j}}{s_{j}^{-1}}} . \label{final LB moment1}
				\end{align}
				Below we will repeatedly use this consequence to complete the induction step.}

			Indeed, recalling definitions \eqref{P1Def} for $\cR_{1,l}$ and \eqref{silDef} for $r_l$ and $R_l$, we write 
			\begin{align}\notag
				r_l^{-d}\sup_{\ta<t_l} \cR_{1,l}(\ta)&\le\underbrace{\vphantom{\sum_{B_{{a^{l+b+1}}}}}\sup_{\ta<t_l}\frac{\br{N_{B_{{a^{l+b+1}}}}}_\ta}{a^{d(l+b)}}}_{\Mr{I}}+\underbrace{\sup_{\ta<t_l}\sum_{x\in B_{{a^{l+b+1}}}}\sum_{\substack{y\in\Lambda}}\abs{J_{xy}}|x-y|^{1+\eps}\frac{\br{ n_y}_\ta }{a^{d(l+b)}}}_{\Mr{II}}.
			\end{align}
			
			2.1.	To bound $\Mr{I}$, we claim that there exists \begin{align}
				\notag
				C_1=C_1		(d,\l_2)>0,
			\end{align} such that 
			\begin{align}
				\label{mBdd}
				\sup_{\ta<t_l}\frac{\br{N_{B_{a^{j+b+1}}}}_\ta}{{a^{d(j+b)}}}\le C_1,\qquad 	j\ge l.
			\end{align}
			Indeed, suppose \eqref{mBdd} holds. Then applying it with $j=l$ yields 
			\begin{align}\label{8 o}
				\Mr{I}&\equiv \sup_{\ta<t_l}\frac{\br{N_{B_{{a^{l+b+1}}}}}_\ta}{a^{d(l+b)}}\le C_1.
			\end{align}			
			This gives the desired upper bound for $\Mr{I}$.
			
			Now we prove \eqref{mBdd}. Thanks to upper bound \eqref{final UB moment1}, which is valid for $j+1$ when $j\ge l$ under the induction hypothesis, there holds
			\begin{align}
				\label{756}
				\sup_{\ta<t_l}\frac{\br{N_{B_{{a^{j+b+1}}}}}_\ta}{a^{d(j+b)}}\le\frac{\br{N_{B_{a^{j+b+2}}}}_0}{ a^{d m}}\me^{{  t_{j+1}}/{s_{j+1}}}. 
			\end{align}
			Applying the upper bound from the controlled density assumption  \eqref{CD}, we find that the r.h.s.~of \eqref{756} is bounded, for all $t_{l+1}\le s_{l+1}$, as
			\begin{align}
				\frac{\br{N_{B_{a^{j+b+2}}}}_0}{a^{d(j+b) }}\me^{{  t_{l+1}}/{s_{l+1}}}&\le\frac{{a^{d(j+b+2)}}\l_2}{a^{d(j+b)}}\me^{{  t_{l+1}}/{s_{l+1}}}\le  a^{2d} \l_2\me.\notag
			\end{align}
			{Now we use the upper bound on $a$ from assumption \eqref{aCond},} which gives 
			\begin{align}\notag
				\frac{\br{N_{B_{a^{j+b+2}}}}_0}{a^{d(j+b) }}\me^{{  t_{l+1}}/{s_{l+1}}}\le4^{2d}\cdot \l_2\me=: C_1\qquad(j\ge l).
			\end{align}
			This proves claim \eqref{mBdd}.

			2.2.		Next, 	we deal with term II. We proceed in the same way as in the proof of \cite[Prop.~4.4]{lemm2023microscopic}, by splitting the sum on $y$ over two regions, $ B_{R_{l+1}}=:B$ and $ B_{R_{l+1}}^c$.

			\begin{align}\label{term II}
				\Mr{II}&=\sup_{\ta<t_l}\sum_{x\in B}\sum_{y\in\Lambda}\abs{J_{xy}}|x-y|^{1+\eps}\frac{\br{ n_y}_\ta }{a^{d(l+b)}}\notag\\
				&\le\sup_{\ta<t_l}\sum_{x\in B}\sum_{y\in B}\abs{J_{xy}}|x-y|^{1+\eps}\frac{\br{ n_y}_\ta }{a^{d(l+b)}}\notag\\
				&\quad+\sup_{\ta<t_l}\sum_{x\in B}\sum_{m>l+b+1}\sum_{a^{m}\le y<a^{m+1}}\abs{J_{xy}}|x-y|^{1+\eps}\frac{\br{ n_y}_\ta }{a^{d(l+b)}}.
			\end{align}
			
			We bound the two terms in the r.h.s. of \eqref{term II}	separately. Recalling the definition of $\kappa_\eps$ in \eqref{k2Cond} and \eqref{8 o}, we bound the  first term on the  r.h.s.~as
			\begin{align}\label{to last rem}
				\sup_{\ta<t_l} \sum_{x\in B}\sum_{y\in B}\abs{J_{xy}}|x-y|^{1+\eps}\frac{\br{ n_y}_\ta }{a^{d(l+b)}}
				&\le \del{\sup_{x\in\Lam}\sum_{y\in\Lam} \abs{J_{xy}}\abs{x-y}^{1+\eps}}\del{\sup_{\ta<t_l}\sum_{\abs{y}\le a^{b+l+1}}\frac{\br{ n_y}_\ta }{a^{d(l+b)}}}\notag \\
				&= \k_\eps\Mr{I}\le \kappa_\eps C_1.
			\end{align}

			To bound the second term in \eqref{term II},  we follow the strategy in the proof of \cite[eq.~(4.20)]{lemm2023microscopic}.		We differentiate two cases: 
			\begin{itemize}
				\item[1)] The set $\Set{y\in\Lambda\;:\;d(y,B)=1}$ is contained in $B^c_{R_{l+1}}$;
				\item[2)] The set $\Set{y\in\Lambda\;;\;d(y,B)=1}$ lies in between $B_{R_l}$ and $B^c_{R_{l+1}}$.
			\end{itemize} 
			
			2.2.1. In case 1) we set\begin{align}
				\label{gaDef}
				\gamma:=\alpha-1-\eps.
			\end{align} 
			It follows from definition \eqref{gaDef} and \eqref{epsDef} that
			\begin{align}
				\label{gaEst}
				\gamma>d+\eps.
			\end{align}
			We compute
			\begin{align}\label{control 2 2'}
				\Mr{II}'&:=	\sup_{\ta<t_l} \sum_{x\in B}\sum_{m>l+b+1}\sum_{a^{m}\le y<a^{m+1}}\abs{J_{xy}}|x-y|^{1+\eps}\frac{\br{ n_y}_\ta }{a^{d(l+b)}}\notag\\
				&\le C_J	\sup_{\ta<t_l} \sum_{x\in B}\sum_{m>l+b+1}\sum_{a^{m}\le y<a^{m+1}}|x-y|^{-\g}\frac{\br{ n_y}_\ta }{a^{d(l+b)}}\notag\\
				&\le	 C_JC_d a^{d(l+b)}(a-1)^{-\gamma} \sum_{m> l+b+1}a^{-(m-1)\g}\frac{\sup_{\ta<t_l}\br{N_{B_{a^{m+1}}}}_\ta}{{a^{d(l+b)}}}\notag\\
				&\le  C_JC_d \del{\frac{4}{\delta_0}}^{\g} \sum_{m> l+b+1}a^{-m(\g-d)}\frac{\sup_{\ta<t_l}\br{N_{B_{a^{m+1}}}}_\ta}{{a^{dm}}} ,
			\end{align}
			where in the last inequality we used the assumption $\delta_0+1<a<4$.
			Applying \eqref{mBdd} for $m>l+b+1$ in   line \eqref{control 2 2'}, we arrive at
			\begin{align}
				\label{control 2 2}
				\Mr{II}'\le C_JC_dC_1\sum_{m> l+b+1}a^{-m(\g-d)}.
			\end{align}
			The sum in \eqref{control 2 2} converges since $a> \delta_0+1\ge 1$ and $\g>d$ (see \eqref{gaEst}). 
			Thus, we conclude
			\begin{align}\label{last for rem}
				\Mr{II}'\le C_JC_1C_{d,\al,\delta_0}.
			\end{align}

			2.2.2. In case 2) we write 
			\begin{align}
				\Mr{II}&=\sup_{\ta<t_l}\sum_{x\in B}\sum_{y\in\Lambda}\abs{J_{xy}}|x-y|^{1+\eps}\frac{\br{ n_y}_\ta }{a^{d(l+b)}}\notag\\
				&\le\sup_{\ta<t_l}\sum_{x\in B}\sum_{y\in B}\abs{J_{xy}}|x-y|^{1+\eps}\frac{\br{ n_y}_\ta }{a^{d(l+b)}}\notag\\
				&\quad+\sup_{\ta<t_l}\sum_{x\in B}\sum_{m>l+b+1}\sum_{a^{m}\le y<a^{m+1}}\abs{J_{xy}}|x-y|^{1+\eps}\frac{\br{ n_y}_\ta }{a^{d(l+b)}}\notag\\
				&\quad+\sup_{\ta<t_l}\sum_{x\in B}\sum_{\substack{y\in B^c\\ d(B,y)=1}}\abs{J_{xy}}|x-y|^{1+\eps}\frac{\br{ n_y}_\ta }{a^{d(l+b)}}.
			\end{align}
			We additionally have the contribution from the lattice site laying in between $B$ and $B^c_{R_{l+1}}$. This is the only term we need to control as we already controlled the first and the second summands in step 2.2.1. To this end, we compute
			\begin{align}\label{II 3}
				\Mr{II}'':=\sup_{\ta<t_l}&\sum_{x\in B}\sum_{\substack{y\in B^c\\ d(B,y)=1}}\abs{J_{xy}}|x-y|^{1+\eps}\frac{\br{ n_y}_\ta }{a^{d(l+b)}}\\
				\le&\kappa_\eps\sup_{\ta<t_l}\sum_{\substack{y\in \Lambda\\d(y,B)=1}}\frac{\br{ n_y}_\ta }{a^{d(l+b)}}\\
				\le&\kappa_\eps\sup_{\ta<t_l}\frac{\br{ B_{R_l+1}}_\ta }{a^{d(l+b)}}.
			\end{align}
			Notice that since we are assuming $\Set{y\in\Lambda\;;\;d(y,B)=1}$ lies in between $B_{R_l}$ and $B^c_{R_{l+1}}$, it follows that $R_l+1<R_{l+1}$. Then, rewriting $R_l=a^{b+l+1}$ and applying \eqref{mBdd} lead to
			\begin{align}\label{contr new term rem}
				\Mr{II}'' \le\kappa_\eps\frac{\br{ B_{a^{b+l+1}+1}}_\ta }{a^{d(l+b)}}    \le \kappa_\eps\frac{\br{ B_{a^{b+l+2}}}_\ta }{a^{d(l+b)}}  \le \kappa_\eps C_1 a^{d}  \le \kappa_\eps  C_1 4^d   .
			\end{align}
			
			2.3. Finally, plugging \eqref{to last rem}, \eqref{last for rem}, and\eqref{contr new term rem} back to   \eqref{8 o}, we arrive at
			\begin{align}
				\label{655}
				r_l^{-d}\sup_{\ta<t_l} \cR_{1,l}(\ta)\le C_{\cR,1},\qquad C_{\cR,1}:=(1+\kappa_\eps +C_JC_{d,\al,\delta_0}+\kappa_\eps  4^d)C_1.
			\end{align}
			This completes the induction step and \propref{cor remainder} is proved.
		\end{proof}
		\begin{corollary}\label{cor7.6}
			Let the assumptions of \propref{Prop 1 ub 1} hold. Then, for any $v\ge \kappa$ and $\delta_0\in(0,1)$, the following estimates hold for all $l\in\mathbb{N}_0$  and {$R>r\ge\rho$} with 
			\begin{align}
				\rho=\rho\del{v,\delta_0, \al, d, C_J,\l_2,\l_1}, \label{rho at the end}
			\end{align}
			$R-r>\delta_0r$ and $a\equiv R/r$ satisfying \eqref{aCond}:
			\begin{align}
				{\br{N_{B_{r_l}}}_{t}}&\le {\br{N_{B_{R_l}}}_0}e^{{ t}{s_l^{-1}}}\notag\\
				{\br{N_{B_{R_l}}}_{t}}&\ge {\br{N_{B_{r_l}}}_0}e^{{ -t}{s_l^{-1}}}, \label{760}
			\end{align}
			for any $t\le t_l$. 
		\end{corollary}
		\begin{proof}
			We apply \propref{Prop 1 ub 1} with $\mu:=\l_1^{-1}C_{\cR,1}$ with $C_{\cR,1}$ from \eqref{CPdef}, for every $t\le t_l$. Inequality \eqref{bound rem 1} then implies
			\begin{align}
				{\br{\Nf{ f_-}{t}{l}}_{t}}&\le {\br{\Nf{ f_-}{0}{l}}_0}e^{{ t}{s_l^{-1}}}\label{still f},\\
				{\br{\Nf{ f_+}{t}{l}}_{t}}&\ge {\br{\Nf{ f_+}{0}{l}}_0}e^{{- t}{s_l^{-1}}}\label{still f 2}.
			\end{align}
			We conclude by applying \lemref{ASTLO and N} to \eqref{still f}-- \eqref{still f 2}, to get rid of the $ f$ dependency.
			
		\end{proof}

		\subsection{Proof of Lower Bound for First Bad Time \texorpdfstring{$T_1$}{T1}}\label{secTest}

		\begin{proposition}\label{Prop 4 lb and ub 1}
			Let the assumptions of \propref{Prop 1 ub 1} hold. Then, for any $v>\kappa$ and $f\in\ctE$, $T_1$ defined in \eqref{T1Def}, and $R>r\ge\rho$ as in \corref{cor7.6},
			\begin{align}
				\label{T1Est}
				T_1\ge (R-r)/3v.
			\end{align}  
		\end{proposition}
		\begin{remark}
			Notice that \propref{Prop 4 lb and ub 1} implies $t_0=(R-r)/3v$. This follows from the definition $t_l:=\min\Set{(R_l-r_l)/3v,T_1}$ and the fact that $R_l-r_l=(a-1)a^{b+l}>(a-1)a^b=R-r$.
		\end{remark}
		\begin{proof}
			
			Assume, for the sake of contradiction to \eqref{T1Est}, that
			\begin{align}
				\label{contra}
				\text{		 	 $T_1<(R-r)/3v$. }
			\end{align}
			Consider a function $ f\in\ctE$ and the set $\cG=\Set{f_-, f_+}$.
			By the definition of $T_1$ in \eqref{T1Def}, it follows from \eqref{contra} that there exists $\Tilde{l}\in\N_0$ such that at time $T_1$,
			\begin{align}\label{contr1}
				e^{-1}\ge \min_{u\in\cG}\frac{\br{\Nf{u}{T_1}{\Tilde{l}}}_{T_1}}{\l_1 \abs{B_{a^{b+\Tilde{l}}}}}= \frac{\br{\Nf{f_+}{T_1}{\Tilde{l}}}_{T_1}}{\l_1 \abs{B_{r_{\Tilde{l}}}}}.
			\end{align}
			The equality is due to \lemref{lem7.3}, which is indeed applicable since $T_1< {R-r}/{3v'}$ under assumption \eqref{contra} (see \eqref{t0}).
			
			Applying lower bound \eqref{760} to the r.h.s. of  \eqref{contr1}, we obtain
			\begin{align}
				e^{-1}\ge\frac{{\br{N_{B_{r_{\Tilde{l}}}}}_{0}}}{\l_1 \abs{B_{r_{\Tilde{l}}}}} e^{-{  T_1}/{s_{\Tilde{l}}}}.\notag
			\end{align}
			The assumption of controlled density \eqref{CD}, {which is valid for $r_{\tilde l}\ge r\ge \rho$,} leads to
			\begin{align}\notag
				e^{-1}\ge e^{-{  T_1}/{s_{\Tilde{l}} }},
			\end{align}
			or, equivalently,
			\begin{align}\label{v dependency in t}
				T_1\ge     s_{\Tilde{l}}.
			\end{align}
			By the definition of $s_l$ in \eqref{slDef}, we conclude from \eqref{v dependency in t} that 
			\begin{align}
				T_1\ge\frac{R_{\tilde l}-r_{\tilde l}}{3v}\ge\frac{R-r}{3v}.\notag 
			\end{align}
			This contradicts the assumption \eqref{contra}, and the proposition is proved.
		\end{proof}
		
		\subsection{Completing the proof of \thmref{thm4} for ${p=1}$}\label{sec proof thm4}
		
		\begin{proof}[Proof of \thmref{thm4} for ${p=1}$]
			Fix $v>12\kappa$ and let $v_0=v/4$. 
			For $R/r\le4$, \eqref{main ub 1}--\eqref{main lb 1} follow for $v_0$ by  taking $v_1:=v_0/3>\kappa$ and applying \corref{cor7.6}, lower bound \eqref{T1Est} for $T_1$, and definition for $t_l$ (see \eqref{t0}).
			This fact, together with \propref{prop7}, implies the desired estimates \eqref{main ub 1}--\eqref{main lb 1}  for $R/r>4$ and $v=4v_0$.
			
		\end{proof}
		
		\subsection{Deriving integral ASTLO estimate for \texorpdfstring{$p=1$}{p=1} from  \thmref{thm4} for \texorpdfstring{$p=1$}{p=1}}\label{int astlo for p=1}
		In this section we show how we can apply \eqref{main lb 1} and \eqref{main ub 1} for $R/r\le 4$ to derive the following integral ASTLO estimate which will constitute part of the induction basis necessary to show \thmref{thm4} for $p\ge2$.
		
		\begin{proposition}\label{prop ind basis}
			Let the assumptions of \propref{Prop 1 ub 1} hold. Then, for any ${v} > \kappa$, $f\in\ctE$,  and 
			{$\delta_0\in(0,1)$,  there exist 
				\begin{align}
					A&\equiv A^{(1)}=A^{(1)}\,(v,d,\l_1,\l_2)>0\notag\\
					B&\equiv B^{(1)}=B^{(1)}\,(v,d,\l_2/\l_1)>0\notag\\
					\rho&\equiv\rho^{(1)}=\rho^{(1)}\,(v,\delta_0,\al,d,C_J,\l_1,\l_2)>0,\notag
				\end{align} 
				such that for all $l\in\mathbb{N}_0$, {$R>r\ge\rho$} with $R-r>\delta_0r$ satisfying \eqref{aCond},} and $t\le t_0= (R-r)/3v$, the following integral ASTLO estimate holds
			\begin{align}\label{IE p=1}
				s_l^{-1}\int_0^t&\br{N_{ f'_\pm,\ta l}}_\ta \rd\ta\le A\,R_l^{d}+\frac{tB}{s_l^{1+\eps}}\sup_{0\le\ta\le t}\cR_{1,l}(\ta).
			\end{align}
		\end{proposition}
		\begin{proof}
			
			Step 1. First we show \eqref{IE p=1} for $ f_-$. Integrating \eqref{bound of derivative of log f, after all steps} between $0$ and $t\le t_0$, and rearranging the terms and applying \eqref{Cf cond} to $C_f$, we obtain
			\begin{align}\label{p=1 for in with f}
				s_l^{-1}\int_0^{t}\frac{\br{N_{f'_-,\ta l}}_\ta}{\br{N_{f_-,\ta l}}_\ta}\rd\ta &\le\frac{2\tilde CC_{v}t}{s_l^{1+\eps}}\sup_{0\le\ta\le t_0}\left(\frac{\br{N_{B_{R_l}}}_t}{\br{\Nf{f_-}{t}{l}}_t}+\sum_{x\in B_{{R_l}}}\sum_{\substack{y\in\Lambda}}\abs{J_{xy}}|x-y|^{1+\eps}\frac{\br{ n_y}_t }{\br{\Nf{f_-}{t}{l}}_t}\right)\notag\\
				&\quad+ C_v\log\frac{\br{N_{f_-,0 l}}_0}{\br{N_{f_-,t l}}_{t}}.
			\end{align}
			By the definition of $t_0$ and thanks to \propref{Prop 4 lb and ub 1} we know that the following holds 
			\begin{align}\label{8 p 1}
				\br{N_{f_{\pm,t}}}_t\ge \l_1r_l^d\me^{-1}\qquad \forall t\le\frac{R-r}{3v},\; \forall l\in\N_0.
			\end{align}
			Applying \eqref{8 p 1} and  \eqref{f and B at t=0} to \eqref{p=1 for in with f} leads to
			\begin{align}\label{p=1 for int}
				s_l^{-1}\int_0^{t}\frac{\br{N_{f'_-,\ta l}}_\ta}{\br{N_{f_-,\ta l}}_\ta}\rd\ta\le C_v\log\frac{\br{N_{B_{{R}_l}}}_0}{\l_1\me ^{-1}r_l^d}+\frac{2\tilde CC_{v}}{\l_1}\frac{t}{s_l^{1+\eps}}\;r_l^{-d}\sup_{0\le\ta\le t_0}\cR_{1,l}(\ta),
			\end{align}
			with $\cR_{1,l}(\ta)$ as in \eqref{P1Def}. To control the logarithm appearing in \eqref{p=1 for int} we apply the assumption of controlled density \eqref{CD} and  on $a$ \eqref{aCond} to obtain
			\begin{align}\label{9 p}
				\log\frac{\br{N_{B_{R_l}}}_0}{\l_1\me ^{-1}r_l^d}\le  \log \left(\frac{\me\l_2}{\l_1}\del{\frac{R}{r}}^d\right)\le \log \left(\frac{\me\l_2}{\l_1}4^d\right).
			\end{align}
			Now let us lower bound the l.h.s of \eqref{p=1 for int}. Given $\rho>0$ defined as in \eqref{rho at the end}, for any $R>r\ge\rho$ inequality \eqref{still f} yields
			\begin{align}\label{r 1}
				s_l^{-1}\int_0^{t}\frac{\br{N_{f'_-,\ta l}}_\ta}{\br{N_{f_-,\ta l}}_\ta}\rd\ta\ge \frac{s_l^{-1}}{ \me^{t{s^{-1}_l}}{\br{N_{f_-,0 l}}}_0}\int_0^{t}{\br{N_{f'_-,\ta l}}_\ta}.    
			\end{align}
			Thanks to the geometric property of $ f\in\ctE$, \eqref{f and b at t}, and the fact that $s_l>t$, from \eqref{r 1} it follows
			\begin{align}\label{r 2}
				\frac{s_l^{-1}}{ \me^{t{s^{-1}_l}}{\br{N_{f_-,0 l}}}_0}\int_0^{t}{\br{N_{f'_-,\ta l}}_\ta}\ge \frac{s_l^{-1}}{ \me{\br{N_{B_{R_l}}}_0}}\int_0^{t}{\br{N_{f'_-,\ta l}}_\ta}.
			\end{align}
			Applying \eqref{CD} to \eqref{r 2}, together with \eqref{r 1}, leads to
			\begin{align}\label{10 p}
				s_l^{-1}\int_0^{t}\frac{\br{N_{f'_-,\ta l}}_\ta}{\br{N_{f_-,\ta l}}_\ta}\rd\ta\ge \frac{s_l^{-1}}{ \me  \l_2R_l^d} \int_0^{t}{\br{N_{f'_-,\ta l}}_\ta}.
			\end{align}
			Applying \eqref{10 p} and \eqref{9 p} to \eqref{p=1 for int} leads to
			\begin{align}\notag
				s_l^{-1}\int_0^{t}{\br{N_{f'_-,\ta l}}_\ta}\le A\, R_l^d+\frac{Bt}{s_l^{1+\eps}}\sup_{0\le\ta\le t_0}\cR_{1,l}(\ta),
			\end{align}
			with 
			\begin{align}
				A&\coloneqq \me\l_2 C_v\log \left(\frac{\me\l_2}{\l_1}4^d\right),\notag\\
				B&\coloneqq 2\tilde C C_v\me4^d\frac{\l_2}{\l_1}.\notag
			\end{align}
			Notice that we also used again that $a\le 4$. This concludes the proof of \eqref{IE p=1} for $f_-$.
			
			Step 2. Following a similar strategy we show \eqref{IE p=1} for $f_+$. Integrating both sides of \eqref{after est comm lb 1} and rearranging the terms, we obtain
			\begin{align}\label{1 1 2}
				s_l^{-1}\int_0^{t}&\frac{\br{N_{f'_+,\ta l}}_\ta}{\br{N_{f_+,\ta l}}_\ta}\rd\ta\notag\\
				\le &\frac{2\tilde CC_{v}t}{s_l^{1+\eps}}\sup_{0\le\ta\le t_l}\left(\frac{\br{N_{B_{R_l}}}_t}{\br{\Nf{f_+}{t}{l}}_t}+\sum_{x\in B_{{R_l}}}\sum_{y\in\Lambda}\abs{J_{xy}}|x-y|^{1+\eps}\frac{\br{ n_y}_t }{\br{\Nf{f_+}{t}{l}}_t}\right)\notag\\
				&+ C_v\log \frac{\br{ N_{f_+,tl}}_t}{\br{ N_{f_+',0l}}_0}.
			\end{align}
			We notice that, for any $t\le (R-r)/v$, it hols that 
			\begin{align}\label{1 1 3}
				\br{\Nf{f+}{t}{l}}_t\le\br{\Nf{f-}{t}{l}}_t\le\br{\Nf{f-}{0}{l}}_0\me^{ts_l^{-1}}\le \br{N_{B_{R_l}}}_0\me^{ts_l^{-1}}\le \me\l_2R_l^d,
			\end{align}
			where the first inequality is due to \lemref{lem7.3}, the second one follows from \eqref{still f} until time $t_l$, and \propref{Prop 4 lb and ub 1} implies $t_l\ge (R-r)/3v$, and the third comes from the assumption of controlled density \eqref{CD} and the fact that $t\le s_l$.
			Inequality \eqref{1 1 3}, together with the geometric properties of $f_+$, \eqref{g and b 0}, and assumption \eqref{aCond}, implies
			\begin{align}\label{1 1 4}
				\log \frac{\br{ N_{f_+,tl}}_t}{\br{ N_{f_+',0l}}_0}\le \log \frac{\me\l_2  R^d}{\l_1r^d}\le \log \frac{4^d\me\l_2  }{\l_1}.
			\end{align}
			Using again \eqref{8 p 1} we can upper bound the remainder term in \eqref{1 1 2} as follows
			\begin{align}\label{1 1 5}
				\sup_{0\le\ta\le t_l}\left(\frac{\br{N_{B_{R_l}}}_t}{\br{\Nf{f_+}{t}{l}}_t}+\sum_{x\in B_{{R_l}}}\sum_{y\in\Lambda}\abs{J_{xy}}|x-y|^{1+\eps}\frac{\br{ n_y}_t }{\br{\Nf{f_+}{t}{l}}_t}\right)\le \frac{\me}{\l_1r_l^d} \sup_{0\le\ta\le t_l}\cR_{1,l}.
			\end{align}
			To lower bound the integrated term in \eqref{1 1 2} we use again \eqref{1 1 3}, 
			\begin{align}\label{1 1 6}
				s_l^{-1}\int_0^{t}&\frac{\br{N_{f'_+,\ta l}}_\ta}{\br{N_{f_+,\ta l}}_\ta}\rd\ta\ge \frac{s_l^{-1}}{ \me  \l_2R_l^d} \int_0^{t}{\br{N_{f'_+,\ta l}}_\ta}.
			\end{align}
			Applying \eqref{1 1 4}--\eqref{1 1 5}--\eqref{1 1 6} to \eqref{1 1 3} leads to the desired estimate \eqref{IE p=1} for $f_+$.
		\end{proof}
		\bigskip

		\section{Proof of \thmref{thm4} Assuming ASTLO Estimates}\label{second mom log}
		In the first part of this section we state the ASTLO estimates, and in the second part we show how they imply \thmref{thm4} for general $p$.
		
		\subsection{Main ASTLO estimates}
		The main result of this subsection are the following ASTLO estimates:
		\begin{itemize}
			\item [1)]\textit{Integral ASTLO estimate}
			\begin{align}\label{IE}
				s_l^{-1}\int_0^t&\br{\del{N_{f_\pm,\ta l}}^{p-1}N_{f'_\pm,\ta l}}_\ta \rd\ta\le A^{(p)}R_l^{dp}+\frac{tB^{(p)}}{s_l^{1+\eps}}\del{\sup_{0\le\ta\le t}\cR_{p,l}(\ta)+\sup_{0\le\ta\le t}\cR_{1,l}(\ta)}.\tag{IE}
			\end{align}
			
			\item [2)]\textit{Particle ASTLO estimates}
			\begin{align}
				&\pm\log\frac{\br{N^p_{f_\mp,t l}}_t}{\br{N^p_{f_\mp,0 l}}_0}\le \frac{C^{(p)}}{R_l^d}+\frac{  t}{ s_{l} } .\tag{PE}\label{PE-}
			\end{align}
		\end{itemize}
		Here we define the remainder term as follows:
		\begin{align}\label{Rp}
			&\cR_{p,l}(\ta):={\br{N^p_{B_{R_l}}}_\ta}+\sum_{x\in B_{{R_l}}}\sum_{y\in\Lambda}\abs{J_{xy}}|x-y|^{1+\eps}{\br{ N^{p-1}_{B_{R_l}}n_y}_\ta }.\tag{Rp}
		\end{align} 
		More precisely we prove the following proposition.
		\begin{proposition}[ASTLO estimates]\label{final prop ind p}
			{Let $\al>d+1$ and let \eqref{CD} hold for $p\ge1$ and some $0<\l_1<\l_2<\infty$.} 
			For  any ${v} >3 \kappa$ and
			$\delta_0\in(0,1)$ we define the following parameter set\begin{align}\label{parameters}
				\mathscr{P}=\Set{v,d,\delta_0,p,\l_1,\l_2}.
			\end{align} 
			Then, there exist constants
			\begin{align}
				A^{(p)}&=A^{(p)}\,(\scrP)>0,\notag\\
				B^{(p)}&=B^{(p)}\,(\scrP)>0,\notag\\
				C^{(p)}&=C^{(p)}\,(\scrP)\ge0,\notag\\
				\rho^{(p)}&=\rho^{(p)}\,(\scrP,\al,C_J)>0,\notag
			\end{align} 
			where $C^{(1)}=0$ and $C^{(p)}>0$ for $p\ge 2$,
			such that for all $l\in\mathbb{N}_0$, {$R>r\ge\rho^{(p)}$} with $R-r>\delta_0r$ satisfying \eqref{aCond}, and $t\le \tilde t_0\equiv (R-r)/v$, estimates \eqref{IE} and \eqref{PE-} hold for $p$ and  any $f\in\ctE$.
		\end{proposition}
		We prove \propref{final prop ind p} by induction on $p$, where the base case  follows from what we have showed in \secref{first mom log}. The proof can be found in \secref{sec proof astlo est}.
		
		\subsection{Completing the proof of \thmref{thm4} assuming main ASTLO Estimates}
		We start the section by proving that if \thmref{thm4} holds  on truncated Fock spaces, then it also holds on the full Fock space. We conclude it by showing how the ASTLO estimates imply \thmref{thm4}.

		\subsubsection{\thmref{thm4} on truncated Fock space implies \thmref{thm4} on the full Fock space}
		Consider the following spaces
		\begin{align}\notag
			&\cF_{\le N_0}\del{\ell^2\del{\Lambda}}\coloneqq\bigoplus_{\hphantom{=\,}n=1\hphantom{+\,}}^{N_0} \mathcal S_n\left(\bigotimes_{j=1}^n \ell^2(\Lambda)\right),\\
			&\cF_{> N_0}\del{\ell^2\del{\Lambda}}\coloneqq\bigoplus_{n=N_0+1}^{\infty} \mathcal S_n\left(\bigotimes_{j=1}^n \ell^2(\Lambda)\right).\notag
		\end{align}
		In this section we show that if, for any fixed $N_0>0$, \thmref{thm4} holds for any initial wave function $\tilde \psi$ such that
		\begin{align}
			\tilde \psi\in\cF_{\le N_0}\del{\ell^2\del{\Lambda}},\notag
		\end{align}
		then \thmref{thm4} holds for any $\psi\in\cF\del{\ell^2\del{\Lambda}}$.

		\begin{proposition}\label{prop big psi}
			Given $\l_2>\l_1>0$ and $p\ge2$. If, for any $N_0>0$, \thmref{thm4} holds for initial states in $\cF_{\le N_0}\del{\ell^2\del{\Lambda}}$ satisfying \eqref{CD} for $p,\, \l_1,$ and $\l_2$, then it also holds for any $\psi\in\cF\del{\ell^2\del{\Lambda}}$ satisfying \eqref{CD} for $p,\, \l_1,$ and $\l_2$.
		\end{proposition}
		\begin{proof}
			Step 1.    Consider $\psi\in\cF\del{\ell^2\del{\Lambda}}.$ We can write it as
			\begin{align}
				\psi=\psi^{\le}+\psi^{>},\notag
			\end{align}
			where
			\begin{align}
				\psi^{\le}\coloneqq P_{N\le N_0}\psi&\in\cF_{\le N_0}\del{\ell^2\del{\Lambda}},\notag\\
				\psi^{>}\coloneqq P_{N> N_0}\psi&\in\cF_{> N_0}\del{\ell^2\del{\Lambda}},\notag
			\end{align}
			where $P_{N\le N_0}$ and $P_{N> N_0}$ are the projections onto $\cF_{\le N_0}\del{\ell^2\del{\Lambda}}$ and $\cF_{> N_0}\del{\ell^2\del{\Lambda}}$, respectively.
			For any $X\in\Lambda$, by linearity it holds
			\begin{align}\label{together}
				\br{\psi_t,N_X^p\psi_t}&=\br{\psi^{\le}_t+\psi^{>}_t,N_X^p\del{\psi^{\le}_t+\psi^{>}_t}}\notag\\
				&=\br{\psi^{\le}_t,N_X^p{\psi^{\le}_t}}+\br{\psi^{>}_t,N_X^p{\psi^{>}_t}}+\br{\psi^{\le}_t,N_X^p{\psi^{>}_t}}+\br{\psi^{>}_t,N_X^p{\psi^{\le}_t}},
			\end{align}
			where 
			\begin{align}\notag
				\psi_t:=\me ^{iHt}\psi.
			\end{align}
			Since $$\sbr{N_X,P_{N\le N_0}}=0=\sbr{N_X,P_{N> N_0}}$$ then
			\begin{align}
				\br{\psi^{>}_t,N_X^p{\psi^{\le}_t}}=0=\br{\psi^{\le}_t,N_X^p{\psi^{>}_t}},\notag
			\end{align}
			and \eqref{together} becomes
			\begin{align}\label{no mixed}
				\br{\psi_t,N_X^p\psi_t}=\br{\psi^{\le}_t,N_X^p{\psi^{\le}_t}}+\br{\psi^{>}_t,N_X^p{\psi^{>}_t}}.  
			\end{align}
			
			Since the total number operator commutes with the Hamiltonian, it holds
			\begin{align}\label{after no mixed}
				\br{\psi^{>}_t,N^p_X{\psi^{>}_t}}\le\br{\psi^{>}_t,N^p{\psi^{>}_t}}=\br{\psi^{>},N^p{\psi^{>}}}.
			\end{align}
			By spectral decomposition we write
			\begin{align}\notag
				\br{\psi^{>},N^p{\psi^{>}}}=\sum_{n>N_0}n^p\norm{P_{N=n}\psi}.
			\end{align}
			Notice that, since we are assuming that $\psi$ satisfies \eqref{CD}, then 
			\begin{align}
				\br{\psi,N^p\psi}=\sum_{n=1}^{\infty}n^p\norm{P_{N=n}\psi}<\infty.\notag
			\end{align}
			The convergence of the sum implies that for any $\eps>0$ there exists an $\tilde N_0$ such that, for every $N_0\ge \tilde N_0$, the following holds
			\begin{align}\label{norm and eps}
				\sum_{n>N_0}n^p\norm{P_{N=n}\psi}\le\eps.
			\end{align}
			Inequality \eqref{norm and eps}, together with \eqref{no mixed}
			and \eqref{after no mixed}, yields
			\begin{align}\label{psi 1}
				\br{\psi_t,N_X^p\psi_t}\le \br{\psi^{\le}_t,N_X^p{\psi^{\le}_t}}+\eps.
			\end{align}

			Choose $X=B_r$, since we are assuming \eqref{main ub 1} holds for $\psi^{\le}$, we can bound the r.h.s. of \eqref{psi 1} by
			\begin{align}\label{psi 2}
				\br{\psi^{\le}_t,N_{B_r}^p{\psi^{\le}_t}}\le\br{\psi^{\le},N_{B_R}^p{\psi^{\le}}}\exp\Set{\frac{C}{R^d}+\frac{tv}{(R-r)}} +\eps.
			\end{align}
			Since $N_X\ge0$ for every $X\in\Lambda$, then \eqref{no mixed} leads to
			\begin{align}\label{psi 3}
				\br{\psi_t^\le,N^p_X\psi_t^\le}\le  \br{\psi_t,N^p_X\psi_t}.
			\end{align}
			Combining \eqref{psi 1}, \eqref{psi 2}, and \eqref{psi 3} we obtain that for every $\eps>0$ there exists an $\tilde N_0$ such that for every $N_0\ge \tilde N_0$ it holds
			\begin{align}\label{psi 4}
				\br{\psi_t,N_{B_r}^p{\psi_t}}\le \br{\psi,N_{B_R}^p{\psi}}\exp\Set{\frac{C}{R^d}+\frac{tv}{(R-r)}} +\eps.
			\end{align}
			Taking the limit $\eps \to 0$, since both the r.h.s and the l.h.s. of \eqref{psi 4} don't depend on $N_0$ anymore,  we obtain the desired estimate.
			
			Step 2. Now let us show that \eqref{main lb 1} also holds for a general $\psi$ satisfying \eqref{CD}. Combining \eqref{psi 1}, \eqref{psi 3}, and \eqref{main lb 1} for $\psi^\le$,  we obtain that for every $\eps>0$ there exists an $\tilde N_0$ such that for every $N_0\ge \tilde N_0$ it holds
			\begin{align*}
				\br{\psi_t,N^p_X\psi_t}&\ge\br{\psi_t^\le,N^p_X\psi_t^\le}\\
				&\ge \br{\psi^\le,N^p_X\psi^\le}\exp\Set{-\del{\frac{C}{R^d}+\frac{tv}{(R-r)}}}\\
				&\ge \del{\br{\psi,N^p_X\psi}-\eps}\exp\Set{-\del{\frac{C}{R^d}+\frac{tv}{(R-r)}}}.
			\end{align*}
			Taking the limit $\eps \to 0$ we conclude the proof.
		\end{proof}
		
		\subsubsection{ASTLO Estimates   Imply \thmref{thm4}}
		\begin{proof}
			We first notice that if \eqref{PE-} holds for an initial state $\psi\in\cF_{\le N_0}\del{\ell^2\del{\Lambda}}$, for some $N_0>0$, that satisfies \eqref{CD}, then \eqref{main lb 1} and \eqref{main ub 1} hold for such initial state and $R,r$ as in \propref{final prop ind p}. This follows by exponentiating both sides of \eqref{PE-}, and making use of the geometric properties of the ASTLOs, \eqref{f and b at t}--\eqref{f and B at t=0}--\eqref{g and b 0}--\eqref{g and B t}.
			
			By \propref{prop big psi}, \eqref{main lb 1} and \eqref{main ub 1} follow for states in $\cF\del{\ell^2\del{\Lambda}}$ with controlled density and radii $R,r$ as in \propref{final prop ind p}.
			
			Thanks to \propref{prop7}, if \eqref{main ub 1} and \eqref{main lb 1} hold for some $v_0\ge3\kappa$ and $a\in(\delta_0+1,4]$, then \eqref{main lb 1} and \eqref{main ub 1} hold for $a>4$ and $v_1:=4v_0$.
		\end{proof}
		
		\bigskip

		\section{Proof of Main ASTLO Estimates Assuming Remainder Estimate}\label{sec proof astlo est} 
		In this section, we assume that the remainder estimate stated below holds, and we derive the ASTLO estimates stated in \propref{final prop ind p}. 
		\begin{proposition}[Remainder estimate]\label{Prop 2 ub 2}
			Let $\al>d+1$ and let \eqref{CD} hold for $p\ge1$ and some $0<\l_1<\l_2<\infty$. 
			Let ${v} > \kappa$, 
			$\delta_0\in(0,1)$, and $\mu>0$. Assume \eqref{IE} for $p-1$ holds for some function $f\in\ctE$, some constants $A^{(p-1)}$, $\,B^{(p-1)}>0$, and radii $ R>r\ge\max\Set{1,\rho^{(p-1)}}$, for some $\rho^{(p-1)}>0$, satisfying \eqref{aCond} and $R-r>\delta_0 r$, $l\in\N_0$ and any $t\le t_0$. Then,	there exist
			\begin{align}
				C_{\cR,p}=C_{\cR,p}\del{A^{(p-1)},\scrP,C_J,\al}>0,\notag
			\end{align}
			such that for
			\begin{align}\notag
				\rho^{(p)}=\rho^{(p)}(C_{\cR,p},C_{\cR,1},B^{(p-1)}, \rho^{(p-1)}, \scrP)>0
			\end{align}
			(see definition in \eqref{9 o}), all $l\in\N_0$, and $R>r\ge\rho^{(p)}$ with $R-r>\delta_0r$ and satisfying \eqref{aCond},
			\begin{align}\label{rem Pp bound}
				r^{-dp}\sup_{0\le\ta\le t_0}\cR_{p,l}(\ta)\le C_{\cR,p}.
			\end{align}
		\end{proposition}
		The proof of this remainder estimate can be found in \secref{second mom log b}. Assuming \propref{Prop 2 ub 2}, we show the ASTLO estimates by induction on the moment parameter $p$, where the base case has been proved in \secref{sec prop 1 and 2}. The strategy of the proof is similar to the one developed in \secref{first mom log}, and can be fund in \secref{sub sec proof strategy astlo est}. The reader can find the full proof in \secref{sub sec proof astlo est}.

		\subsection{Proof Strategy for Main ASTLO Estimates}\label{sub sec proof strategy astlo est}
		The main challenge of proving \propref{final prop ind p} lays, as in the case $p=1$, in controlling the contribution from far away particle. To overcome it, we follow closely the method we developed in \secref{sec prop 1 and 2}, with two main differences:
		\begin{itemize}
			\item[a)] To control higher moments of the ASTLO we need to construct an induction on the moment parameter $p$. \secref{sec prop 1 and 2} constitutes the base case of this induction.
			\item[b)] We do not need to define a new bad time until when all scales are populated enough. 
		\end{itemize}
		The second point is due to the fact that, thanks to \propref{Prop 4 lb and ub 1}, given $v>\kappa$, for every $t\le \tilde t_0$ with
		\begin{align}\notag
			t_0= \frac{(R-r)}{3v},
		\end{align}
		\begin{align}\notag
			\br{N_{f_{\pm,t}}}_t\ge \l_1r^{d}\me^{-1}.
		\end{align}
		This, together with Jensen inequality and the fact that $N_{f_{\pm,t}}$ is a positive operator, leads to 
		\begin{align}\label{jens + p=1}
			\br{N^p_{f_{\pm,t}}}_t\ge \left(\l_1r^d\me^{-1}\right)^p,
		\end{align}
		for every $p\ge 1$ and $t\le t_0$.
		Since we aim to control higher moments of the number operator we need to assume \eqref{CD} holds for $p\ge2$, that means that we require additionally that there exist ${0<\l_1<\l_2<\infty}$  s.th.
		\begin{align}\notag
			(\l_1r^d)^q\le \wzdel{N_{B_r(x)}^q}\le (\l_2r^d)^q \qquad  (1\le q\le p,x\in \Lam, r\ge1).
		\end{align}
		In this and the following sections, thanks to \propref{prop big psi}, we can restrict our attention to initial states ${\psi\in\cF_{\le N_0}\del{\ell^2\del{\Lambda}}}$ for some $N_0>0$ which satisfy \eqref{CD} for some $\l_2>\l_1>0$ and $p\ge1$. The proof of \propref{final prop ind p} follows the following structure. We start by differentiating $\log\br{ N^p_{f,t}}_t$ with respect to time, obtaining the following inequality
		\begin{align}
			\pm\log\frac{\br{\Nf[p]{f_\mp}{t}{l}}_t}{\br{\Nf[p]{f_\mp}{0}{l}}_0}\le& \frac{pv'}  {s_l\left(\l_1r_l^de^{-1}\right)^p}\int_0^t{\br{\left(N_{f_\mp,\ta l}+1\right)^{p-2}N_{f'_\mp,\ta l}}_\ta}\rd \ta\notag\\
			&+\frac{ ptC_f}{s_l^{1+\eps}\left(\l_1r_l^de^{-1}\right)^p}\sup_{0\le \ta\le t_0}\del{{\cR}_{p,l}(\ta)+{\cR}_{1,l}(\ta)}.\label{before int est}
		\end{align}
		We notice that the integral ASTLO estimates \eqref{IE} would allow us to bound the integrated term appearing in \eqref{before int est}. Realizing that the integral ASTLO estimates \eqref{IE} and the particle ASTLO estimates \eqref{PE-} are intertwined, we set up an induction on the moment parameter $p$, that will allow us to show both sets of estimates in parallel. 
		\begin{itemize}
			\item \propref{prop ind basis} and \thmref{thm4} for $p=1$, that we showed in \secref{sec prop 1 and 2}, constitute the base case of this induction. 
			\item In \propref{prop pe implies ie}, we show that if \eqref{IE} hold for $p-1$ and \eqref{PE-} hold for $p$, then \eqref{IE} will hold for $p$ too.
			\item To close the induction we make use of the remainder estimate, \propref{Prop 2 ub 2}, together with \lemref{before bound rem}, to show that \eqref{IE} and \eqref{PE-} for $p-1$, imply \eqref{PE-} for $p$.
		\end{itemize}
		The reader can refer to the flow chart in \figref{flow chart 1} to have a better understanding on the interplay between \secref{sec prop 1 and 2} and the present section. 
		
		\begin{figure}[t]
			\centering
			\begin{tikzpicture}[
				node distance=1.5cm and 3cm,  line width=0.3mm,
				box/.style = {rounded corners, fill=blue!12, align=center, minimum height=2.2em},
				box2/.style = {rounded corners, fill=green!25, align=center, minimum height=2em},
				arrow/.style = {thick, -Stealth},
				every label/.append style = {font=\scriptsize, align=center},
				box3/.style = {rounded corners, fill=purple!20, align=center, minimum height=2.3em, font=\large}  , 
				highlight/.style={ line width=2pt, draw=yellow!85!red},
				box4/.style={highlight,rounded corners, fill=yellow!5, align=center, font=\Small},
				arrow/.style={-stealth,shorten <=2pt, shorten >=2pt},
				shorten <=2pt, shorten >=2pt,
				myfit/.style = {rounded corners, draw, dashed},
				scale=0.7
				]
				
				\node[ rounded corners,  align=center, minimum height=2.3em, font=\large,draw=purple!35, line width=0.5mm, fill=purple!10] (proofthm) {~\large Proof of main ASTLO estimates (\propref{final prop ind p})~};
				\node [rounded corners,  line width=0.5mm, fill=purple!10, below=.2cm of proofthm] (wh prop 83) {\textit{Particle estimate ASTLO \eqref{PE-}} and
					\textit{Integrated ASTLO estimate \eqref{IE}}};

				\node[below left=2cm and 2cm of proofthm] (1) {1)};
				
				\node[box2,highlight, right=1mm of 1] (wh prop82) {Differential inequality for $\log \br{ N^p_{f,t}}_t$} ;
				\node[box, right=5mm of wh prop82](prop82) {Prop. \ref{Prop 1 ub 2}~};
				\node[box4,  right=1cm of prop82, font=\Small] (new 1) { differentiating\\the log will produce an\\additional term at denominator,\\ key to bound remainder};
				
				\node [box2, below=.8cm of wh prop82] (w prop 83) {~Induction on $p$ to control the derivative of $\log \br{ N^p_{f,t}}_t$
				};

				\node[box, right=5mm of w prop 83] (prop83) {Prop. \ref{final prop ind p}};

				\node [left=1mm of w prop 83] (2) {2)};
				
				\draw[arrow] (wh prop82)--(w prop 83);
				%\draw[arrow] (w prop 83)--(wh prop 83);
				\node[below right=-2mm and .7cm of 2] (2a) {};

				\node [box3, below right=2.cm and 0.58cm of 2a] (ind base) {\normalsize Induction Base: $p=1$};
				\node [left=1mm of ind base] (a) {};
				\draw[arrow] (ind base) -| (2a.south west);
				
				\node [box2, highlight, right=2cm of ind base] (badtime) {Define a bad time $T_1$, until which\\ all scales are populated enough};
				\node [box4, below=1.5cm of new 1, font=\Small] (new 2) {necessary to control the new\\term at denominator  produced in step 1).\\
					Useful thanks to Prop. \ref{Prop 4 lb and ub 1} };
				
				%\node [box2, below=0.8cm of badtime] (step 1) { Prove 1) for $p=1$ to control time\\ evolution of the ASTLO until $T_1$};
				\node [box2, below=0.8cm of badtime] (step 1) { Derive differential inequality for $\log\br{ N_{f,t}}_t$\\ to control the ASTLO until $T_1$};
				\node[box, right=.5cm of step 1] (step1 1) {Prop. \ref{Prop 1 ub 1} \\Prop. \ref{prop ind basis}};
				
				\node [box2, below=0.8 cm of step 1] (ms) {Use multiscale induction to\\ control the remainder and \\obtain propagation bound until $T_1$};
				\node [box, right=0.5cm of ms] (ms 1) {Prop. \ref{cor remainder}};
				
				\node [box2, below=0.8cm of ms] (lb t) {Lower bound $T_1$};
				\node [box, right=0.5cm of lb t ] (lb t1) {Prop. \ref{Prop 4 lb and ub 1}};
				\node[fit=(badtime)(step 1)(ms)(lb t)(ms 1)(step1 1),myfit] (myfit1) {};
				\node[box3, below right=4.9cm and -3.6cm of ind base, font=\normalsize] (ind step) {Induction Step: $p-1\Rightarrow p$};
				\node [below right=.8mm and 2.5cm of myfit1.south] (where){\secref{sec prop 1 and 2} \checkmark};
				\draw[arrow] (ind step) -| (2a.south west);

				% \begin{scope}[transform canvas={xshift=-1.5cm}]
					% \draw[dashed, line width=0.2mm] (ind step)--+(0,-2.1)--+(3.3,-2.1);
					% \end{scope}
				% \begin{scope}[transform canvas={xshift=-1.5cm}]
					% \draw[dashed, line width=0.2mm] (ind step)--+(0,-4.1)--+(3.3,-4.1);
					% \end{scope}
				
				\node [box2, below right=1.1cm and -3.1cm of ind step]  (PE and IE) 
				{\eqref{PE-} for $p$ $+$ \eqref{IE} for $p-1$
					\\$\Rightarrow\;$ \eqref{IE} for $p$};
				\node [box,right=7mm of PE and IE] (prop85) {Prop. \ref{prop pe implies ie}};
				\node[left=1mm of PE and IE] (pe ie) {};
				\draw[arrow] (PE and IE) -| ([xshift=1cm]ind step.south west);

				\node [box2, below right=.8cm and 1mm of pe ie] (ms p){Multiscale induction to show\\ \eqref{PE-} + \eqref{IE} for $p-1 \qquad\Rightarrow\qquad$ \eqref{PE-} for $p$\qquad};
				\node [box, right=7mm of ms p] (prop 93) {Prop. \ref{Prop 2 ub 2}};
				\draw[arrow] (ms p) -| ([xshift=1cm]ind step.south west);
				
				\node [below=0.51cm of prop85] (b prop85) {};
				\node[above=5.5cm of prop85] (a prop85) {};
				%\draw[->, line width=0.4mm] (prop85) -- (b prop85);
				%\draw[->, line width=0.4mm] (a prop85) -- (prop85) ;
				
				\draw[arrow] (badtime)--(step 1);
				\draw[arrow] (step 1)--(ms);
				\draw[arrow] (ms) -- (lb t);
				
				%\draw[arrow] (lb t)--+(0,-1.8)--+(-3.9,-1.8);
				\draw ([xshift=4.2cm]myfit1.south west) |- ([xshift =-1mm,yshift=0.7cm] PE and IE.north);
				\draw [arrow]([yshift=0.82cm] PE and IE.north)-- (PE and IE);
				\draw (PE and IE)--(prop85);
				
				\draw ([xshift=4.2cm]myfit1.south west) |- ([xshift =5.1cm,yshift=2.555cm] ms p.north);
				\draw [arrow]([xshift =5cm,yshift=2.67cm] ms p.north)-- ([xshift=5cm]ms p.north);
				
				\draw (ms p)--(prop 93);
				
				\draw[arrow]([yshift=-0.9cm]myfit1.north west)  -- (ind base);
				\draw (step1 1)--(step 1);
				\draw (ms 1)--(ms);
				\draw (lb t1)--(lb t);
				
				\draw (prop83)--(w prop 83);
				\draw (prop82)--(wh prop82);
				\begin{scope}[transform canvas={xshift=.7em}]
					\draw [arrow,dashed, draw=yellow!85!red, line width=0.4mm] (new 1) -- node [sloped, above, font=\Small] {requires} (new 2);
				\end{scope}
				\draw[yellow!85!red, line width=0.4mm] ([xshift=-3.7cm]myfit1.north)|-(new 2);
				\draw[yellow!85!red, line width=0.4mm] (wh prop82)|-([yshift=-.2cm]new 1.north west);
				
				\node[box2, below=0.6cm of ms p] (hol) {Separate far away and close by particles via H\"older};
				\node[box, right=0.6cm of hol] (lem) {\lemref{lemma holder}};
				\coordinate (bbox1) at (current bounding box.south);
				\coordinate (bbox2) at (current bounding box.north);
				% \node[box4, left=0.5cm of hol] (new 3) {allows to suitably\\ separate the contribution from\\ far away and close up particles};
				% \draw[line width=.4mm, draw=yellow!85!red](new 3)--(hol);
				\draw[arrow] (hol)--(ms p);
				\draw (lem)--(hol);
				\pgfresetboundingbox
				\path[use as bounding box] (bbox1) rectangle (bbox2);

			\end{tikzpicture}
			\caption{Flowchart illustrating the proof structure of \thmref{thm4}. The yellow boxes correspond to conceptually new ideas beyond our prior multiscale induction scheme \cite{lemm2023microscopic}.}
			\label{flow chart 1}
		\end{figure}

		\subsection{Proof of ASTLO estimates (\propref{final prop ind p}) assuming remainder estimate (\propref{Prop 2 ub 2}) }\label{sub sec proof astlo est}
		
		Again we consider a function $f\in\ctE$, and we differentiate the logarithm of the $p$-th moment of the ASTLO. We control the commutator deriving from such process as it was done in \cite{lemm2023microscopic}. This leads to the following proposition.
		
		\begin{proposition}\label{Prop 1 ub 2} Let $\al>d+1$ and let \eqref{CD} hold for $p\ge1$ and some $0<\l_1<\l_2<\infty$. Then there exists  a constant $C_p$ such that, for any $f\in\ctE,\,{v} >\kappa$, $\eps>0$, $R>r\ge1$, $l\in\N_0$, and $t\le t_0$, it holds
			\begin{align}
				\pm\log\frac{\br{\Nf[p]{f_\mp}{t}{l}}_t}{\br{\Nf[p]{f_\mp}{0}{l}}_0}\le& \frac{pv'}  {s_l\left(\l_1r_l^de^{-1}\right)^p}\int_0^t{\br{\left(N_{f_\mp,\ta l}+1\right)^{p-2}N_{f'_\mp,\ta l}}_\ta}\rd \ta\notag\\
				&+\frac{ tC_p}{s_l^{1+\eps}\left(\l_1r_l^de^{-1}\right)^p}\sup_{0\le \ta\le t_0}\del{{\cR}_{p,l}(\ta)+{{\cR}}_{1,l}(\ta)},\label{1 p}
			\end{align}
			where ${{\cR}}_{p,l}$ is defined as in \eqref{Rp}.
		\end{proposition}
		
		The proof of this proposition can be found in \secref{prop diff log p}.
		Our next goal is to derive a bound for the integrated term in \eqref{1 p}, namely the integral ASTLO estimate \eqref{IE}. The idea is to show \eqref{IE} and \eqref{PE-} by induction on the moment parameter $p$ using what we have done in \secref{sec prop 1 and 2} as the base case. In fact,
		
		\begin{itemize}
			\item \propref{prop ind basis} implies \eqref{IE} for $p=1$,
			\item  \propref{Prop 4 lb and ub 1} together with \eqref{still f} and \eqref{still f 2}, leads to \eqref{PE-} for $p=1$.
		\end{itemize}
		
		As for the induction step, we notice that  if \eqref{PE-} holds for a certain $p\ge1$, then \eqref{IE} will also hold for the same $p$.
		
		\begin{proposition}[Particle ASTLO estimates imply integral ASTLO estimates]\label{prop pe implies ie}
			Let the assumptions of \propref{Prop 1 ub 2} hold.
			Consider ${v} > \kappa$ 
			and $\eps>0$ as in \eqref{epsDef}. Assume estimates \eqref{IE} for $p-1$ and \eqref{PE-} for $p$ hold for some constants $A^{(p-1)},\,B^{(p-1)}>0, C^{(p)}\ge0$, radii $R> r\ge 1 $ satisfying \eqref{aCond}, $l\in\N_0$, $f\in \ctE$, and  $t\le  t_0$. Then there exist 
			\begin{align}
				A^{(p)}&=A^{(p)}\,(A^{(p-1)},C^{(p)},\Tilde{ \mathscr{P}})>0,\notag\\
				B^{(p)}&=B^{(p)}\,(B^{(p-1)},C^{(p)}, \Tilde{ \mathscr{P}})>0,\notag
			\end{align} 
			such that estimate \eqref{IE} holds for $p$, $R> r\ge 1 $, $l\in\N_0$, $t\le  t_0$,  and $f\in\ctE$. We defined the following set of parameters $\Tilde{ \mathscr{P}}:=(d,v,\l_1,\l_2,p)$.
		\end{proposition}
		We prove the proposition in \secref{proof pe impl ip}. To close the induction we need to show that \eqref{IE} and \eqref{PE-} for $p-1$, imply \eqref{PE-} for $p$. To do so, we need \propref{Prop 2 ub 2} and the following lemma.
		
		\begin{lemma}\label{before bound rem}
			In the same setting as in \propref{Prop 2 ub 2}	there exist
			\begin{align}
				C^{(p)}&=C^{(p)}\,(A^{(p-1)}, \scrP')>0\notag\\
				\rho^{(p)}&=\rho^{(p)}\,( B^{(p-1)},C_{\cR,1}, \rho^{(p-1)},\scrP',\mu,\eps,\delta_0)>0,\label{rhp p}
			\end{align} 
			such that for all {$R>r\ge\rho^{(p)}$} the following inequalities hold
			\begin{align}
				\pm\log\frac{\br{\Nf[p]{f_\mp}{t}{l}}_t}{\br{\Nf[p]{f_\mp}{0}{l}}_0}&\le {+\frac{C^{(p)}}{R_l^d}+\frac{  t}{\mu s_{l} } \del{\frac{1}{r^{dp}}\sup_{0\le \ta\le t_l}\cR_{p,l}(\ta)+C_{\cR,1}}},\label{4 p}
			\end{align}
			for any $t\le t_0$, with $C_{\cR,1}$ as in \eqref{655}. Here ${ \mathscr{P}'}:=(d,v,\l_1,p)$.
		\end{lemma}
		The proof of this lemma can be found in \secref{proof lemma}.
		We notice that, if we are able to properly control the remainder term appearing on the r.h.s. of \eqref{4 p}, we recover exactly \eqref{PE-}, and we would be able to close the induction on the moments. The remainder estimate stated in 
		\propref{Prop 2 ub 2} concludes the proof of \propref{final prop ind p}, as shown below.

		\subsubsection{Proof of \propref{final prop ind p}}\label{proof final prop ind}
		\begin{proof}
			Since we already showed in \propref{prop pe implies ie} that \eqref{IE} holds for $p$ if we assume \eqref{PE-} holds for $p$ and \eqref{IE} holds for $p-1$, what is left to close the induction on $p$ is to show that \eqref{IE} and \eqref{PE-} for $p-1$ imply \eqref{PE-} for $p$  and velocity $\tilde v >\kappa$ up to time $t_0=(R-r)/3\tilde v$. This follows applying \lemref{before bound rem} with $\mu=C_{\cR,p}+C_{\cR,1}$  with $C_{\cR,p}$ as in \eqref{rem Pp bound} and choosing $\rho^{(p)}$ as in \propref{Prop 2 ub 2}. The bound on the remainder \eqref{rem Pp bound} leads to the desired estimates \eqref{PE-} for $p$ and velocity $\tilde v >\kappa$ up to time $t_0=(R-r)/3\tilde v$. We conclude by rescaling the velocity as follows $v:=3\tilde v$.
		\end{proof}
		
		\subsubsection{Proof of \propref{Prop 1 ub 2}}\label{prop diff log p}
		
		To show the proposition we combine the strategy used in \cite[Prop.~5.1]{lemm2023microscopic} with the method we developed  in \secref{sec prop 1 and 2}. 
		
		Let us first show \eqref{1 p} for $f_-$. We start by differentiating the logarithm of the $p$-th moment of the ASTLOs with respect to time.
		\begin{align}\label{differentiatin log with f 2}
			\frac{\rd}{\rd t}\log\br{\Nf[p]{f_-}{t}{l}}_t=\frac{1}{\br{\Nf[p]{f_-}{t}{l}}_t}\left(-\frac{pv'}{s_l}\br{\Nf{f_-'}{t}{l}\,\Nf[p-1]{f_-}{t}{l}}_t+\br{\left[iH,\Nf[p]{f_-}{t}{l}\right]}_t\right).
		\end{align}
		Again, the denominator in front of the r.h.s. of \eqref{differentiatin log with f 2} does not change the proof strategy of \cite[Prop.~5.1]{lemm2023microscopic} for $n=1$,  so we can carry out all the steps and,  obtain the following bound
		\begin{align}\label{before int f 2}
			\frac{\rd}{\rd t}\log\br{\Nf[p]{f_-}{t}{l}}_t\le& \frac{pv'}  {s_l}\frac{\br{\left(N_{f_-,tl}+1\right)^{p-2}N_{f'_-,tl}}_t}{\br{\Nf[p]{f_-}{t}{l}}_t}\\
			&+\frac{p(\kappa-v')} {s_l}\frac{\br{\left(N_{f_-,tl}+1\right)^{p-1} N_{f'_-,tl}}_t}{\br{\Nf[p]{f_-}{t}{l}}_t}\notag\\
			&+\frac{ pC_{f}}{s_l^{1+\eps}}Q^-_p(\ta)\notag.
		\end{align}
		with
		\begin{align}\notag
			Q^\pm_p(\ta)\coloneqq  &\frac{\br{\left(N_{f_\pm,tl}+1\right)^{p-1}N_{B_{{R}_l}}}_t}{\br{\Nf[p]{f_\pm}{t}{l}}_t}+\sum_{x\in B_{{R}_l}}\sum_{y\in\Lambda}\abs{J_{xy}}|x-y|^{1+\eps}\frac{\br{\left(N_{f_\pm,tl}+1\right)^{p-1} n_y}_t}{\br{\Nf[p]{f_\pm}{t}{l}}_t} .
		\end{align}
		Thanks to \eqref{Cf cond}, we can get rid of the $f$-dependency in the constant appearing in the last line of \eqref{before int f 2}:
		\begin{align}\label{before int f 3}
			\frac{\rd}{\rd t}\log\br{\Nf[p]{f_-}{t}{l}}_t\le& \frac{pv'}  {s_l}\frac{\br{\left(N_{f_-,tl}+1\right)^{p-2}N_{f'_-,tl}}_t}{\br{\Nf[p]{f_-}{t}{l}}_t}\\
			&+\frac{p(\kappa-v')} {s_l}\frac{\br{\left(N_{f_-,tl}+1\right)^{p-1} N_{f'_-,tl}}_t}{\br{\Nf[p]{f_-}{t}{l}}_t}\notag\\
			&+\frac{ 2p\tilde C}{s_l^{1+\eps}}Q^-_p(\ta)\notag.    
		\end{align}
		Now we integrate both sides of \eqref{before int f 3} until time $t\le t_0$, we apply \eqref{jens + p=1} to the denominators, and, since $ f_{-,tl}\le f_{-,0l}$ for every $t>0$, we apply \eqref{f and B at t=0} to obtain
		
		\begin{align}\label{d}
			\log\frac{\br{\Nf[p]{f_-}{t}{l}}_t}{\br{\Nf[p]{f_-}{0}{l}}_0}\le& \frac{pv'}  {s_l\left(\l_1r_l^de^{-1}\right)^p}\int_0^t{\br{\left(N_{f_-,\ta l}+1\right)^{p-2}N_{f'_-,\ta l}}_\ta}\rd \ta\notag\\
			&+\frac{p(\kappa-v')} {s_l}\int_0^t\frac{\br{\left(N_{f_-,\ta l}+1\right)^{p-1} N_{f'_-,\ta l}}_\ta}{\br{\Nf[p]{f_-}{\ta}{l}}_\ta}\mathrm{d} \ta\notag\\
			&+\frac{ 2pt\tilde C}{s_l^{1+\eps}\left(\l_1r_l^de^{-1}\right)^p}\sup_{0\le \ta\le t_0}{\hat{\cR}}_{p,l}(\ta).
		\end{align}
		with ${\hat{\cR}}_{p,l}(\ta)$ defined as 
		\begin{align}
			\hat{\cR}_{p,l}(\ta):={\br{\left(N_{B_{R_l}}+1\right)^{p-1}N_{B_{R_l}}}_\ta}+\sum_{x\in B_{{R_l}}}\sum_{y\in\Lambda}\abs{J_{xy}}|x-y|^{1+\eps}{\br{ \left(N_{B_{R_l}}+1\right)^{p-1}n_y}_\ta }.\notag
		\end{align}
		Since  $\del{1+x}^{p-1}\le C_p\del{1+x^{p-1}}$, it holds
		\begin{align}\label{Pp=Pp+P1}
			\hat{\cR}_{p,l}\del{\ta}\le C_p\del{\cR_{p,l}\del{\ta}+\cR_{1,l}\del{\ta}},
		\end{align}
		with $\cR_{p,l}$ as in \eqref{Rp}. Since $v'>\k$, the second term on the r.h.s. of \eqref{d} is negative. Dropping it and applying \eqref{Pp=Pp+P1}, we derive \eqref{1 p}.
		
		We follow a similar procedure in order to show \eqref{1 p} for $f_+$, with the only difference that we will use the ASTLO obtained through the second reparametrization of $f\in\cE.$ We start by differentiating $-\log\br{\Nf[p]{f_+}{t}{l}}_t$ with respect to time.
		\begin{align}\label{differentiatin log with f 2 ub}
			-\frac{\rd}{\rd t}\log\br{\Nf[p]{f_+}{t}{l}}_t=\frac{1}{\br{\Nf[p]{f_+}{t}{l}}_t}\left(-\frac{v'p}{s_l}\br{\Nf{f_+'}{t}{l}\Nf[p-1]{f_+}{t}{l}}_t+\br{\left[-iH,\Nf[p]{f_+}{t}{l}\right]}_t\right).
		\end{align}
		To estimate the r.h.s. of \eqref{differentiatin log with f 2 ub} we follow the same steps as before and we obtain 
		\begin{align}\label{eq 9}
			-\frac{\rd}{\rd t}\log\br{\Nf[p]{f_+}{t}{l}}_t\le& \frac{pv'}  {s_l}\frac{\br{\left(N_{f_+,tl}+1\right)^{p-2}N_{f'_+,tl}}_t}{\br{\Nf[p]{f_+}{t}{l}}_t}\notag\\
			&+\frac{p(\kappa-v')} {s_l}\frac{\br{\left(N_{f_,tl}+1\right)^{p-1} N_{f'_+,tl}}_t}{\br{\Nf[p]{f_+}{t}{l}}_t}\notag\\
			&+\frac{ pC_{f}}{s_l^{1+\eps}} Q^+_p(t).
		\end{align}    
		Recalling $f_{+,tl}\le f_{-,tl}$ for every $t<t_0$ and following the same steps as in step 1 we derive \eqref{1 p}.
		\qed
		
		\bigskip
		\subsubsection{Proof of \propref{prop pe implies ie}}\label{proof pe impl ip}
		
		Step 1. Assume \eqref{PE-} holds for $p$ with some constant $C^{(p)}$ and for some $R>r\ge 1$ then, thanks to $s_l> t_0$  and \eqref{ub Nf 0}, it  implies
		\begin{align}\label{Pe- no s}
			\br{N^p_{f_-,t l}}_t\le \br{N^p_{f_-,0 l}}_0\me^{C^{(p)}+1}\le  \del{\l_2R^d}^p \me^{C^{(p)}+1}.
		\end{align}
		Recall inequality \eqref{d}, which holds for such choice of $R,r$ and apply \eqref{Pp=Pp+P1} to it to obtain,
		\begin{align}\label{11 p}
			\log\frac{\br{\Nf[p]{f_-}{t}{l}}_t}{\br{\Nf[p]{f_-}{0}{l}}_0}\le& \frac{pv'}  {s_l\left(\l_1r_l^de^{-1}\right)^p}\int_0^t{\br{\left(N_{f_-,\ta l}+1\right)^{p-2}N_{f'_-,\ta l}}_\ta}\rd \ta\notag\\
			&+\frac{p(\kappa-v')} {s_l}\int_0^t\frac{\br{\left(N_{f_-,\ta l}+1\right)^{p-1} N_{f'_-,\ta l}}_\ta}{\br{\Nf[p]{f_-}{\ta}{l}}_\ta}\mathrm{d} \ta\notag\\
			&+\frac{ tC_p}{s_l^{1+\eps}\left(\l_1r_l^de^{-1}\right)^p}\sup_{0\le \ta\le t_0}\del{{\cR}_{p,l}(\ta)+{\cR}_{1,l}(\ta)},
		\end{align}
		where ${{\cR}}_{p,l}(\ta)$ is defined as in \eqref{Rp}. 
		The geometric property of the ASTLO \eqref{f and B at t=0}, together with the assumption of controlled density at initial time \eqref{CD}, implies the following bound
		\begin{align}\label{ub Nf 0}
			\br{\Nf[p]{f_-}{0}{l}}_0\le \br{N^p_{B_R}}_0\le \del{\l_2R^d}^p.
		\end{align}
		Thanks to  \eqref{jens + p=1}, \eqref{ub Nf 0}, and $a\le 4$ we can lower bound the logarithm in \eqref{11 p} by
		\begin{align}\label{13 p}
			\log\frac{\br{\Nf[p]{f_-}{t}{l}}_t}{\br{\Nf[p]{f_-}{0}{l}}_0}\ge p\log\del{\frac{\l_1}{4^d\me\l_2}}.
		\end{align}
		
		Inequality \eqref{Pe- no s} leads to the following lower bound the second summand on the r.h.s. of \eqref{11 p} 
		\begin{align}\label{12 p}
			\int_0^t\frac{\br{\left(N_{f_-,\ta l}+1\right)^{p-1} N_{f'_-,\ta l}}_\ta}{\br{\Nf[p]{f_-}{\ta}{l}}_\ta}\mathrm{d} \ta\ge \frac{\me^{-(C^{(p)}+1)}}{\del{\l_2 R_l^d}^p}   \int_0^t{\br{\left(N_{f_-,\ta l}+1\right)^{p-1} N_{f'_-,\ta l}}_\ta}\mathrm{d} \ta.
		\end{align}
		We apply  \eqref{13 p} and \eqref{12 p} to \eqref{11 p}, after rearranging the terms and dividing both sides by $p$, we obtain
		\begin{align}\label{o 1}
			&\frac{(v'-\k)\me^{-(C^{(p)}+1)}}{s_l\del{\l_2 R_l^d}^p}\int_0^t{\br{\left(N_{f_-,\ta l}+1\right)^{p-1} N_{f'_-,\ta l}}_\ta}\mathrm{d} \ta\notag\\
			&\quad\le \frac{v'}  {s_l\left(\l_1r_l^de^{-1}\right)^p}\int_0^t{\br{\left(N_{f_-,\ta l}+1\right)^{p-2}N_{f'_-,\ta l}}_\ta}\rd \ta\notag\\
			&\qquad + \frac{Ct}{(\l_1r_l^{d}e^{-1})^ps_l^{1+\eps}}\del{\sup_{0\le\ta\le t}\cR_{p,l}(\ta)+\sup_{0\le\ta\le t}\cR_{1,l}(\ta)}\notag\\
			&\qquad +\log\del{\frac{4\me\l_2}{\l_1}}.
		\end{align}
		Assume \eqref{IE} holds for $p-1$, for some positive constants $A^{(p-1)}$ and $B^{(p-1)}$. Then, applying \eqref{IE} to the first term on the r.h.s. of \eqref{o 1}, we obtain 
		\begin{align}
			&\frac{(v'-\k)\me^{-(C^{(p)}+1)}}{s_l\del{\l_2 R_l^d}^p}\int_0^t{\br{\left(N_{f_-,\ta l}+1\right)^{p-1} N_{f'_-,\ta l}}_\ta}\mathrm{d} \ta\notag\\
			&\hspace{2.em}\le \frac{v'}  {\left(\l_1r_l^de^{-1}\right)^p}\del{A^{(p-1)}R_l^{d(p-1)}+\frac{B^{(p-1)}t}{s_l^{1+\eps}}\del{\sup_{0\le\ta\le t}\cR_{p-1,l}(\ta)+\sup_{0\le\ta\le t}\cR_{1,l}(\ta)}}\notag\\
			&\hspace{3em} +\frac{ Ct}{s_l^{1+\eps}\left(\l_1r_l^de^{-1}\right)^p}\sup_{0\le \ta\le t_0}\del{{\cR}_{p,l}(\ta)+{\cR}_{1,l}(\ta)}\notag\\
			&\hspace{3em} +\log\del{\frac{4\me\l_2}{\l_1}}.\label{14 p}
		\end{align}
		
		Notice that, given $X\subset\Lambda$, $N_X$ has integer eigenvalues. This implies that ${N^p_X\ge N^q_X}$, for every ${p\ge q\ge1}$, leading to
		\begin{align}\label{Pp and Pq big q}
			{ \cR_{p,l}\del{\ta}}\ge \cR_{q,l}\del{\ta}\qquad \text{for }p\ge q\ge 2.
		\end{align}
		Notice that, if $q=1$, for every $p\ge2$ it holds 
		\begin{align}\label{Pp and Pq small q}
			{ \cR_{p,l}\del{\ta}}+\cR_{1,l}\del{\ta}\ge \cR_{1,l}\del{\ta} =\frac{1}{2}\del{\cR_{q,l}\del{\ta}+\cR_{1,l}\del{\ta}}.
		\end{align}
		Inequalities \eqref{Pp and Pq big q} and \eqref{Pp and Pq small q} show us that there exists a constant $C_1\in\Set{1,1/2}$, depending on $q$, such that
		\begin{align}\label{Pp and Pq}
			{ \cR_{p,l}\del{\ta}}+\cR_{1,l}\del{\ta}\ge C_1\del{\cR_{q,l}\del{\ta}+\cR_{1,l}\del{\ta}}\qquad \text{for }p\ge q\ge 1.
		\end{align}
		After dividing both sides of \eqref{14 p} by 
		\begin{align}
			\frac{(v'-\k)\me^{-(C^{(p)}+1)}}{\del{\l_2 R_l^d}^p},\notag
		\end{align}
		and applying \eqref{Pp and Pq}, we obtain
		\begin{align}\label{2 o}
			&s_l^{-1}\int_0^t{\br{\left(N_{f_-,\ta l}+1\right)^{p-1} N_{f'_-,\ta l}}_\ta}\mathrm{d} \ta\notag\\
			&\hspace{2.em}\le \frac{t}{s_l^{1+\eps}} \del{\frac{\me\l_2}{\l_1}a^d}^p\frac{\me^{C^{(p)}+1}}{v'-\kappa}\del{v'B^{(p-1)}+C}\del{\sup_{0\le\ta\le t}\cR_{p,l}(\ta)+\sup_{0\le\ta\le t}\cR_{1,l}(\ta)}\notag\\
			&\hspace{3em} +\frac{\del{\l_2 R_l^d}^p\me^{C^{(p)}+1}}{(v'-\k)}\del{\frac{v'A^{(p-1)}}{(\l_1e^{-1})^p} a^{dp}R_l^{-d}+\log\del{\frac{4\me\l_2}{\l_1}} }.
		\end{align}
		Since $a\le4$ and $R_l\ge 1$, inequality \eqref{2 o}, implies \eqref{IE} for $p$ and $f_-$, with constants
		\begin{align}
			&B^{(p)}:=\del{\frac{\me\l_2}{\l_1}4^d}^p\frac{\me^{C^{(p)}+1}}{v'-\kappa}\del{v'B^{(p-1)}+C},\notag\\
			&A^{(p)}:=\frac{\l_2^p\me^{C^{(p)}+1}}{(v'-\k)}\del{\frac{v'A^{(p-1)}}{(\l_1e^{-1})^p} 4^{dp}+\log\del{\frac{4\me\l_2}{\l_1}} }.\notag
		\end{align}
		
		2. Now we show \eqref{IE} for $f_+$ and $p$, assuming \eqref{PE-} holds for $p$ and \eqref{IE} holds for $f_+$ and $p-1$. Consider \eqref{eq 9}, after integrating it until time $t\le t_0$  we apply \eqref{jens + p=1} to the denominators and \eqref{f and b at t} to the remainder to obtain
		\begin{align}\label{01}
			\log\frac{\br{\Nf[p]{f_+}{0}{l}}_0}{\br{\Nf[p]{f_+}{t}{l}}_t}\le& \frac{pv'}  {s_l\left(\l_1r_l^de^{-1}\right)^p}\int_0^t{\br{\left(N_{f_+,\ta l}+1\right)^{p-2}N_{f'_+,\ta l}}_\ta}\rd \ta\notag\\
			&+\frac{p(\kappa-v')} {s_l}\int_0^t\frac{\br{\left(N_{f_+,\ta l}+1\right)^{p-1} N_{f'_+,\ta l}}_\ta}{\br{\Nf[p]{f_+}{\ta}{l}}_\ta}\mathrm{d} \ta\notag\\
			&+\frac{ 2pt\tilde C}{s_l^{1+\eps}\left(\l_1r_l^de^{-1}\right)^p}\sup_{0\le \ta\le t_0}{\hat{\cR}}_{p,l}(\ta).
		\end{align}
		Notice that again we made use of the property \eqref{Cf cond}.
		Since $t\le t_0$, applying \lemref{lem7.3} and \eqref{Pe- no s} yields
		\begin{align}\label{bound on f+ above}
			\br{\Nf[p]{f+}{t}{l}}_t\le\br{\Nf[p]{f-}{t}{l}}_t\le \del{\l_2R^d}^p \me^{C^{(p)}+1}.
		\end{align}
		Making use of \eqref{g and b 0}, \eqref{CD} and \eqref{bound on f+ above}, we can upper bound the l.h.s. of \eqref{01} by
		\begin{align}\label{02}
			\log\frac{\br{\Nf[p]{f_+}{0}{l}}_0}{\br{\Nf[p]{f_+}{t}{l}}_t}\ge \log  \del{\del{\frac{\l_1}{\l_2a^{d}}}^p\me^{-\del{C^{(p)}+1}}}.
		\end{align} 
		Since $a\le 4$, \eqref{02} implies 
		\begin{align}\label{03}
			\log\frac{\br{\Nf[p]{f_+}{0}{l}}_0}{\br{\Nf[p]{f_+}{t}{l}}_t}\ge C_3,
		\end{align}
		for some constant $C_3=C_3(C^{(p)}, p, \l_1,\l_2, v d)$.
		Now we can apply \eqref{bound on f+ above} to the denominator of the second summand appearing on the r.h.s. of \eqref{01},
		\begin{align}\label{04}
			\int_0^t\frac{\br{\left(N_{f_+,\ta l}+1\right)^{p-1} N_{f'_+,\ta l}}_\ta}{\br{\Nf[p]{f_+}{\ta}{l}}_\ta}\mathrm{d} \ta\ge \frac{\me^{-(C^{(p)}+1)}}{\del{\l_2 R_l^d}^p} \int_0^t{\br{\left(N_{f_+,\ta l}+1\right)^{p-1} N_{f'_+,\ta l}}_\ta}\mathrm{d} \ta.
		\end{align}
		Rearranging the terms in \eqref{01} and combining \eqref{03}, \eqref{04}, and \eqref{Pp=Pp+P1}, we derive
		\begin{align}\label{05}
			&\frac{p(v'-\kappa)\me^{-(C^{(p)}+1)}} {s_l\del{\l_2 R_l^d}^p}\int_0^t{\br{\left(N_{f_+,\ta l}+1\right)^{p-1} N_{f'_+,\ta l}}_\ta}\mathrm{d} \ta\\
			&\quad\le \frac{pv'}  {s_l\left(\l_1r_l^de^{-1}\right)^p}\int_0^t{\br{\left(N_{f_+,\ta l}+1\right)^{p-2}N_{f'_+,\ta l}}_\ta}\rd \ta\notag\\
			&\qquad+\frac{ tC_{p}}{s_l^{1+\eps}\left(\l_1r_l^de^{-1}\right)^p}\del{\sup_{0\le \ta\le t_0}{{\cR}}_{p,l}(\ta)+\sup_{0\le \ta\le t_0}{{\cR}}_{1,l}(\ta)}\notag\\
			&\qquad+C_3\notag.
		\end{align}
		Applying \eqref{IE} for $p-1$ and $f_+$, with constants $A^{(p-1)} \text{ and } B^{(p-1)}$, on the first term on the r.h.s. of \eqref{05} leads to
		\begin{align}\label{06}
			&\frac{p(v'-\kappa)\me^{-(C^{(p)}+1)}} {s_l\del{\l_2 R_l^d}^p}\int_0^t{\br{\left(N_{f_+,\ta l}+1\right)^{p-1} N_{f'_+,\ta l}}_\ta}\mathrm{d} \ta\notag\\
			&\quad\le\frac{pv'}  {\left(\l_1r_l^de^{-1}\right)^p}\del{A^{(p-1)}R_l^{d(p-1)}+\frac{B^{(p-1)}t}{s_l^{1+\eps}}\del{\sup_{0\le\ta\le t}\cR_{(p-1),l}(\ta)+\sup_{0\le\ta\le t}\cR_{1,l}(\ta)}}\notag\\
			&\qquad+\frac{ tC_{p}}{s_l^{1+\eps}\left(\l_1r_l^de^{-1}\right)^p}\del{\sup_{0\le \ta\le t_0}{{\cR}}_{p,l}(\ta)+\sup_{0\le \ta\le t_0}{{\cR}}_{1,l}(\ta)}\notag\\
			&\qquad +C_3.
		\end{align}
		
		We divide both sides of \eqref{06} by 
		\begin{align}
			\frac{p(v'-\kappa)\me^{-(C^{(p)}+1)}} {\del{\l_2 R_l^d}^p},\notag
		\end{align}
		apply \eqref{Pp and Pq}, and rename the constants to obtain
		\begin{align}
			s_l^{-1}&\int_0^t{\br{\left(N_{f_+,\ta l}+1\right)^{p-1} N_{f'_+,\ta l}}_\ta}\mathrm{d} \ta\notag\\
			&\quad\le R^{dp}\del{C_4+C_5\del{\frac{R}{r}}^{dp}}+ C_6\del{\frac{R}{r}}^{dp}\frac{t}{s_l^{\eps+1}}\del{\sup_{0\le \ta\le t_0}{{\cR}}_{p,l}(\ta)+\sup_{0\le \ta\le t_0}{{\cR}}_{1,l}(\ta)}.\notag
		\end{align}
		Notice that we also used the fact $R>1$. Since $a\le4$ we obtain the desired estimate with
		\begin{align}
			&A^{(p)}\coloneqq C_4+C_54^{dp},\notag\\
			&B^{(p)}\coloneqq C_64^{dp}.\notag
		\end{align}
		
		\subsubsection{Proof of \lemref{before bound rem}}\label{proof lemma}
		
		Step 1. We start by showing \eqref{4 p}  for $f_-$, where $f\in\ctE$. Consider \eqref{1 p} for $f_-$, and apply \eqref{IE} for $p-1$ to the first term on the r.h.s. of \eqref{1 p} to obtain
		\begin{align}
			\log\frac{\br{\Nf[p]{f_-}{t}{l}}_t}{\br{\Nf[p]{f_-}{0}{l}}_0}\le&\frac{pv'}{\left(\l_1r_l^de^{-1}\right)^p} \sbr{A^{(p-1)}R_l^{d(p-1)}+\frac{B^{(p-1)}t}{s_l^{1+\eps}}\del{\sup_{0\le\ta\le t}\cR_{p-1,l}(\ta)+\sup_{0\le\ta\le t}\cR_{1,l}(\ta)}}\notag\\
			&\qquad+\frac{ t\tilde C_p}{s_l^{1+\eps}\left(\l_1r_l^d\me^{-1}\right)^p}\del{\sup_{0\le\ta\le t}\cR_{p,l}(\ta)+\sup_{0\le\ta\le t}\cR_{1,l}(\ta)}\notag,
		\end{align}
		for any $R>r\ge \max\Set {\rho^{(p-1)},1}$ satisfying $R-r>\delta_0r$. Using \eqref{Pp and Pq}, the fact that $a\le 4$, and absorbing all the constants we arrive at
		\begin{align}\label{15 p}
			\log\frac{\br{\Nf[p]{f_-}{t}{l}}_t}{\br{\Nf[p]{f_-}{0}{l}}_0}\le& \frac{C^{(p)}}{R_l^d} +\frac{Ct}{s_l^{1+\eps}r_l^{dp}}\del{\sup_{0\le\ta\le t}P_{p,l}(\ta)+\sup_{0\le\ta\le t}P_{1,l}(\ta)},
		\end{align}
		for some positive constants $C^{(p)}=C^{(p)}(A^{(p-1)},p,\l_1,v,d)$  and $C=C(B^{(p-1)},p,\l_1)$. For any $\mu>0$, we define the following quantity,
		\begin{align}\label{16 p}
			\tilde \rho^{(p)}=\tilde \rho^{(p)}(B^{(p-1)},p,\l_1,\mu,\eps,\delta_0,v)\coloneqq \frac{3v}{2}\frac{1}{\delta_0}\del{C\mu}^{1/\eps}.
		\end{align}
		Recalling definition \eqref{slDef}, \eqref{16 p} implies that, for every $R>r\ge\tilde \rho^{(p)}$, the following holds
		\begin{align}\label{17 p}
			s_l^\eps\ge C\mu.
		\end{align}
		Inequalities \eqref{15 p} and \eqref{17 p} yield
		\begin{align}\notag
			\log\frac{\br{\Nf[p]{f_-}{t}{l}}_t}{\br{\Nf[p]{f_-}{0}{l}}_0}\le& \frac{C^{(p)}}{R_l^d} +\frac{t}{s_l\mu }\frac{1}{~r_l^{dp}}\del{\sup_{0\le\ta\le t}\cR_{p,l}(\ta)+\sup_{0\le\ta\le t}\cR_{1,l}(\ta)} ,
		\end{align}
		for any $R>r\ge \max\Set{\rho^{(p-1)},\tilde \rho^{(p)},1}$.
		Thanks to \propref{cor remainder}, there exists a constant $C_{\cR,1}=C_{\cR,1}(\al,d,C_J,\l_2/\l_1)>0$ and  $\rho$ is defined as in \eqref{Cdef2} and depending on $C_{\cR,1}$ such that for every
		\begin{align}\label{9 o}
			R>r\ge\max\Set{\rho,\tilde \rho^{(p)}, \rho^{(p-1)},1}\eqqcolon\rho^{(p)},
		\end{align} it holds 
		\begin{align}
			r^{-d}\sup_{0\le\ta\le t_l}\cR_{1,l}(\ta)\le C_{\cR,1}. \notag
		\end{align}
		Therefore,
		\begin{align}\label{18 p}
			\log\frac{\br{\Nf[p]{f_-}{t}{l}}_t}{\br{\Nf[p]{f_-}{0}{l}}_0}\le& \frac{C^{(p)}}{R_l^d} +\frac{t}{s_l\mu }\del{\frac{1}{~r_l^{dp}}\sup_{0\le\ta\le t}\cR_{p,l}(\ta)+C_{\cR,1}}.
		\end{align}
		Exponentiating \eqref{18 p} yields the desired estimate \eqref{4 p} for $f_-$.
		\begin{align}\label{19 p}
			\br{\Nf[p]{f_-}{t}{l}}_t\le \br{\Nf[p]{f_-}{0}{l}}_0\exp\Set{\frac{C^{(p)}}{R_l^d} +\frac{t}{s_l\mu }\del{\frac{1}{~r_l^{dp}}\sup_{0\le\ta\le t}\cR_{p,l}(\ta)+C_{\cR,1}}}.
		\end{align}  
		Step 2. Now, we prove \eqref{4 p}, for $f_+$. Consider \eqref{1 p} for $f_+$,  and bound the first summand on the r.h.s. using \eqref{IE} for $f_+$ and $p-1$,
		\begin{align}
			\log\frac{\br{\Nf[p]{f_+}{t}{l}}_0}{\br{\Nf[p]{f_+}{t}{l}}_t}\le&\frac{pv'}  {\left(\l_1r_l^de^{-1}\right)^p}\sbr{A^{(p-1)}R_l^{d(p-1)}+\frac{B^{(p-1)}t}{s_l^{1+\eps}}\del{\sup_{0\le\ta\le t}\cR_{p-1,l}(\ta)+\sup_{0\le\ta\le t}\cR_{1,l}(\ta)}}\notag\\
			&\qquad+\frac{ 2pt\tilde C}{s_l^{1+\eps}\left(\l_1r_l^de^{-1}\right)^p}\del{\sup_{0\le \ta\le t_0}{{\cR}}_{p,l}(\ta)+\sup_{0\le \ta\le t_0}{{\cR}}_{1,l}(\ta)}.\notag 
		\end{align}
		Following the same strategy as before, \eqref{4 p} for $f_+$ follows.

		\section{Proof of Remainder Estimate (\propref{Prop 2 ub 2}) via Multiscale Induction}\label{second mom log b}
		In this section we conclude the proof of \propref{final prop ind p} by proving \propref{Prop 2 ub 2}. The idea is to follow a similar multiscale induction scheme as in  \propref{cor remainder}. The additional tool that we need is Hölder inequality for commuting positive operators, that we state in the following lemma, that we prove in \secref{proof holder}.

		\begin{lemma}[Hölder inequality for positive commuting operators]\label{lemma holder}
			Consider a finite dimensional Hilbert space, $\mathscr{H}$, $\dim \mathscr{H}=n$, with scalar product $\br{\cdot\,,\cdot}$. For every positive operators $A$ and $B$ acting on $\mathscr{H}$, such that $\sbr{A,B}=0$, it holds that for every $p,q\in\del{1,\infty}$ such that $p^{-1}+q^{-1}=1$,
			\begin{align}
				\br{\psi,AB\psi}\le \br{\psi,A^p\psi}^{1/p}\br{\psi,B^q\psi}^{1/q} \qquad \text{for all }\psi\in\mathscr{H}.\notag
			\end{align}
		\end{lemma}

		\subsection{Proof of \propref{Prop 2 ub 2} assuming \lemref{lemma holder}}
		
		The next step consists in implementing the multiscale induction coupled with the 
		Hölder inequality in order to control the remainder term. Recall that in this section we are only considering initial states $\psi\in\cF_{\le N_0}\del{\ell^2\del{\Lambda}}$. Since it is a finite dimensional Hilbert space we can apply \lemref{lemma holder}.

		\begin{proof}
			We adapt the multiscale induction scheme as in \eqref{cor remainder}, and we refer to \cite{lemm2023microscopic} for the details.
			
			\textbf{Induction base.} We fix a large $L>0$ such that $\Lambda\subset B_L$. We want to show that there exists $M=M(L)>0$ such that \eqref{rem Pp bound} holds for all $l\ge M$. Recalling the definition of the remainder \eqref{Rp}, using \eqref{k2Cond}, the assumption of controlled density \eqref{CD}, and the fact that particle number is conserved we obtain the following
			\begin{align}
				r^{-dp}_l\cR_{,l}(\ta)\le&\frac{\br{N^p_\Lam}_\ta}{r_l^{dp}}+\sum_{x\in \Lambda}\sum_{y\in\Lambda}\abs{J_{xy}}|x-y|^{1+\eps}\frac{\br{ N^{p-1}_\Lam n_y}_\ta }{ r_l^{dp}}\notag  \\
				\le& \frac{\br{N^p_\Lam}_0}{r_l^{dp}}+\kappa_\eps\frac{\br{N^p_\Lam}_0}{r^{dp}_l}\notag\\
				\le& \l_2(1+\kappa_\eps)\frac{L^{dp}}{a^{dp(l+b)}},\notag
			\end{align}
			which leads to the desired estimate \eqref{rem Pp bound} for every $l\ge M\coloneqq\log_aL$. 
			
			\textbf{Induction step.} Assume \eqref{rem Pp bound} holds for $l+1,\,l+2,\dots$, we show it for $l$. Consider  \eqref{19 p} with $$\mu=C_{\cR,p}+C_{\cR,1},$$ 
			with $C_{\cR,1}$ as in \eqref{CPdef}. Then the induction hypothesis together with \eqref{4 p} implies
			\begin{align}\notag
				\br{\Nf[p]{f_-}{t}{j}}_t\le \br{\Nf[p]{f_-}{0}{j}}_0\exp\Set{\frac{C^{(p)}}{R_j^d} +\frac{t}{s_j}},
			\end{align}
			for all $t\le t_0$, $j\ge l+1$, and $r\ge \rho^{(p)}$ with $\rho^{(p)}$ as in \eqref{rhp p}. For the same range of parameters, thanks to the geometric properties of $f\in\ctE$, \eqref{f and b at t} and \eqref{f and B at t=0}, we derive
			\begin{align}\label{21 p}
				\br{N_{B_{{r}_j}}^{p}}_t\le \br{N_{B_{{R}_j}}^{p}}_0\exp\Set{\frac{C^{(p)}}{R_j^d} +\frac{t}{s_j}}.
			\end{align}
			Inequality \eqref{21 p} implies that there exists a constant $C_1>0$ that depends on $d,\l_1,\l_2,p$ such that
			\begin{align}\label{key estiamte rem p}
				r_j^{-dp}\sup_{0\le \ta\le t_0}{\br{N_{B_{{R}_j}}^{p}}_\ta}\le C_1,\qquad j\ge l.
			\end{align}
			In fact, by \eqref{21 p} it follows, for $j\ge l$,
			\begin{align}\label{ 22 p}
				r_j^{-dp}\sup_{0\le \ta\le t_0}{\br{N_{B_{{R}_j}}^{p}}_\ta}\le r_j^{-dp}{\br{N_{B_{R_{j+1}}}^{p}}_0}\exp\Set{\frac{C^{(p)}}{R_{j+1}^d} +\frac{t}{s_{j+1}}}.
			\end{align}
			Thanks to the assumption of controlled density \eqref{CD} and since $a\le 4$, inequality \eqref{ 22 p} yields
			\begin{align}
				r_j^{-dp}\sup_{0\le \ta\le t_0}{\br{N_{B_{{R}_j}}^{p}}_\ta}&\le \l_2a^{2dp}\exp\Set{\frac{C^{(p)}}{R_{j+1}^d} +\frac{t}{s_{j+1}}}\notag\\
				&\le \l_24^{2dp}\exp\Set{\frac{C^{(p)}}{R_{j+1}^d} +\frac{t}{s_{j+1}}}.\notag
			\end{align}
			Notice that $t\le t_0\le s_{j+1}$ and $R_{j+1}>1$ imply
			\begin{align}
				{r_j^{-dp}}\sup_{0\le \ta\le t_0}{\br{N_{B_{{R}_j}}^{p}}_\ta}&\le \l_24^{2dp}\me^{C^{(p)} +1}\eqqcolon  C_1,\notag
			\end{align}
			which shows \eqref{key estiamte rem p}. Estimate \eqref{key estiamte rem p} is the key to bound the $p$-th order remainder term defined as
			\begin{align}\label{31 p}
				& r_l^{-dp} \sup_{0\le \ta \le t_0}{\cR_{p,l}}\del{\ta}= \underbrace{r_l^{-dp}\sup_{0\le \ta \le t_0}\vphantom{\sum_{x\in B_{{{R}_l}}}}{\br{N_{B_{{R}_l}}^{p}}_\ta}}_{\Mr{I}}+\underbrace{r_l^{-dp}\sup_{0\le \ta \le t_0}\sum_{x\in B_{{{R}_l}}}\sum_{y\in\Lambda}\abs{J_{xy}}|x-y|^{1+\eps}{{\br{ N_{B_{{R}_l}}^{p-1}n_y}_\ta }}}_{\Mr{II}}.
			\end{align}
			In fact, thanks to \eqref{key estiamte rem p} for $j=l$, we can bound term I by
			\begin{align}\label{30 p}
				\Mr{I}\equiv r_l^{-dp}\sup_{0\le \ta\le t_0}{\br{N_{B_{{R}_l}}^{p}}_\ta}\le C_1.
			\end{align}
			Now we focus on II. We follow the same strategy as in \propref{cor remainder}, differentiating in II the contribution from close particles and from far away particles. Hence, we split the sum on $y$ over the two regions, $B_{R_{l+1}}=:B$ and $B^c_{R_{l+1}}.$ We differentiate two cases: 
			1) the set $\Set{y\in\Lambda\;:\;d(y,B)=1}$ is contained in $B^c_{R_{l+1}}$, and 2) the set $\Set{y\in\Lambda\;;\;d(y,B)=1}$ lies in between $B_{R_l}$ and $B^c_{R_{l+1}}$.

			2.1. In case 1) we write
			\begin{align}\label{22 p}
				\Mr{II}&=\sup_{0\le \ta \le t_0}\sum_{x\in B}\sum_{y\in\Lambda}\abs{J_{xy}}|x-y|^{1+\eps}\frac{{\br{ N_{B_{{R}_l}}^{p-1}n_y}_\ta }}{r_l^{dp}}\notag\\
				&\le \sup_{0\le \ta \le t_0}\sum_{x\in B}\sum_{y\in B}\abs{J_{xy}}|x-y|^{1+\eps}\frac{{\br{ N_{B_{{R}_l}}^{p-1}n_y}_\ta }}{r_l^{dp}}\notag\\
				&\quad + \sup_{0\le \ta \le t_0}\sum_{x\in B}\sum_{m>l+b+1}\sum_{a^m\le y<a^{m+1}}\abs{J_{xy}}|x-y|^{1+\eps}\frac{{\br{ N_{B_{{R}_l}}^{p-1}n_y}_\ta }}{r_l^{dp}}.
			\end{align}
			The first term on the r.h.s. of \eqref{ 22 p} can be bounded as follows
			\begin{align}\label{23 p}
				\Mr{II}_a:&=\sup_{0\le \ta \le t_0}\sum_{x\in B}\sum_{y\in B}\abs{J_{xy}}|x-y|^{1+\eps}\frac{{\br{ N_{B_{{R}_l}}^{p-1}n_y}_\ta }}{r_l^{dp}}\notag\\
				&\le \del{\sup_{x\in B}\sum_{y\in \Lambda}\abs{J_{xy}}|x-y|^{1+\eps}}\del{\sup_{0\le \ta\le t_0}\sum_{y\in B}\frac{{\br{ N_{B_{{R}_l}}^{p-1}n_y}_\ta }}{r_l^{dp}}}\notag \\
				&=\k_\eps\,I\le \k_\eps C_1.
			\end{align}
			Now fix 
			\begin{align}
				\g:= \al-1-\eps.\notag
			\end{align}
			Recalling \eqref{gaDef} and \eqref{epsDef} it follows
			\begin{align}
				\g>d+\eps.\notag
			\end{align}
			Then, following the same steps as in the proof of \cite[Prop.~5.3]{lemm2023microscopic} we can bound the second term on the r.h.s. of \eqref{22 p} as
			\begin{align}\label{24 p}
				\Mr{II}_b&:= \sup_{0\le \ta \le t_0}\sum_{x\in B}\sum_{m>l+b+1}\sum_{a^m\le y<a^{m+1}}\abs{J_{xy}}|x-y|^{1+\eps}\frac{{\br{ N_{B_{{R}_l}}^{p-1}n_y}_\ta }}{r_l^{dp}}\notag\\
				&\le C_J  \sup_{0\le \ta \le t_0}\sum_{x\in B}\sum_{m>l+b+1}\sum_{a^m\le y<a^{m+1}}|x-y|^{-\g}\frac{{\br{ N_{B_{{R}_l}}^{p-1}n_y}_\ta }}{r_l^{dp}}\notag\\
				&\le C_J C_d a^{d(l+b)}(a-1)^{-\g} \sup_{0\le \ta \le t_0}\sum_{m>l+b+1}a^{-(m-1)\g}\frac{\br{ N_{B_{{R}_l}}^{p-1}N_{B_{a^{m+1}}}}_\ta }{r_l^{dp}}\notag\\
				&\le C_J C_d \del{\frac{4}{\delta_0}}^\g\sup_{0\le \ta \le t_0}\sum_{m>l+b+1}a^{-m\g}\frac{\br{ N_{{R}_l}^{p-1}N_{B_{a^{m+1}}}}_\ta }{r_l^{d(p-1)}},
			\end{align}
			where we used the fact $1+\delta_0<1\le 4$.
			We can split the expectation value appearing in \eqref{24 p} thanks to \lemref{lemma holder} since $(p/(p-1))^{-1}+p^{-1}=1$, and we obtain
			\begin{align}\label{27 p}
				& \sup_{0\le \ta \le t_0}\sum_{m>l+b+1}a^{-m\g}\frac{{\br{ N_{B_{{R}_l}}^{p-1}N_{B_{a^{m+1}}}}_\ta }}{r_l^{d(p-1)}}\notag\\
				\le& \sum_{m>l+b+1}a^{-m\g}\sup_{0\le \ta \le t_0}\frac{\br{N_{B_{R_l}}^{p}}_t^{(p-1)/p}}{r_l^{d(p-1)}}\br{N_{B_{a^{m+1}}}^{p}}_t^{1/p}\notag\\
				\le& \sbr{r^{-d}_l\del{\sup_{0\le \ta \le t_0}\br{N_{B_{R_l}}^{p}}_t}^{1/p}}^{p-1}\sbr{~
					\sum_{m>l+b+1}a^{-m\g}\del{\sup_{0\le \ta \le t_0}\br{N_{B_{a^{m+1}}}^{p}}_t}^{1/p}}.
			\end{align}
			Thanks to \eqref{key estiamte rem p} we can bound
			\begin{align}\label{25 p}
				r_l^{-d}\del{\sup_{0\le \ta \le t_0}\br{N_{B_{R_l}}^{p}}_t}^{1/p}\le C_1^{1/p},
			\end{align}
			and
			\begin{align}\label{26 p}
				\sum_{m>l+b+1}a^{-m\g}\del{\sup_{0\le \ta \le t_0}\br{N_{B_{a^{m+1}}}^{p}}_t}^{1/p}\le  C_1^{1/p} \sum_{m>l+b+1}a^{-m(\g-d)}.
			\end{align}
			The sum in \eqref{26 p} converges since $\g>d$ and $a>\delta_0+1>1$. This together with \eqref{27 p}, \eqref{25 p}, and \eqref{26 p} implies
			\begin{align}\label{28 p}
				\Mr{II_b}\le C_J C_{d,\al} C_1.
			\end{align}

			2.2.	In case 2, so when the set $\Set{y\in \Lambda: d(y,B)=1}$ lies in between $B_{R_l} $ and $B_{R_{l+1}}$, we have an additional term
			\begin{align}\label{22 p1}
				\Mr{II}&=\sup_{0\le \ta \le t_0}\sum_{x\in B}\sum_{y\in\Lambda}\abs{J_{xy}}|x-y|^{1+\eps}\frac{{\br{ N_{B_{{R}_l}}^{p-1}n_y}_\ta }}{r_l^{dp}}\notag\\
				&\le \sup_{0\le \ta \le t_0}\sum_{x\in B}\sum_{y\in B}\abs{J_{xy}}|x-y|^{1+\eps}\frac{{\br{ N_{B_{{R}_l}}^{p-1}n_y}_\ta }}{r_l^{dp}}\notag\\
				&\quad + \sup_{0\le \ta \le t_0}\sum_{x\in B}\sum_{m>l+b+1}\sum_{a^m\le y<a^{m+1}}\abs{J_{xy}}|x-y|^{1+\eps}\frac{{\br{ N_{B_{{R}_l}}^{p-1}n_y}_\ta }}{r_l^{dp}}\\
				&\quad+\sup_{\ta<t_1}\sum_{x\in B}\sum_{\substack{y\in B^c\\ d(B,y)=1}}\abs{J_{xy}}|x-y|^{1+\eps}\frac{\br{ N_{B_{{R}_l}}^{p-1}n_y}_\ta }{r_l^{dp}}.
			\end{align}
			The first two summand on the r.h.s. of \eqref{22 p1} are the same terms we treated in case 1, so we can focus on the new term.
			\begin{align}
				\Mr{II_c}&:=\sup_{\ta<t_0}\sum_{x\in B}\sum_{\substack{y\in B^c\\ d(B,y)=1}}\abs{J_{xy}}|x-y|^{1+\eps}\frac{\br{ {N_{B_{{R}_l}}}^{p-1}n_y}_\ta }{r_l^{dp}}\notag\\
				&\le \k_\eps\sup_{\ta<t_0}\sum_{\substack{y\in B^c\\ d(B,y)=1}}\frac{\br{ {N_{B_{{R}_l}}}^{p-1}n_y}_\ta }{r_l^{dp}}\notag\\
				&\le \k_\eps\sup_{\ta<t_0}\frac{\br{ {N_{B_{{R}_l+1}}}^{p}}_\ta }{r_l^{dp}}\notag.
			\end{align}
			By the assumptions on case 2 we know $R_l+1\le R_{l+1}$. Then, writing $R_l=a^{b+l+1}, \,r_l=a^{b+l}$,applying \eqref{key estiamte rem p}, and recalling $a\le 4$ leads to
			\begin{align}\label{29 p1}
				\Mr{II_c}&\le \k_\eps\sup_{\ta<t_0}\frac{\br{ N_{B_{a^{b+l+2}}}^{p}}_\ta }{a^{d(b+l)}}\le \k_\eps 4^2C_1.
			\end{align}
			Applying \eqref{23 p}, \eqref{28 p}, and \eqref{29 p1} to \eqref{ 22 p} we obtain
			\begin{align}\label{29 p}
				\Mr{II}\le C_1\del{\k_\eps+C_JC_{d,\al}+\k_\eps 4^2}.
			\end{align}

			3.	We can conclude the bound on the remainder combining \eqref{31 p}, \eqref{30 p}, and \eqref{29 p} 
			\begin{align}
				r^{-dp}\sup_{\ta\le t}\cR_{p,l}(\ta)\le C_1\del{1+\k_\eps+C_JC_{d,\al}}=:C_{\cR,p}.\notag
			\end{align}
			This completes the induction in $l\in\N_0$.
		\end{proof}

		\bigskip
		\subsection{Proof of \lemref{lemma holder}}\label{proof holder}
		
		Since $A$ and $B$ commute, there exists a common orthonormal basis of eigenvectors $\Set{\phi_\al}_{\al=1}^n$ and relative eigenvalues, $\Set{\l_\al},\,\Set{\mu_\al}$, with $\l_a,\mu_\al\ge0$, such that
		\begin{align}
			A\phi_\al=\l_\al\phi_\al\qquad B\phi_\al=\mu_\al\phi_\al.\notag
		\end{align}
		Given a general $\psi\in\mathscr{H}$ it can be written w.r.t. this basis as
		\begin{align}
			\psi=\sum_{\al=1}^nc_\al\phi_\al\qquad \text{where }c_\al\coloneqq\br{\psi,\phi_\al}.\notag
		\end{align}
		This way we can write
		\begin{align}
			\br{\psi,AB\psi}&= \sum_{\al,\beta=1}^n c_\al^*c_\beta \br{\phi_\al, AB\phi_\beta}\notag\\
			&=\sum_{\al,\beta=1}^n c_\al^*c_\beta \mu_\beta\l_\beta\br{\phi_\al,\phi_\beta}\notag\\
			&=\sum_{\al=1}^n \abs{c_\al}^2 \mu_\al\l_\al.\notag
		\end{align}
		Since $p^{-1}+q^{-1}=1$, we can write $2=2/p+2/q$, using Hölder inequality for numbers, and the fact that $\mu_\al,\l_\al\ge0$, we obtain 
		\begin{align}
			\sum_{\al=1}^n \abs{c_\al}^2 \mu_\al\l_\al&=\sum_{\al=1}^n \del{\abs{c_\al}^{2/p} \mu_\al}\del{\abs{c_\al}^{2/q} \l_\al}\notag\\
			&\le \del{\sum_{\al=1}^n\abs{c_\al}^{2} \mu^p_\al}^{1/p}\del{\sum_{\beta=1}^n\abs{c_\beta}^{2} \l^q_\al}^{1/q}\notag\\
			&=\br{\psi,A^p\psi}^{1/p}\br{\psi,B^q\psi}^{1/q}.\notag
		\end{align}
		\qed

		\section{Proof of \thmref{thm1} on Light-cone Approximation}\label{secPfLCAgen}

		To prove \thmref{thm1}, we will repeatedly use the propagation bounds obtained in our previous work \cite{lemm2023microscopic},  which are reproduced  below for the readers' convenience. Recall $n$, $\beta$ are defined in \eqref{betaDef}. 		We abbreviate $C_H\equiv C_{H,\al}$, with 
		\begin{equation}
			C_{H,\al}:=C_{J,\al}+C_{V,\al}\label{CHdef}.
		\end{equation}
		\begin{theorem}[\cite{lemm2023microscopic}*{Thm.~2.1}]\label{thmTSMVB}
			{Suppose the assumptions of Theorem \ref{thm1} hold.}
			Then, for any ${v} > \kappa$,
			{ $\delta_0>0$ and $p=1,2$, there exists ${C=C(\al,d, C_H, v,\delta_0,p)}>0$ such that for all $\l,\,R,\,r>0$ with $R-r>\max(\delta_0r,1)$,} there holds
			\begin{align}\label{MVBp}
				\sup_{0\le{t}<(R-r)/v}\br{N_{B_r}^p}_t\le {\left(1+ C(R-r)^{-1}\right)} {{\br{N_{B_{R}}^p}_0}} + {C(R-r)^{-\beta}\l^p.}
			\end{align}
		\end{theorem}
		
		The remainder of this section is structured as follows.
		In \secref{sec71},  we use \thmref{thmTSMVB} to derive   propagation estimates on annular regions  that are later used in the proof of \thmref{thm1}. In \secref{sec72}, we isolate a key estimate in \propref{propClaim} that constitutes the basis of the proof of \thmref{thm1}. This proposition is proved by induction in Sects.~\ref{sec73}--\ref{sec74}, assuming some technical lemmas, whose proofs are deferred to \secref{secPfLems}. Finally, in \secref{sec75} we complete the proof of \thmref{thm1}.

		\subsection{Preliminaries}\label{sec71}
		First, we extend \thmref{thmTSMVB} via a covering argument to  curved annular regions  (c.f.~\cite[Thm.~2.3]{faupin2022lieb}). 
		This will later facilitate the study of propagation properties of particles starting from arbitrary bounded domain  $X\subset\Lam$.
		
		To this end, we introduce some notations.
		For $\xi\ge1$ and $0\le \g\le1$, set% 
		\begin{align}
			\label{NgxiDef}
			N_{\g,\xi} := N_{X_{(1+\g)\xi}}-N_{X_{(1-\g)\xi}}.
		\end{align}			
		$N_{\g,\xi}$ corresponds to the particle numbers associated to the annular region  of width $2\g\xi$ around the light-cone boundary, $\di X_\xi:=\Set{x\in\Lam:d_X(x)=\xi}$.
		The corresponding regions are illustrated in \figref{figAnnDecomp} below.
		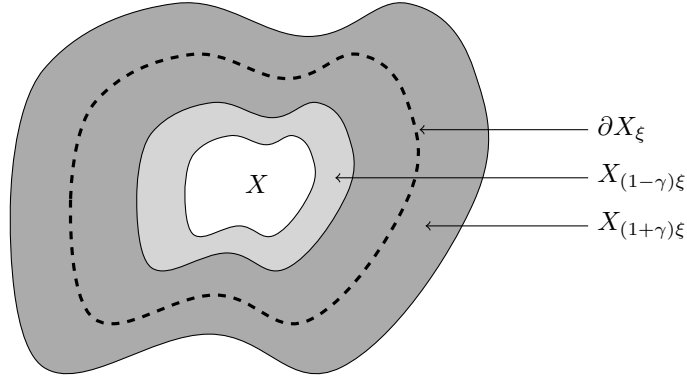
\begin{figure}[H]
			\centering
			\begin{tikzpicture}[scale=.8]
				\draw[fill=gray!66]  plot[shift={(-0.3,-0.45)}, scale=1.1,smooth, tension=.7] coordinates {(-3,0.5) (-2.5,2.5) (-.5,3.5) (1.5,3) (3,3.5) (4,2.5) (4,0.5) (2,-2) (0,-1.5) (-2.5,-2) (-3,0.5)};
				
				\draw[dashed,very thick]  plot[shift={(-0.2,-0.25)}, scale=.8,smooth, tension=.7] coordinates {(-3,0.5) (-2.5,2.5) (-.5,3.5) (1.5,3) (3,3.5) (4,2.5) (4,0.5) (2,-2) (0,-1.5) (-2.5,-2) (-3,0.5)};
				
				\draw[fill=gray!33]  plot[scale=.5,smooth, tension=.7] coordinates {(-3,0.5) (-2.5,2.5) (-.5,3.5) (1.5,3) (3,3.5) (4,2.5) (4,0.5) (2,-2) (0,-1.5) (-2.5,-2) (-3,0.5)};
				
				\draw[fill=white]  plot[shift={(0.2,.15)},scale=.3,smooth, tension=.7] coordinates {(-3,0.5) (-2.5,2.5) (-.5,3.5) (1.5,3) (3,3.5) (4,2.5) (4,0.5) (2,-2) (0,-1.5) (-2.5,-2) (-3,0.5)};
				
				\node at (.5,.4) {$X$};

				\draw [->] (6,.5)--(1.8,.5);
				\node [right] at (6,.5) {$X_{(1-\g)\xi}$};
				
				\draw [->] (6,-.3)--(3.3,-.3);
				\node [right] at (6,-0.3) {$X_{(1+\g)\xi}$};
				
				\draw [->] (6,1.3)--(3.2,1.3);
				\node [right] at (6,1.3) {$\di X_\xi$};
			\end{tikzpicture}
			\caption{Schematic diagram illustrating the curved annular region in \eqref{NgxiDef}.}
			\label{figAnnDecomp}
		\end{figure}
		We prove the following: 	 	
		\begin{theorem}[MVB on curved annular region]\label{thmAnnMVB}
			Suppose the assumptions of \thmref{thmTSMVB} hold. 
			Then, for any $v>\kappa$, $p=1,2$, $1\ge \g_2>\g_1\ge0$, $\xi>(\g_2-\g_1)^{-1}$,  and bounded subset $X\subset\Lam$,
			there exists positive constant 
			$$ C=C(\al,d, C_H, v,p,\g_1,\g_2,X),$$ 
			such that the solution $\psi_t$ satisfies, 
			\begin{align}\label{annEst}
				\sup_{0\le t<(\g_2-\g_1)\xi/v}	\br{N_{\g_1,\xi}^p}_t
				\le C\Set{\br{N_{\g_2,\xi}^p}_0+C_2[(\g_2-\g_1)\xi]^{-\beta} \l^p}.
			\end{align}
		\end{theorem}

		\begin{proof}
			Step 1. 	 Fix bounded $X\subset\Lam$ and write $X_{\g,\xi}:={X_{(1+ \g)\xi}}\setminus X_{(1-\g)\xi}$ for $0\le \g\le 1$, $\xi\ge1$. We first show that  there exist positive integers $M,\, D$ \textit{independent of $\xi$} s.th.~ $X_{\g,\xi}$ can be covered with $M$ balls in $\Rb^d$, with each point in $X_{\g,\xi}$ contained in at most $D$ balls.

			Take  $\eps>0$ to be determined later, and let $\tilde \g = (1+\eps)\g.$ Denote by $\tilde Y=\mathrm{Conv}\,(X)$, the convex hull of $X$, and $Y_{\g,\xi}:=\Set{x\in\Rb^d: (1-\g)\xi<\dist_Y(x)\le (1+\g)\xi}$. Since $X$ is bounded, $Y\subset \Rb^d$ is compact, and  there exist positive integers
			\begin{align}
				\label{Pdef}
				D=D(\g,\eps)\ge1,\quad	M=M(X,\g,\eps)\ge1
			\end{align} 
			and balls $B^1_{\tilde\g},\ldots, B^M_{\tilde\g}\subset Y_{\tilde\g, 1}$  s.th. % 
			\begin{align}
				\label{coverCond1}
				Y_{\tilde \g,1}	\subset\bigcup_{j=1}^MB^j_{\tilde\g},\qquad B^{j_1}_{\tilde\g}\cap \ldots\cap B^{j_{D+1}}_{\tilde\g}=\emptyset\text{ for distinct }1\le j_1,\ldots,j_{D+1}\le M.
			\end{align}

			Denote by $g(x)$ the map sending $x$ to its closest point on $\di Y$. Then, for any $\xi\ge1$, the map $x\mapsto g(x)+\xi(x-g(x))$ is a surjection from $Y_{\g,1}$ onto $Y_{\g,\xi}$. It follows that for any $\xi\ge1$, a suitable rescaling of the choice from \eqref{coverCond1} yields a cover 
			$B^1_{\tilde\g,\xi},\ldots, B^M_{\tilde\g,\xi}\subset Y_{\tilde\g,\xi}$
			for the curved annular region  $Y_{\g,\xi}$, satisfying
			\begin{align}
				\label{cover1}
				& Y_{\g,\xi}	\subset\bigcup_{j=1}^M B^j_{\tilde\g,\xi},\\ \label{coverCond'}&B^{j_1}_{\tilde\g,\xi}\cap \ldots\cap B^{j_{D+1}}_{\tilde\g,\xi}=\emptyset\text{ for distinct }1\le j_1,\ldots,j_{D+1}\le M.
			\end{align}  
			Since $X_{\g,\xi}\subset Y_{\g,\xi}$, it follows that one can cover all the lattice points in the annular region $X_{\g,\xi}$ with the same number of balls  independent of $\xi$, as desired.
			
			Step 2.  Given $v>\kappa$ and $1\ge\g_2>\g_1\ge0$, we choose $\eps>0$ small enough so that $\g_2>\tilde \g_1:=(1+\eps)\g_1$ and $\tilde v :=\frac{\g_2-\tilde \g_1}{\g_2-\g_1} v>\kappa$. 
			Then we apply \thmref{thmTSMVB} with this choice of $\tilde v$,  $\delta_0:=\frac{\g_2}{\tilde \g_1}-1>0$, and 
			$(R,r):=( \g_2\xi, \tilde \g_1\xi)$  to obtain some constant $ C_0=C_0(\al,d, C_H, v,\delta_0,p)>0$  such that
			\begin{align}
				\label{annEst1}
				\sup_{0\le{t}<(\g_2-\tilde \g_1)\xi/v}\br{N_{B^j_{\tilde \g_1,\xi}}^p}_t\le {\left(1+ C_0[(\g_2-\tilde\g_1)\xi]^{-1}\right)} {{\br{N_{B^j_{\g_2, \xi}}^p}_0}} + {C_0[(\g_2-\tilde\g_1)\xi]^{-\beta}\l^p,}  	
			\end{align}
			where  $B^1_{\al,\xi},\ldots, B^M_{\al,\xi}\subset {X_{(1+\al)\xi}}\setminus X_{(1-\al ) \xi}$ for $\al>0$.
			
			Since $N_S\ge0$ for any $S\subset\Lam$, by the covering property \eqref{cover1} and generalized mean inequality, for $p\ge1$  we have
			$$N_{\g_1,\xi}^p \le\del{\sum_{j=1}^M N_{B^j_{\tilde\g_1,\xi}}}^p\le M^{p-1}\sum_{j=1}^MN_{B^j_{\tilde\g_1,\xi}}^p.$$
			This fact, together with \eqref{annEst1}, yields
			\begin{align}
				\label{annEst2}
				&\sup_{0\le{t}<(\g_2-\tilde\g_1)\xi/v}	\br{N_{\tilde\g_1,\xi}^p}_t\notag\\\le& \sup_{0\le{t}<(\g_2-\tilde\g_1)\xi/v}M^{p-1}\sum_{j=1}^M\br{N_{B^j_{\tilde\g_1,\xi}}^p}_t\notag
				\\\le&M^{p-1}\Set{\left(1+ C_0[(\g_2-\tilde\g_1)\xi]^{-1}\right)  {{\sum_{j=1}^M\br{N_{B^j_{\g_2,\xi}}^p}_0}} + {{C_0M}[(\g_2-\tilde\g_1)\xi]^{-\beta}\l^p}} .
			\end{align}
			Owing to the positivity $N_S\ge0$,  the non-intersection condition \eqref{coverCond'}, and the fact that $B^j_{\g_2,\xi}\subset X_{(1+\g_2)\xi}\setminus X_{(1-\g_2)\xi}$, we have
			% $$\bigcup_{j=1}^M B^j_{(\g_2-\g_1)\xi}	 \subset	{X_{(1+\g_2)\xi}}\setminus{X_{(1-\g_2)\xi}},$$
			%and so
			$$\sum_{j=1}^M N_{B^j_{\g_2,\xi}}^p\le\del{\sum_{j=1}^M N_{B^j_{\g_2,\xi}}}^p\le  D^p N_{\g_2,\xi}^p.$$
			Using this relation to bound the sum in line \eqref{annEst2}, we find
			\begin{align}
				\label{annEst3}
				&\sup_{0\le{t}<(\g_2-\tilde\g_1)\xi/v}	\br{N_{\tilde\g_1,\xi}^p}_t\notag
				\\\le&M^{p-1}\Set{D^p{\left(1+ C_0[(\g_2-\tilde\g_1)\xi]^{-1}\right)} {{\br{N_{\g_2,\xi}^p}_0}} + {C_0M[(\g_2-\tilde\g_1)\xi]^{-\beta}\l^p}  }	.
			\end{align}
			
			Step  3. Finally, set
			$$C_1:= M^{p-1}D^p(1+C_0),\qquad C_2:= M^{p}C_0,$$
			and choose $\eps>0$ small so that $\xi (\g_2-\tilde \g_1)\ge1$. Recalling  that \textit{$D$ and $M$ in \eqref{Pdef} are independent of  $\xi$} (which is shown in the construction in Step 1), 
			we conclude the desired estimate \eqref{annEst} from \eqref{annEst3}. This completes the proof of Theorem \ref{thmAnnMVB}.
		\end{proof}
		
		\begin{remark}
			A straightforward modification of the proof above as in \cite{sigal2022propagation,lemm2023information} yields similar propagation estimates on a general bounded domain,
			\begin{align}\notag
				\sup_{0\le t<\xi/v}	\br{N_{X}^p}_t
				\le C_X ({\br{N_{X_\xi}^p}_0+\xi^{-\beta}\l^p}).
			\end{align}
			
		\end{remark}

		\subsection{Key estimate}\label{sec72}
		Recall the notations  $X^\cp:=\Lam\setminus X$, $X_\xi\equiv \Set{x\in\Lam:d_X(x)\le \xi}$ for $\xi\ge0$,  $X_\xi^\cp\equiv (X_\xi)^\cp$ (see \figref{fig:Xxi}), and that an observable (i.e.~bounded operator) $A$ is said to be localized in $X\subset\Lambda$ if \eqref{A-loc} holds, 
		written as $A\in\cB_X$. % for an observable localized in $X$, 
		To simplify notations, in what follows we write the localized evolution to $X_\xi$ as
		\begin{equation}\label{Axi}
			A_s^\xi\equiv \al_s^{X_\xi}(A)=e^{isH_{X_\xi}}Ae^{-isH_{X_\xi}},
		\end{equation} 
		where, recall, $H_S,\,S\subset \Lam$ is the localized Hamiltonian defined by \eqref{H}  with $S$ in place of $\Lam$. With the notation above, it is easily verified using  definitions \eqref{A-loc} and \eqref{BX} that   % 
		\begin{align}
			\label{AlocEvol}
			A_s^\xi\in \cB_{X_\xi}\quad  \text{ for } \quad A\in\cB_X, s\in\Rb,\,\xi\ge0.	
		\end{align}

		Let $A_t=\al_t(A)$ be the full evolution \eqref{1.5} and define the Heisenberg derivative % of $A(t)$, 
		\begin{align}&D A_t=\frac{\partial}{\partial t}A_t+i[H_\Lam, A_t],\notag
		\end{align}
		so that
		\begin{align}	\label{dt-Heis}
			&\di_t\al_t(A_t) =\al_t(DA_t).
		\end{align}
		By the time-reversal symmetry, without loss of generality we take $t\ge0$ in the sequel. By the fundamental theorem of calculus, we have$$	A_t-A_t^\xi=\int_0^t \di_r\al_r(\al_{t-r}^{X_\xi}(A))\,dr.$$	 Using identity \eqref{dt-Heis} for $\al_r$ and $\al_{t-r}^{X_\xi}$ in the integrand above, together with the fact that $\al_{t-r}^{X_\xi}([H_{X_\xi},A])=[H_{X_\xi},\al_{t-r}^{X_\xi}(A)]$, we find 
		\begin{equation}\label{410}
			A_t-A_t^\xi=\int_0^t\al_r(i[R',A_{t-r}^\xi])\,dr,
		\end{equation}
		%checked 11.12, cf B11
		where $R':=H-H_{X_\xi}$. Since $A_s^\xi$ is localized in $X_\xi$, only terms in $R'$ which connect $X_\xi$ and $X_\xi^\cp$ contribute to $[R',A_{t-r}^\xi]$ (see \figref{fig:splitting}). 
		Our goal now is to control this commutator, whence the desired estimate in \thmref{thm1} follow by integration.
		\begin{figure}[H]
			\centering
			\begin{tikzpicture}[scale=1.2]
				\draw  plot[scale=.5,smooth, tension=.7] coordinates {(-3,0.5) (-2.5,2.5) (-.5,3.5) (1.5,3) (3,3.5) (4,2.5) (4,0.5) (2,-2) (0,-1.5) (-2.5,-2) (-3,0.5)};

				% 		\draw  plot[shift={(-0.2,-0.25)}, scale=.8,smooth, tension=.7] coordinates {(-3,0.5) (-2.5,2.5) (-.5,3.5) (1.5,3) (3,3.5) (4,2.5) (4,0.5) (2,-2) (0,-1.5) (-2.5,-2) (-3,0.5)};
				
				\node [right] at (-.5,.5) {$X_\xi$};
				
				\draw [->] (4,.5)--(1.3,.5);
				\node [right] at (4,.5) {$\text{contributing to }H_{X_\xi}$};
				
				\draw [thick] (1,.7) -- (1.4,.3);
				\draw [fill] (1,.7) circle [radius=0.05];
				\draw [fill] (1.4,.3) circle [radius=0.05];
				
				\draw [->] (4,0)--(2.3,0);
				\node [right] at (4,0) {$\text{contributing to }H_{X_\xi^\cp}$};
				
				\draw [thick] (2.05,.-.2) -- (2.4,.3);
				\draw [fill] (2.05,.-.2) circle [radius=0.05];
				\draw [fill] (2.4,.3) circle [radius=0.05];
				
				\draw [->] (4,-0.5)--(1.2,-0.5);
				\node [right] at (4,-0.5) {$\text{connecting ${X_\xi}$ and ${X_\xi^\cp}$}$};
				
				\draw [thick] (1,-0.3) -- (1.3,-1);
				\draw [fill] (1,-0.3) circle [radius=0.05];
				\draw [fill] ((1.3,-1) circle [radius=0.05];

			\end{tikzpicture}
			\caption{Schematic diagram illustrating the splitting of $H$.}
			\label{fig:splitting}
		\end{figure}
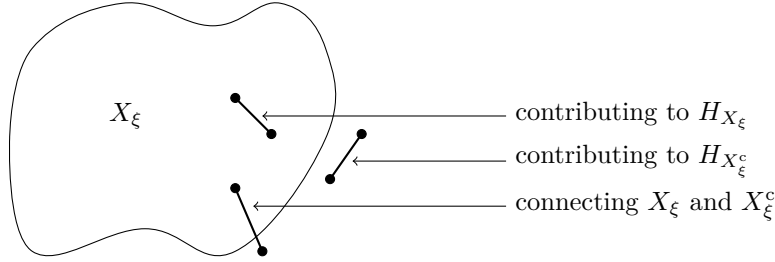

		Define the remainder operator 
		\begin{align}
			\label{Rem-def}
			\Rem_t\equiv \Rem_t(A):=A_t-A_t^\xi.
		\end{align} 
		The proof of \thmref{thm1} relies on the following key estimate on $\undel{\Rem_t}$:
		
		\begin{proposition}\label{propClaim}
			Let the assumptions of \thmref{thm1} hold. 
			
			Then, for every $v>2\kappa$ and bounded subset $X\subset \Lam$, there exists positive constant
			\begin{align}
				\label{C1C2}
				{C=C(\al,d, C_H, v,\l,X)},	
			\end{align}
			s.th.~for all $\xi\ge\max(2,\diam X)$ and      operator   $A\in \cB_X$, 
			we have, for all $0\le  {t}< \xi/v$ and states $\varphi,\,\psi$ satisfying \eqref{g0-cond}, \eqref{UDBp},
			\begin{align}
				\label{remEstMain}
				\abs{\undel{\Rem_t}}\le&  C{t } \norm{A}\tau(\varphi)^{1/2}\tau(\psi)^{1/2},\\
				\notag
				\tau(\phi):=& \phidel{N_{1,\xi}N_\Lam }+ \xi^{-\beta} .
			\end{align}
		\end{proposition}
		{In the next two subsections, we prove this proposition by a down scale induction on the size of $\xi$.}

		\subsection{Proof of \propref{propClaim}: Base case} \label{sec73}Take a large constant $L>0$ s.th.~the domain $\Lam\subset [-L/2,L/2]^d$. 
		%Let $C_0>0$ be some constant independent of $a,\,L,\,k$ to be determined later.  
		For the base case, we prove that there exists a large integer $k_0=k_0(L)>0$ s.th.\ \eqref{remEstMain} holds for all $\xi\ge 2^k$ with $k\ge k_0$. 
		
		Under the decay assumptions \eqref{CJdef}--\eqref{CVdef}, for any integer	$n\le\al-d-1$,
		there exists $c>0$ depending only on $\al$, $d$, and $C_{H,\al}$ (see \eqref{CHdef}) such that
		\begin{align}
			{		\kappa_\nu:=\sup_{x\in\Lam}\sum_{y\in\Lam} \del{\abs{J_{xy}}+ \abs{V _{xy}} } \abs{x-y}^{\nu+1}\le c\qquad (\nu=0,\ldots,n).\label{Kcond}}
		\end{align}
		Recall also the definition for $N_{\g,\xi}$ in \eqref{NgxiDef}. We have the following lemma
		\begin{lemma}
			\label{lemNq}
			Suppose \eqref{Kcond} holds with some $n\ge1$.  Then, for any $v > \kappa$, $1\ge \g_2>\g_1\ge0$,
			there exists $C=C(n,\kappa_n,v,\g_1,\g_2)>0$ s.th.~for every $q\ge0$, $\xi\ge1 $ and two subsets $X,\,S\subset\Lam$ with $X_{2\xi}\subset S$,
			we have the following estimate for all $0\le {t}< \xi/v$:
			\begin{align}\label{MVE-nuq}
				\al_t^S(N_{\g_1,\xi}^{q+1})\le &C
				\del{N_{\g_2,\xi}N^{q}_S +\xi^{-n} N^{q+1}_S},
			\end{align}
			where $\al_t^S$ is the localized evolution defined in \eqref{alloc}.
		\end{lemma}
		
		\begin{proof}[Proof of \lemref{lemNq}]
			We use \cite[Thm.~1, part (i)]{lemm2023information} with $\al^S_t(\cdot)$ in place of $\al_t(\cdot)$, which is possible because the local Hamiltonian $H_S$ also satisfies \eqref{Kcond} with the same $n\ge1$ as in the assumption. The resulting estimate reads, for any subset $Y$ and $\eta>0$ with $Y_\eta\subset S$, 
			\begin{align}
				\label{MVES}
				\al_t^S(N_Y)\le C (N_{Y_\eta}+\eta^{-n}N_S).
			\end{align}
			Choose now $Y:=X_{(1+\g_1)\xi}\setminus X_{(1-\g_1)\xi}$ and $\eta:=(\g_2-\g_1)\xi$.
			Using	\eqref{MVES}, together with the relations $N_Y\equiv N_{\g_1,\xi}$, $N_{Y_\eta}\equiv N_{\g_2,\xi}$ (see \figref{figAnnDecomp}), $N_{\g,\xi}^{q+1}\le N_{\g,\xi}N^{q}$, and $[N_S,H_S]=0$, we arrive at  
			$$
			\begin{aligned} 
				\al_t^S\del{N_{\g_1,\xi}^{q+1}}\le &\al_t^S\del{N_{\g_1,\xi} N_S^{q}}
				\\= & \al_t^S\del{N_{\g_1,\xi} }N_S^{q}
				\\\le&	C \del{N_{\g_2,\xi} +[(\g_2-\g_1)\xi]^{-n} N_S}N_S^q.
			\end{aligned}
			$$
			This gives \eqref{MVE-nuq}. 
		\end{proof}

		Owing to condition \eqref{Kcond}, we can apply \lemref{lemNq} with $S\equiv \Lam$ and proceed exactly as in the proof of \cite[Eq.~(D.51)]{sigal2022propagation}
		to obtain, for some $C_0>0$\textit{ independent of $k,\,\alpha,\,L$},
		\begin{align}
			\label{remEstMain'}
			\abs{\undel{\Rem_t}}\le&   C_0t  \norm{A}\tilde{\tau}(\varphi)^{1/2}\tilde{\tau}(\psi)^{1/2},\\
			\notag
			\tilde{\tau}(\phi):=& {\phidel{N_{1,\xi}N_\Lam }+\xi^{-n} \phidel{N_\Lam^2} }.
		\end{align}
		Under assumption \eqref{UDBp} and $\abs{\Lam}\le L^d$, we have $\br{N_\Lam^p}_0\le \l^p L^{dp}$, and so $\tilde{\tau}$ verifies
		\begin{align}\label{propEst4}
			\tilde{\tau}(\phi)\le \phidel{N_{1,\xi}N_\Lam } + {2^{-kn} L^{2d }\l^2}.
		\end{align}
		We conclude from \eqref{remEstMain'}--\eqref{propEst4} that \eqref{remEstMain} holds for $C\ge C_0$ and all $k\ge k_0:= \log_2(L).$
		This completes the proof of \propref{propClaim} for the base case.

		\subsection{Proof of \propref{propClaim}:  Induction step} \label{sec74}Next, 
		assuming \eqref{remEstMain} holds for  all $\xi$ with $\xi\ge 2^{k+1}$, 
		we prove it  for  $2^{k+1}>\xi\ge 2^k$.

		{We write
			\begin{align}
				\notag
				\Rem_t\equiv\Rem_t^{(1)}+\Rem_t^{(2)},\quad  \Rem_t^{(1)}(A)=A_t-A_t^{2\xi},\quad\Rem_t^{(2)}(A)=A_t^{2\xi}-A_t^\xi.
			\end{align}
			By the induction hypothesis, $\Rem_t^{(1)}$ verifies estimate \eqref{remEstMain} in place of $\Rem_t$. Thus it remains to prove the same estimate for $\Rem_t^{(2)}$. See \figref{figInd} for an illustration.}

		\begin{figure}[H]
			\begin{tikzpicture}[scale=.8]
				\draw  plot[scale=.5,smooth, tension=.7] coordinates {(-3,0.5) (-2.5,2.5) (-.5,3.5) (1.5,3) (3,3.5) (4,2.5) (4,0.5) (2,-2) (0,-1.5) (-2.5,-2) (-3,0.5)};

				\draw  plot[shift={(-0.2,-0.25)}, scale=.8,smooth, tension=.7] coordinates {(-3,0.5) (-2.5,2.5) (-.5,3.5) (1.5,3) (3,3.5) (4,2.5) (4,0.5) (2,-2) (0,-1.5) (-2.5,-2) (-3,0.5)};
				
				\node [right] at (-.5,.5) {$X_\xi$};
				\node [right] at (-.5-1-.1-.1-.1-.1-.1,.5+1+.1+.1) {$X_{2\xi}$};

				\draw [->] (4,.5)--(1.3-1-.2-1,.5-1+.2);
				\draw [->] (4,.5)--(1.31,.5+.2+1+.2+.1);
				\node [right] at (4,.5) {$\text{handled by the induction hypothesis}$};
				
				\draw [thick] (-1,.4) -- (-1.4,-2.3);
				\draw [fill] (-1,.4) circle [radius=0.05];
				\draw [fill] (-1.4,-2.3) circle [radius=0.05];

				\draw [thick] (2.05-1,.-.2+1+1) -- (2.4-1,.3+1+1.2);
				\draw [fill] (2.05-1,.-.2+1+1) circle [radius=0.05];
				\draw [fill] (2.4-1,.3+1+1.2) circle [radius=0.05];
				
				\draw [->] (4,-0.5)--(1.2,-0.5);
				\node [right] at (4,-0.5) {$\text{handled in the induction step (to be proved)}$};
				
				\draw [thick] (1,-0.3) -- (1.3,-1);
				\draw [fill] (1,-0.3) circle [radius=0.05];
				\draw [fill] ((1.3,-1) circle [radius=0.05];

			\end{tikzpicture}
			\caption{Schematic diagram illustrating the induction scheme.}
			\label{figInd}
		\end{figure}
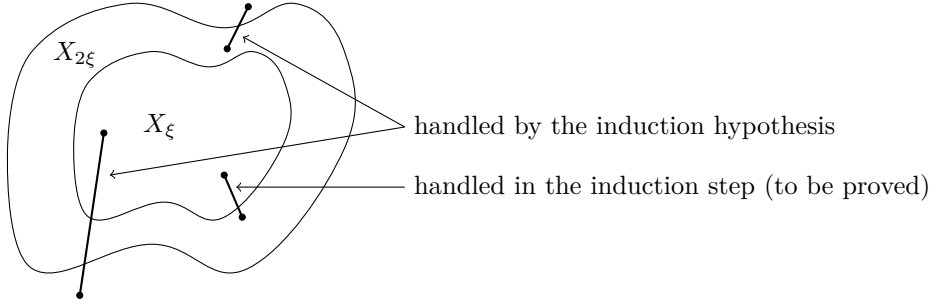

		1.	Write $X_{a,b}=X_b\setminus X_a$ for $b>a\ge0$ and introduce the coupling operator  % 
		\begin{align}
			\label{Rdef}
			R:=H_{X_{2\xi}}-H_{X_\xi}-H_{X_{\xi,2\xi}}.
		\end{align}
		By these definitions and relation \eqref{AlocEvol}, we find % 
		\begin{align}
			\label{HAcomm1}
			[H_{X_{\xi,2\xi}},A^\xi_s]=0 \qquad
			(s\in\Rb,\xi>0).	
		\end{align} 
		Combining \eqref{Rdef}--\eqref{HAcomm1} and following the derivation of \eqref{410}, we find that
		\begin{equation}\label{4.41}
			\begin{aligned}
				\Rem_t^{(2)} 
				= & \int_0^t \al_r\del{i\sbr{R,A^\xi_{t-r}}}\,dr.
			\end{aligned}
		\end{equation}
		To control the integrand in \eqref{4.41}, we split the coupling term $R$ as follows:
		\begin{align}
			\label{5.12}
			R:=& S+W, &&\\
			\label{5.10}
			S:=&S'+(S')^*,& S':=& \sum_{x\in X_\xi,y\in X_{\xi,2\xi}}J_{xy}a_x^*a_y,\\
			W:= & \sum_{x\in X_\xi,y\in X_{\xi,2\xi}}V_{xy} a_x^*a_y^*a_ya_x.\label{W-def}
		\end{align} 
		Write $$\phi_r=e^{-irH_{X_{2\xi}}}\phi.$$ Eqs.~\eqref{4.41}--\eqref{5.12} imply, for any $\varphi,\,\psi$ satisfying \eqref{g0-cond},
		\begin{align}\label{4.42}
			\abs{\undel{\Rem_t^{(2)} }}\le   t\sup_{0\le r \le t}\del{\abs{\urdel{\sbr{S,A^\xi_{t-r}}}}+\abs{\urdel{\sbr{W,A^\xi_{t-r}}}}}.
		\end{align}
		To complete the induction step, we estimate the r.h.s.~of \eqref{4.42}.
		
		2. 	 We begin with some abstract lemmas, proved in \secref{secPfLems} below, that will later be used in estimating \eqref{5.12}--\eqref{W-def}.  
		
		Fix $v_1>\kappa$, $\xi\ge2$,  and $r,\,s,\,t\ge0$ with $r+s=t< \xi/v_1.$ We have the following general results:
		\begin{lemma}\label{lemIest}
			Let $n\ge1$ and let $h=(h_{xy})$ be an $1$-body matrix satisfying, for some $K>0$,  
			\begin{align}
				\sup_{x\in\Lam}\sum_{y\in\Lam} \abs{h_{xy}} \abs{x-y}^{\nu+1}\le K\qquad (\nu=0,\ldots,n).\label{hcond}
			\end{align}
			Suppose \eqref{annEst} holds for state $\phi_r$ and some $p\ge1$. Then, {for $\xi\ge\diam X$,}
			\begin{align}
				&\mI_p=\mI_p(r, h,\phi):=\sum_{x\in X_\xi,y\in X_{\xi,2\xi}} \abs{h_{xy}}\phirdel{n_y^p},\notag
			\end{align}
			there holds
			\begin{align}
				\notag
				\mI_p\le C	\tau_p(\phi),
			\end{align}
			where  ${C=C(\al,d, C_H, v,p,X,K)}$,  and
			\begin{align}
				\label{taupDef}
				\tau_p(\phi):={ \phidel{  N_{1,\xi}N_{X_{2\xi}}^{p-1} }+ \xi^{-n+dp} \l^p}.	
			\end{align}
		\end{lemma}

		\begin{lemma}\label{lemIIest}
			Let the assumptions of \lemref{lemIest} hold and let $A\in\cB_X$. Then, for 
			\begin{align}
				\notag
				\mII_p=\mII_p(r,h,\phi,A):=\sum_{x\in X_\xi,y\in X_{\xi,2\xi}} \abs{h_{xy}}\phirdel {A^\xi_{s}  n_x^p (A^\xi_{s})^*},
			\end{align}
			there holds
			\begin{align}
				\notag
				\mII_p\le C	\norm{A}^2\tau_p(\phi).
			\end{align}
		\end{lemma}
		
		\begin{lemma}\label{lemIIIest}
			Let the assumptions of \lemref{lemIIest} hold and put
			\begin{align}
				\mIII_p=\mIII_p(r,h,\phi,A):=& 
				\sum_{x\in X_\xi,y\in X_{\xi,2\xi}} \abs{h_{xy}}\phirdel {n_x^{p/2} A^\xi_{s}  (A^\xi_{s})^*n_x^{p/2}}.\notag
			\end{align}
			Then there holds
			\begin{align}
				\mIII_p\le C	\norm{A}^2\tau_p(\phi).\label{IIIpEst}
			\end{align}
		\end{lemma}

		Our general strategy for proving these lemmas is as follows: First, 
		we split the annulus $X_{0,2\xi}$ into four annuli, 
		\begin{equation}\label{55s}
			X_{0,(1-\g)\xi},\quad X_{(1-\g)\xi,\xi}, \quad X_{\xi,(1+\g)\xi,},\quad X_{(1+\g)\xi,2\xi},
		\end{equation}see \figref{fig:N} below.	We then use	the annular MVB from \thmref{thmTSMVB} in the second and the third annuli and the long-distance decay properties of $h_{xy}$ in the first and the fourth ones. The details are deferred to Section \ref{secPfLems}.
		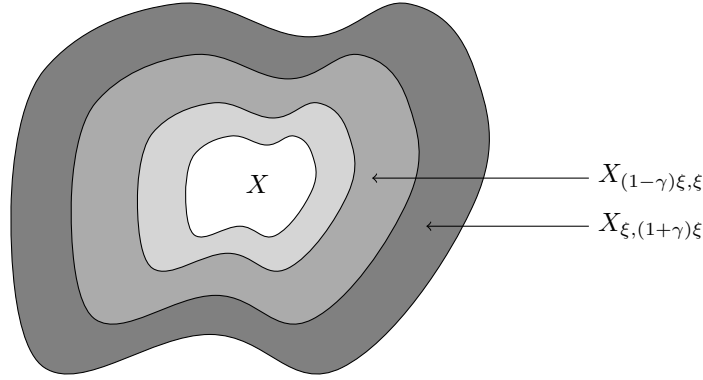
\begin{figure}[H]
			\centering
			\begin{tikzpicture}[scale=.8]
				\draw[fill=gray!100]  plot[shift={(-0.3,-0.45)}, scale=1.1,smooth, tension=.7] coordinates {(-3,0.5) (-2.5,2.5) (-.5,3.5) (1.5,3) (3,3.5) (4,2.5) (4,0.5) (2,-2) (0,-1.5) (-2.5,-2) (-3,0.5)};
				
				\draw[fill=gray!66]  plot[shift={(-0.2,-0.25)}, scale=.8,smooth, tension=.7] coordinates {(-3,0.5) (-2.5,2.5) (-.5,3.5) (1.5,3) (3,3.5) (4,2.5) (4,0.5) (2,-2) (0,-1.5) (-2.5,-2) (-3,0.5)};
				
				\draw[fill=gray!33]  plot[scale=.5,smooth, tension=.7] coordinates {(-3,0.5) (-2.5,2.5) (-.5,3.5) (1.5,3) (3,3.5) (4,2.5) (4,0.5) (2,-2) (0,-1.5) (-2.5,-2) (-3,0.5)};
				
				\draw[fill=white]  plot[shift={(0.2,.15)},scale=.3,smooth, tension=.7] coordinates {(-3,0.5) (-2.5,2.5) (-.5,3.5) (1.5,3) (3,3.5) (4,2.5) (4,0.5) (2,-2) (0,-1.5) (-2.5,-2) (-3,0.5)};
				
				\node at (.5,.4) {$X$};

				\draw [->] (6,.5)--(2.4,.5);
				\node [right] at (6,.5) {$X_{(1-\g)\xi,\xi}$};
				
				\draw [->] (6,-.3)--(3.3,-.3);
				\node [right] at (6,-0.3) {$X_{\xi,(1+\g)\xi}$};
			\end{tikzpicture}
			\caption{Schematic diagram illustrating the decomposition \eqref{55s}.}
			\label{fig:N}
		\end{figure}

		3. Thanks to our power-law decay conditions on the system Hamiltonian,   estimate \eqref{Kcond} is verified, and therefore the assumption \eqref{hcond} is satisfied for $J=(J_{xy})$ and $V=(V_{xy})$. Thus 
		we proceed to apply Lems.~\ref{lemIest}--\ref{lemIIIest} to estimate the r.h.s.~of \eqref{4.42}.
		Let $s:=t-r$ with	\begin{equation}\label{s-cond}
			{0\le t< \xi/v. }
		\end{equation} By formula \eqref{5.10} and the triangle inequality,
		we have the estimate
		\begin{align*}
			&\abs{\urdel{\sbr{S',A^\xi_{s}}}}\\\le& \sum_{x\in X_\xi,y\in X_{\xi,2\xi}} \abs{J_{xy}}\del{\abs{\urdel{A^\xi_{s}a_x^*a_y}}+\abs{\urdel{a_x^*a_yA^\xi_{s}}}}\\
			=& \sum_{x\in X_\xi,y\in X_{\xi,2\xi}} \abs{J_{xy}}\del{\abs{\inn{a_x(A^\xi_{s})^*\varphi_r}{a_y \psi_r}}+\abs{\inn{(A^\xi_{s})^*a_x\varphi_r}{a_y \psi_r}}}.
		\end{align*}
		Using the Cauchy-Schwarz inequality and the fact that $\norm{a_z\varphi_r}=\ordel{n_z}^{1/2}$ (recall $n_z=a_z^*a_z$), we find
		\begin{align}
			&\abs{\urdel{\sbr{S',A^\xi_{s}}}}
			\le \mI_1^{1/2}(r,J,\psi)\mII_1^{1/2}(r,J,\varphi,A)+\mI_1^{1/2}(r,J,\psi)\mIII_1^{1/2}(r,J,\varphi,A).\label{4.71}
		\end{align}
		{	By a straightforward adaption of the procedure above, we also find
			\begin{align}
				&\abs{\urdel{\sbr{W ,A^\xi_{s}}}}
				\le \mI_2^{1/2}\mII_2^{1/2}+\mI_2^{1/2}\mIII_2^{1/2} .\label{WcommEst1}
		\end{align}}

		Applying Lems \ref{lemIest}--\ref{lemIIIest} to estimate the terms in the r.h.s.~of \eqref{4.71} and \eqref{WcommEst1}, we conclude that for all $t$ satisfying \eqref{s-cond},
		\begin{align*}
			\sup_{0\le r \le t}\abs{\urdel{\sbr{S',A^\xi_{t-r}}}}\le C \norm{A} \tau_1(\varphi)^{1/2}\tau_1(\psi)^{1/2}.
		\end{align*}
		Since $\abs{\urdel{\sbr{(S')^*,A^\xi_{t-r}}}}=\abs{\inn{\psi_r}{\sbr{S',(A^\xi_{t-r})^*}\varphi_r}}$ and $(A^\xi_{t-r})^*=(A^*)^\xi_{t-r}$ (see \eqref{Axi}),  interchanging the roles of $A^*$ (resp. $\varphi_r$) and $A$ (resp. $\psi_r$) yields the exact same estimate for $$\sup_{0\le r \le t}\abs{\urdel{\sbr{(S')^*,A^\xi_{t-r}}}}$$ as above. The estimates for $W$ is similar. 
		
		Recalling  definitions \eqref{5.12}--\eqref{W-def}    and the validity interval \eqref{s-cond},  we conclude  that {for every $0\le t< \xi/v$,}
		\begin{align}\label{4.84}
			\sup_{0\le r \le t}\del{\abs{\urdel{\sbr{S,A^\xi_{t-r}}}}+\abs{\urdel{\sbr{W,A^\xi_{t-r}}}}}\le C \norm{A}\sum_{q=1}^2\tau_q(\varphi)^{1/2}\tau_q(\psi)^{1/2}.
		\end{align}
		Plugging \eqref{4.84} back to \eqref{4.42}, we conclude that 
		\begin{align}
			\abs{\undel{\Rem_t}}\le&   Ct  \norm{A}\sum _{q=1}^2\tau_q(\varphi)^{1/2}\tau_q(\psi)^{1/2}.\label{RemFormEst}
		\end{align}
		This completes the induction step and \propref{propClaim} is proved.
		
		\subsection{Completing the proof of \thmref{thm1}}\label{sec75}
		
		Finally, for any operator $B\in\cB_Y$ with $\dist(X,Y)\ge2\xi$, we choose $\varphi=B^*\psi$ in \eqref{RemFormEst} to obtain
		\begin{align}\label{lcae-gen-pf1}
			|\wndel{B\,\Rem_t}|\le&  \abs{\inn{B^*\psi }{\Rem_t\psi}} \le Ct \norm{A}\sum _{q=1}^2 \tau_q(B^*\psi)^{1/2}\tau_q(\psi)^{1/2}.
		\end{align}
		Since $B\in\cB_Y$ and $Y\subset X_{2\xi}^\cp$, we have $[B,N_\Lam]=[B,N_{1,\xi}]=0$ (see \eqref{NgxiDef}). Therefore, by definition \eqref{taupDef}, we have $\tau_q(B^*\psi) \le \norm{B}^2\tau_q(\psi)$ for every $q$.
		This, together with estimate \eqref{lcae-gen-pf1} and definition \eqref{taupDef},  yields
		\begin{align}
			|\wndel{B\,\Rem_t}|\le&  Ct \norm{A}\norm{B} \del{  \wndel{N_{1,\xi}N_\Lam}+ \xi^{-\beta}}.\notag
		\end{align}
		Owning to the particle-free shell condition \eqref{locCond}, the first term on the r.h.s.~vanishes. 
		This completes the proof of Theorem \ref{thm1}. 	\qed
		
		\subsection{Generalization of \thmref{thm1} }
		Since  Lems.~\ref{lemIest}--\ref{lemIIIest} are valid for higher moments, for any integer $Q\ge1$, a straightforward adaption of the argument above yields the following generalization of \thmref{thm1} to Hamiltonians of the form \eqref{HQ}:
		\begin{proposition}\label{propThm2gen}
			Let $Q\geq 1$ and assume that
			\begin{align} 
				{\al>\max\{3dQ/2+1,2d+1\}}.
			\end{align} 
			Let $X\subset \Lam$. Assume that the initial state satisfies the particle-free shell condition \eqref{locCond} and, for some $\l>0$, the uniform density upper bound
			\begin{align} 
				\wzdel{ N_{B_r(x)}^q}\le(\l r^d)^q \qquad  (1\le q\le Q,x\in\Lam, r\ge1).
			\end{align}
			Then, for every $v>2\kappa$,  there exists a positive constant
			\begin{align}
				C=C(\al,d, C_{J}, C_{V^{(1)}},\ldots, C_{V^{(Q)}},v,Q,X),	
			\end{align}
			s.th.~for all {$\xi\ge\max(2,\diam X)$}, $Y\subset\Lam$ with $\dist(X, Y)\ge2\xi$, and  operators $A\in \cB_X,\,B\in \cB_Y$, 
			{			\begin{align}
					|\br{B(\al_t^{H^Q_\Lam}(A)-\al_t^{H^Q_{X_\xi}}(A))}_0|\le C \norm{A}\norm{B} \abs{t}{  {\xi^{-\tilde\beta}\sum_{q=1}^Q\l^q}},
			\end{align}}
			for all $\abs{t}< \xi/v$. The error exponent $\tilde \beta$ is given by
			\begin{align}
				\tilde \beta:=&\left\{ \begin{aligned} &						\lfloor\al-2d-1\rfloor,\quad Q=1, \\ &\lfloor\al-{\tfrac{3dQ}{2}}-1\rfloor,\quad Q\ge2.
				\end{aligned}\right.
			\end{align}
		\end{proposition}	
		
		\section{Bounding Terms $\mathrm{(I)}$-$\mathrm{(III)}$ (Proofs of Lemmas~\ref{lemIest}--\ref{lemIIIest})}\label{secPfLems}
		In this section, we supply the bounds on terms $\mathrm{(I)}$-$\mathrm{(III)}$ that were used in \secref{secPfLCAgen}. 
		
		As a starting point, the following lemma, which is a direct consequence of the particle propagation bound \eqref{MVBp}, gives a crude estimate that controls the scale of growth of particles:
		\begin{lemma}\label{lemScale}
			Suppose estimate \eqref{MVBp} holds for some $p\ge1$.	Then, for any $v>\kappa$, $\mu,\,\mu'>0$, and bounded subset $X\subset\Lam$, there exists   $C =C (\al, d, C_J,p, v,\mu,\mu')>0$ s.th.~for all  {$\xi>\frac\mu2\diam X$,} 
			\begin{align}
				\label{scaleEst}
				\xi^{-dp}\sup_{vt\le \xi}\langle N_{X_{\mu'\xi}}^p\rangle_{t}\le C\l^p.
			\end{align}
			
		\end{lemma}\begin{proof}
			Set $r:=\frac12\diam X$ so that $X_{\mu'\xi}\subset B_{r+\mu'\xi}$ for any $\mu',\xi\ge0$. Using \eqref{MVBp}, the uniform density assumption \eqref{UDBp}, 
			we find
			\begin{align}
				\label{propEst5'}
				\xi^{-dp}\sup_{ vt\le \xi} \br{ N_{X_{\mu'\xi}}^p}_ {t}\le& \xi^{-dp}\sup_{ vt\le \xi} \br{ N_{B_{r+\mu'\xi}}^p}_ {t}\notag
				\\\le& \xi^{-dp}
				\del{\del{1+C\xi^{-1}}\br{N_{B_{r+(1+\mu')\xi}}^p}_0+C\l^p}\notag\\
				\le&C\del{\xi^{-dp}(r+(1+\mu')\xi)^{dp}+1}\l^p.
			\end{align}
			By the assumption $\xi>\mu r$, we have $\xi^{-d}(r+(1+\mu')\xi)^d\le (1+\mu'+\mu^{-1})^d$. This, together with  \eqref{propEst5'}, implies the desired estimate \eqref{scaleEst}.
		\end{proof}

		Let $v>2\kappa$. For the remainder of this section we fix a number % 
		\begin{align}
			\label{gChoice}
			\g=\g(v,\kappa)\equiv\g(v,C_H,\al)\in(0,1/3)	
		\end{align} such that, for $\kappa$  defined in \eqref{kappa},
		\begin{equation}\label{c1def}
			v_1:=\frac{(1-\g)v}{2}>\kappa.
		\end{equation}
		Results below are valid for $r,s\ge0$ with
		\begin{align}
			\label{rstRel}
			r+s=t< \frac{1-\g}{2}\xi/v_1\equiv\xi/v,
		\end{align}
		c.f.~\eqref{s-cond}.
		
		For ease of notation we introduce the local number operators  counting  particles in the curved annular regions corresponding to decomposition \eqref{55s} (see \figref{fig:N}):  \begin{equation}\label{N1def}
			N_{\g,\xi}:=N'_{\g,\xi}+N''_{\g,\xi},\quad	N'_{\g,\xi}  := N_{ X_{(1-\g)\xi,\xi}},\quad 	N''_{\g,\xi}  := N_{ X_{\xi,(1+\g)\xi}}.
		\end{equation} 
		In the remainder of this section,  constant  $C>0$ might increase from one line to the next, but have the same dependence as in \eqref{C1C2}. 
		
		\subsection{Proof of \lemref{lemIest}}Recall $n\ge1$ from assumption \eqref{hcond} and 
		the decomposition from \figref{fig:N},
		\begin{equation}\label{X-decomp}
			X_\xi=X_{(1-\g)\xi}\cup  X_{(1-\g)\xi, \xi}.
		\end{equation}
		Since $\abs{x-y}\ge\g\xi$ for all $x\in X_{(1-\g)\xi}$ and $y\in X_{\xi,2\xi}$, we have
		\begin{align}
			\mI_p
			&=\sum_{y\in X_{\xi,2\xi}} |h_{xy}|\del{\sum_{x\in  X_{(1-\g)\xi,\xi}} \phirdel{n_x^p}  +\sum_{x\in X_{(1-\g)\xi}} \phirdel{n_x^p}}   \notag \\ 
			&\le {\del{\sup_{x\in\Lam}\sum_{y\in X_{\xi,2\xi}} |h_{xy}|}}\sum_{x\in  X_{(1-\g)\xi,\xi}} \phirdel{n_x^p} \notag \\&\quad +(\g\xi)^{-n} \del{\sup_{x\in\Lam}\sum_{y\in X_{\xi,2\xi}} |h_{xy}||x-y|^n}\sum_{x\in X_{(1-\g)\xi}} \phirdel{n_x^p}. \notag 
		\end{align}
		Using condition \eqref{hcond}, together with the relation $\sum_{x\in S} n_x^p\le N_S^p$ (since $n_x\ge0$), we deduce
		\begin{equation}
			\mI_p \le K\del{ \phirdel{(N'_{\g,\xi})^p }+{(\g\xi)^{-n }\phirdel{(N_{X_{(1-\g)\xi}})^p} }}. \label{4.811}
		\end{equation}
		Here $K>0$ is as in \eqref{hcond}.

		To estimate the first term in line \eqref{4.811},  we use the relation $N_{\g,\xi}'\le N_{\g,\xi}$ (see \eqref{N1def}) and apply the annular propagation estimate
		\eqref{annEst} with  the choice (c.f.~\eqref{c1def}){$$(v,\g_1,\g_2)= (v_1,\g,1).$$ This way we obtain,  for all $0\le r< (1-\g)\xi/v_1$,
			\begin{align}\label{8.909}
				\phirdel{(N_{\g,\xi}')^p}\le&C\tau_p(\phi).
				%				\\\tau_1(\phi):=&{C_1\phidel{  N_{1,\xi} }+C_2(\g\xi)^{-n+d} \l}.\label{tau1def}
			\end{align}

			Next, to bound the second term in the r.h.s.~of \eqref{4.811}, {we use \lemref{lemScale}}. }
		Applying estimate \eqref{scaleEst} and recalling definition \eqref{gChoice}, we find that
		\begin{align}
			\notag
			(\g\xi)^{-n }\sup_{r<\xi/v_1}\phirdel{(N_{X_{(1-\g)\xi}})^p} \le C\xi^{-n+dp}\lam^p.
		\end{align}
		Plugging this   together with \eqref{8.909} back to \eqref{4.811} yields
		\begin{equation}\notag
			\mI_p\le C\tau_p(\phi),
		\end{equation}
		uniformly for all $r$ with  {$0\le r< (1-\g)\xi/v_1$}.
		
		This completes the proof of \lemref{lemIest}.\qed

		\subsection{Proof of \lemref{lemIIest}}
		Using  decomposition \eqref{X-decomp}, we compute 
		\begin{align}
			\mII_p =&\sum_{y\in X_{\xi,2\xi}} |h_{xy}|\sum_{x\in X_{(1-\g)\xi}}\phirdel {A^\xi_{s}  n_x^p (A^\xi_{s})^*} \label{4.802} \\
			& + \sum_{y\in X_{\xi,2\xi}} |h_{xy}|\sum_{x\in  X_{(1-\g)\xi,\xi}}\phirdel {A^\xi_{s}  n_x^p (A^\xi_{s})^*}  \label{4.803}.
		\end{align}
		We first bound the term in line \eqref{4.802}.
		Since $A\in\cB_X$ (see \eqref{BX}), we have
		\begin{align}
			\label{ANcomm}
			0=[A,N]=[A,N_{X_\xi}]=[A,N_X].
		\end{align}
		Through \eqref{ANcomm} and \eqref{Axi}, we arrive at the relation $[A_s^\xi,N_{X_\xi}]=0$.
		This, together with estimate \eqref{scaleEst} implies, for all $r<\xi/v_1$ and $K>0$ from \eqref{hcond},   
		\begin{align}
			&\quad\sum_{y\in X_{\xi,2\xi}} |h_{xy}|\sum_{x\in X_{(1-\g)\xi}}\phirdel {A^\xi_{s}  n_x^p(A^\xi_{s})^*} \notag \\ 
			&\le (\g\xi)^{-n}  \sum_{y\in X_{\xi,2\xi}} |h_{xy}||x-y|^n\sum_{x\in X_{(1-\g)\xi}} \phirdel{A^\xi_{s}  n_x^p(A^\xi_{s})^*}\notag \\ 
			&\le K(\g\xi)^{-n } \phirdel{A^\xi_{s}N_{X_{\xi}}^p(A^\xi_{s})^*}\notag\\
			&\le K\norm{A}^2(\g\xi)^{-n}\phirdel{N_{X_{\xi}}^p}\notag\\
			&\le C \norm{A}^2 \xi^{-n+dp}\l^p.\label{4.802'}
		\end{align}
		This bounds the term in \eqref{4.802}.
		
		To bound the term in line \eqref{4.803}, we 
		define $ \phi_{r,s}:=e^{-isH_{X_\xi}} \phi_r$ and recall $N_{\g,\xi} '\equiv N_{ X_{(1-\g)\xi,\xi}}$.
		Then, we have, by definition \eqref{Axi} for the local evolution,
		\begin{align}\label{4.831'}
			&\sum_{y\in X_{\xi,2\xi}} |h_{xy}|\sum_{x\in  X_{(1-\g)\xi,\xi}}\phirdel\notag {A^\xi_{s}  n_x^p (A^\xi_{s})^*} 
			\\\le& K \phirdel {A^\xi_{s} (N_{\g,\xi} ')^p(A^\xi_{s})^*} \notag
			\\=&K\inn{   \phi_{r,s}} {  A\al^{X_\xi}_{-s}( (N_{\g,\xi}')^p)A^* \phi_{r,s}}.
		\end{align}
		To estimate the quantity in  line  \eqref{4.831'}, we apply \lemref{lemNq} with evolution $\al_{-s}^{X_\xi}(\cdot)$ and initial state $A^* \phi_{r,s}$ to obtain,   for all {$s <  \tfrac{1-\g}{2}\xi/v_1$}, that
		\begin{equation}
			\label{4.831}
			\begin{aligned}
				&\inn{ A^* \phi_{r,s}} {  \al^{X_\xi}_{-s}( (N_{\g,\xi}')^p)A^* \phi_{r,s}}\\
				\le& C\del{\inn{ \phi_{r,s}}{ A  N_{(1+\g)/2,\xi}' N_{X_\xi}^{p-1} A^* \phi_{r,s} }+ (\g\xi)^{-n}  \inn{ \phi_{r,s}}{ A   N_{X_\xi}^p A^* \phi_{r,s} }}   . 
			\end{aligned}
		\end{equation}
		For the remainder term above we use that $0<\g<1/3$ and so  $\tfrac{1+\g}{2}-\g= \tfrac{1-\g}{2}>\g$. 
		%where in the last line we use the assumption $\g<1/4$ so that $(\tfrac12-\g)^{-p}<\g^{-p}$ for $p>0$. 
		{		Note that this is the only place $\g<1/3$ is used. }
		
		Since $\supp A\subset X$, $\supp N_{(1+\g)/2,\xi}'\subset X^\cp$, and $[A,N_{X_\xi}]=0$ (see \eqref{ANcomm}), we can pull out the $A$'s from \eqref{4.831} to obtain 
		\begin{align}
			\label{4.832}
			&\inn{ A^* \phi_{r,s}} {  \al^{X_\xi}_{-s}( N_{\g,\xi}')A^* \phi_{r,s}}\notag\\
			\le&C\norm{A}^2 \del{\inn{ \phi_{r,s}}{ N_{(1+\g)/2,\xi}' N_{X_\xi}^{p-1}\phi_{r,s} }+ (\g\xi)^{-n}  \inn{ \phi_{r,s}}{    N_{X_\xi}^p  \phi_{r,s} }   }.
		\end{align}
		Now we apply \lemref{lemNq} twice on the first term in line \eqref{4.832}, first with the evolution $\al^{X_\xi}_s(\cdot)$ and then with $\al_r^{X_{2\xi}}(\cdot)$. For the first application, we have
		\begin{align}
			&\inn{ \phi_{r,s}}{ N_{(1+\g)/2,\xi}' N_{X_\xi}^{p-1}\phi_{r,s} }\notag\\=&	\inn{ \phi_{r}}{ \al_s^{X_\xi}\del{N_{(1+\g)/2,\xi}'} N_{X_\xi}^{p-1}\phi_{r} }\notag\\
			\le&	\inn{ \phi_{r}}{ \al_s^{X_\xi}\del{N_{(1+\g)/2,\xi}} N_{X_\xi}^{p-1}\phi_{r} }\notag\\
			\le& C\del{\inn{ \phi_{r}}{  \del{N_{\frac{1+\g}{2}+\frac{sv_1}{\xi},\xi}} N_{X_\xi}^{p-1}\phi_{r} }+\xi^{-n}\phirdel{N_{X_\xi}^p}}.\label{519}
		\end{align}
		For the second application, we use \eqref{lemNq} on the leading term in line \eqref{519} to obtain that
		\begin{align}
			\label{520}
			&	\inn{ \phi_{r}}{  \del{N_{\frac{1+\g}{2}+\frac{sv_1}{\xi},\xi}} N_{X_\xi}^{p-1}\phi_{r} }\notag\\\le&\inn{ \phi_{r}}{  \del{N_{\frac{1+\g}{2}+\frac{sv_1}{\xi},\xi}} N_{X_{2\xi}}^{p-1}\phi_{r} }\notag\\
			=&\inn{ \phi}{  \al_r^{X_{2\xi}}\del{N_{\frac{1+\g}{2}+\frac{sv_1}{\xi},\xi}} N_{X_{2\xi}}^{p-1}\phi}\notag\\
			\le&C\del{\inn{ \phi}{  \del{N_{\frac{1+\g}{2}+\frac{(s+r)v_1}{\xi},\xi}} N_{X_{2\xi}}^{p-1}\phi}+\xi^{-n}\phidel{N_{X_{2\xi}}^p}}.
		\end{align}
		Combining \eqref{519}--\eqref{520}, we obtain that for {$r+s=t< \tfrac{1-\g}{2}\xi/v_1$},
		\begin{align}\label{4.833}
			&\inn{ \phi_{r,s}}{  (N_{(1+\g)/2,\xi}')^p   \phi_{r,s} } \notag\\\le& C\del{\inn{ \phi}{  \del{N_{1,\xi}} N_{X_{2\xi}}^{p-1}\phi}+\xi^{-n}\phirdel{N_{X_\xi}^p}+\xi^{-n}\phidel{N_{X_{2\xi}}^p}}.
		\end{align}
		The first term in \eqref{4.833} is of the desired form (see \eqref{taupDef}). To bound the second term, we use the \lemref{lemScale}, which gives 
		\begin{align}
			\label{4.834'}
			\xi^{-n}\phirdel{N_{X_\xi}^p}\le C\xi^{-n+dp}\l^p.
		\end{align}
		To bound the third term, we use the uniform density bound \eqref{UDBp} for $\phi$. In conclusion, we arrive at
		\begin{align}
			\notag
			\inn{ \phi_{r,s}}{ N_{(1+\g)/2,\xi}' N_{X_\xi}^{p-1}\phi_{r,s} }\le C\tau_p(\phi).
		\end{align} 
		This bounds the first term in line \eqref{4.832}.

		Next, by the conservation of $N_{X_\xi}$ under $\al_{(\cdot)}^{X_\xi}$, we find 
		\begin{equation}\notag
			\inn{ \phi_{r,s}}{  N_{X_\xi}^p \phi_{r,s}}=\phirdel{N_{X_\xi}^p}
		\end{equation}
		in the second term in the r.h.s. of \eqref{4.832}. This, together with \eqref{4.834'}, yields
		\begin{align}\label{525}
			\inn{ \phi_{r,s}}{  N_{X_\xi}^p \phi_{r,s}}\le C\xi^{-n+dp}\l^p.
		\end{align}
		Combining \eqref{4.802'}--\eqref{525} yields
		\begin{equation}\notag
			\begin{aligned}
				\mII_p\le &  C\norm{A}^2 \tau_p(\phi),
			\end{aligned}
		\end{equation}
		which holds uniformly for all $r,s,t$ satisfying \eqref{rstRel}.
		This bounds the second term in line \eqref{4.832}. 	
		
		This completes the proof of \lemref{lemIIest}.\qed
		
		\subsection{Proof of \lemref{lemIIIest}}
		
		By definition \eqref{Axi}, we have $\norm{A_s^\xi}\equiv \norm{A}$ for all $s,\,\xi$. Therefore, 
		$$
		\mIII_p\le \norm{A_t^\xi}^2 \sum_{x\in X_\xi,y\in X_{\xi,2\xi}} \abs{h_{xy}}\phirdel{n_x^p}=\norm{A}^2 \sum_{x\in X_\xi,y\in X_{\xi,2\xi}} \abs{h_{xy}}\phirdel{n_x^p}.
		$$
		Applying \lemref{lemIIest} with $\mII_p(r,h,\phi,\1)$ in the last term above gives \eqref{IIIpEst}.  This proves \lemref{lemIIIest}.\qed

		\section*{Acknowledgment}	
		The research of M.L.\ is supported by the DFG through the grant TRR 352 – Project-ID 470903074 and by the European Union (ERC Starting Grant MathQuantProp, Grant Agreement 101163620).\footnote{Views and opinions expressed are however those of the authors only and do not necessarily reflect those of the European Union or the European Research Council Executive Agency. Neither the European Union nor the granting authority can be held responsible for them.}
        The research of C.R. is supported by the DFG through the grant TRR 352 – Project-ID 470903074.
		J.Z.~is supported by National Key R \& D Program of China Grant 2022YFA100740, China Postdoctoral Science Foundation Grant 2024T170453, National Natural Science Foundation of China Grant 12401602, and the Shuimu Scholar program of Tsinghua University.  
		
		\section*{Declarations}
		\begin{itemize}
			\item Conflict of interest: The Authors have no conflicts of interest to declare that are relevant to the content of this article.
			\item Data availability: Data sharing is not applicable to this article as no datasets were generated or analysed during the current study.
		\end{itemize}
		
		\appendix
		\section{Basic Properties of the ASTLOs}
		\subsection{Proof of \lemref{ASTLO and N}}\label{proof lemma aslto and N}

		In this proof we write 
		\begin{align}\label{s defin}
			s\equiv s_\s=\frac{2(R-r)}{3v}.
		\end{align}
		
		Step 1. We start by showing \eqref{f and B at t=0}. We observe that 
		\begin{align}
			\quad f\del{\xi}\equiv 0\text{ for } \xi\le \frac{\omega}{2}.\notag
		\end{align}
		Since $f(\xi)\in[0,1]$ for every $\xi\in\R$, this implies
		\begin{align}\label{f le chi}
			f\del{\xi} \le \chi_{\Set{\xi\ge \omega/2}}.
		\end{align}
		Now fix $t=0$, inequality \eqref{f le chi} implies
		\begin{align}\label{a 1}
			f\left(\frac{R-\abs{x}}{s}\right)\le \chi_{\Set{\frac{R-\abs{x}}{s}\ge \frac{\omega}{2}}}=\chi_{\Set{\abs{x}\le R-{s\omega}/{2}}}.
		\end{align}
		Since 
		\begin{align}
			\chi_{\Set{\abs{x}\le R-\frac{s\omega}{2}}}\le \chi_{\Set{\abs{x}\le R}},\notag
		\end{align}
		inequality \eqref{a 1} implies 
		\begin{align}\label{a 2}
			f\left(\frac{R-\abs{x}}{s}\right)\le \chi_{\Set{\abs{x}\le R}}.
		\end{align}
		Recalling the definition of the ASTLOs \eqref{ASTLOdef}, estimate \eqref{a 2} leads to \eqref{f and B at t=0}.
		
		2. Now we show \eqref{f and b at t}. By the definition of $\cE$, we have that
		\begin{align}
			f\del{\xi}\equiv 1\text{ for }\xi \ge \omega.\notag 
		\end{align}
		This implies that 
		\begin{align}\label{f ge chi}
			f\del{\xi} \ge \chi_{\Set{\xi\ge \omega}}.
		\end{align}
		Thanks to \eqref{f ge chi} we obtain
		\begin{align}\label{a 5}
			f\del{\frac{R-v't-\abs{x}}{s}}\ge \chi_{\Set{\frac{R-v't-\abs{x}}{s}\ge \omega}}=\chi_{\Set{\abs{x}\le R-v't-\omega s}}.
		\end{align}
		Since $s\ge t$ and setting 
		\begin{align}
			& v'\coloneqq \frac{\kappa+v}{2},\notag\\
			&\omega\coloneqq v-v'>0,\label{eps defin}	
		\end{align}
		we obtain
		\begin{align}\label{a 3}
			\Set{\abs{x}\le R-v't-\omega s}\supseteq\Set{\abs{x}\le R-vs}.
		\end{align}
		Applying \eqref{s defin} to \eqref{a 3} yields
		\begin{align}\label{a 4}
			\Set{\abs{x}\le R-v't-\omega s}\supseteq \Set{\abs{x}\le \frac{R}{3}+\frac{2r}{3}}\supseteq       \Set{\abs{x}\le r}.
		\end{align}
		Combining \eqref{a 5} and \eqref{a 4} leads to
		\begin{align}
			f\del{\frac{R-v't-\abs{x}}{s}}\ge \chi_{\Set{ \abs{x}\le r}},\notag
		\end{align}
		which leads to \eqref{f and b at t}.
		
		3. As a next step we derive \eqref{g and b 0}. From 
		\eqref{f ge chi} we obtain
		\begin{align}\label{a 6}
			f\left(\frac{(R+2r)/3-\abs{x}}{s}\right)\ge \chi_{\Set{\frac{(R+2r)/3-\abs{x}}{s}\ge \omega}}=\chi_{\Set{\abs{x}\le \frac{R+2r}{3}-\omega s}}.
		\end{align}
		From \eqref{s defin} and \eqref{eps defin} it follows
		\begin{align}\label{a 7}
			\Set{  \abs{x}\le \frac{R+2r}{3}-\omega s }=\Set{\abs{x}\le \frac{4r}{3}-\frac{R}{3}+\frac{2}{3}\frac{v'}{v}\del{R-r}}.
		\end{align}
		Since $v'>v/2$, we obtain
		\begin{align}\label{a 8}
			\Set{\abs{x}\le \frac{4r}{3}-\frac{R}{3}+\frac{2}{3}\frac{v'}{v}\del{R-r}}\supseteq \Set{\abs{x}\le r}.
		\end{align}
		Inequalities \eqref{a 6}, \eqref{a 7}, and \eqref{a 8} together imply
		\begin{align}\notag
			f\left(\frac{(R+2r)/3-\abs{x}}{s}\right)\ge\chi_{\Set{\abs{x}\le r}},
		\end{align}
		which then implies \eqref{g and b 0}.
		
		4. In this last part we show \eqref{g and B t}. Thanks \eqref{f le chi} it holds that
		\begin{align}\label{a 11}
			f\del{\frac{(R+2r)/3+v't-\abs{x}}{s}}\le \chi_{E},
		\end{align}
		where
		\begin{align}
			E \coloneqq\Set{\frac{(R+2r)/3+v't-\abs{x}}{s}>\frac{\omega}{2}} = \Set{\abs{x}<\frac{R}{3}+\frac{2r}{3}+v't-\frac{\omega s}{2}}. \notag
		\end{align}
		Since $s\ge t$, recalling \eqref{s defin} and \eqref{eps defin}, 
		\begin{align}\label{a 9}
			E\subseteq\Set{\abs{x}\le \frac{R}{3}+\frac{2r}{3}+\del{\frac{3v'}{v}-1}\frac{R-r}{3}}.
		\end{align}
		Since $v'\le v$, \eqref{a 9} implies
		\begin{align}\label{a 10}
			E\subseteq \Set{\abs{x}\le R}.
		\end{align}
		Combining \eqref{a 11} and \eqref{a 10}, we obtain
		\begin{align}
			f\del{\frac{(R+2r)/3+v't-\abs{x}}{s}}\le \chi_{\Set{\abs{x}\le R}}.\notag
		\end{align}
		This concludes the proof.

		\qed
		\subsection{Proof of \lemref{lem7.3}}
		\label{secPfLem7.3}
		
		We observe that, by definition \eqref{ASTLOdef}, ${\Nf{f+}{t}{\si}}\le{\Nf{f-}{t}{\si}}$ if and only if 
		\begin{align}\label{619}
			f_+(\abs{x}_{t,\si})\le f_-(\abs{x}_{t,\si}).
		\end{align}
		By the definitions of $f_+,\;f_-$ (see \eqref{f-}, \eqref{f+}), \eqref{619} is equivalent to
		\begin{align}\label{b}
			f\left(\frac{(R+2r)/3+v't-\abs{x}}{s_{\si}}\right)\le f\left(\frac{R-v't-\abs{x}}{s_{\si}}\right).
		\end{align}
		Since $f\in\cE$ it is  non-decreasing, \eqref{b} is equivalent to
		\begin{align}
			\frac{R+2r}{3} +v't\le R-v't.\notag
		\end{align}
		After rearranging   this yields the desired inequality.
		\qed
		\subsection{Proof of \lemref{new taylor}}\label{secPfExp}
		
		By the mean value theorem, for any $x\le y$ there exists $\xi\in[x,y]$ s.th.
		\begin{align}
			\label{mvt}
			f(x)-f(y)=f'(\xi)(x-y).
		\end{align}
		Write $f'(\xi)=f'(x)+(f'(\xi)-f'(x))$ and  $u:=\sqrt{f'}$, equation \eqref{mvt} becomes
		\begin{align}
			\label{fExp1}
			&f(x)-f(y)\notag\\=&f'(x)(x-y)+(f'(\xi)-f'(x))(x-y)\notag\\
			=&u(x)u(y)(x-y) + u(x)\left(u(x)-u(y)\right)(x-y) +(f'(\xi)-f'(x))(x-y).
		\end{align}
		Recalling definition \eqref{classE}, we see that for every function $f\in\cE$, the derivative $f'$ is compactly supported in $(1/2,1)$ and Lipshitz, and similarly for $u\equiv \sqrt{f'}$. It follows that the H\"older semi-norm
		\begin{align}
			\label{holEst}
			\abs{g}_{0,\eps}\le \abs{g}_{0,1},\qquad g\in\Set{f',u},\,0<\eps\le 1 . 
		\end{align} 
		Inserting \eqref{holEst} into \eqref{fExp1} yields
		\begin{align}\notag
			\abs{	f(x)-f(y)}\le u(x)u(y)\abs{x-y}+C_f\abs{x-y}^{1+\eps}\qquad0<\eps\le 1 ,
		\end{align}
		with
		\begin{align}
			C_f:=\norm{u}_{L^\infty}\abs{u}_{0,\eps}+\abs{f'}_{0,\eps}.\label{Cf def}
		\end{align}
		The desired estimate \eqref{fExp} follows from here.
		\qed

		\bibliography{LRZ}
	\end{document}